\documentclass[12pt,a4paper]{article}
\usepackage[utf8]{inputenc}
\usepackage[T1]{fontenc}

\usepackage[top=1.25in, bottom=1.25in, left=1.25in, right=1.25in]{geometry}
\usepackage{setspace} 	
\usepackage[shortlabels]{enumitem}
\usepackage{hyperref} % pagebackref=false, linktoc=all
\hypersetup{bookmarks=true, unicode=true, urlcolor=blue, linkcolor=blue, filecolor=blue} %

\usepackage{amsmath,amssymb,amsfonts,bbm,amsthm,mathrsfs,mathtools}
\usepackage{stmaryrd}	% brackets
\usepackage[subnum]{cases}	% multi-case equations with separate number for each

\usepackage{xcolor,subfigure,epsfig,graphicx}
\graphicspath{{./Images/},{./Images/comNet/},{./Images/bisNet/},{./Images/simulations/},{./Images/additional_plots/}}
\usepackage[font=footnotesize, width=0.90\linewidth]{caption}
\definecolor{darkblue}{RGB}{0,29,119}
\definecolor{darkgreen}{rgb}{0.0,0.5,0.0}
\usepackage{multirow,multicol,float,array,rotating,booktabs}	
\usepackage[section]{placeins}	% floats not go to next section

\usepackage{dsfont}
\def\R {\mathds{R}}
\def\C {\mathds{C}}
\def\T {\mathds{T}}
\def\E {\mathds{E}}

\allowdisplaybreaks			% breaks "align" in multiple pages
\makeatletter
\newcommand*{\distas}[1]{\mathbin{\overset{#1}{\kern\z@\sim}}}	
\makeatother
\newcommand*{\norm}[1]{\big\lVert#1\big\rVert}		% norm
\newcommand*\abs[1]{\big|#1\big|}		% absolute value
\DeclareMathOperator*{\trr}{tr}			% trace operator
\newcommand*\tr[1]{\trr\big(#1\big)}
\DeclareMathOperator*{\vect}{vec}		% vectorization operator
\newcommand*\vecc[1]{\vect\big(#1\big)}
\newcommand{\bigzero}{\mbox{\normalfont\Large O}}

\newtheoremstyle{custom}%    <name>
                {\topsep}%   <space above>
                {\topsep}%   <space below>
                {\itshape}%  <body font>
                {}%          <indent amount>
                {\bfseries}% <Theorem head font>
                {.}%         <punctuation after theorem head>
                {\newline}%  <space after theorem head> (default .5em)
                {}%          <Theorem head spec>
\theoremstyle{custom}
\newtheorem{proposition}{Proposition}[section]
\newtheorem{lemma}{Lemma}[section]
\newtheorem{definition}{Definition}[section]
\newtheorem{example}{Example}[section]
\newtheorem{remark}{Remark}[section]

\usepackage{authblk}		% affiliation of authors

\title{Bayesian Dynamic Tensor Regression\protect\thanks{
We are grateful to Federico Bassetti, Sylvia Fr{\"u}hwirth-Schnatter, Christian Gouriéroux, S{\o}ren Johansen, Siem Jan Koopman, Gary Koop, André Lucas, Alain Monfort, Peter Phillips, Christian Robert, Mike West, for their comments and suggestions. Also, we thank the seminar participants at: CREST, University of Southampton, Vrije University of Amsterdam, London School of Economics, Maastricht University, Polytechnic University of Milan.
Moreover, we thank the conference and workshop participants at: ``ICEEE 2019'' in Lecce, 2019, ``CFENetwork 2018'' in Pisa, 2018, ``29th EC2 conference'' in Rome, 2018, ``12th RCEA Annual meeting'' in Rimini, 2018, ``8th MAF'' in Madrid, 2018, ``CFENetwork 2017'' in London, 2017, ``ICEEE 2017'' in Messina, 2017, ``3rd  Vienna Workshop on High-dimensional Time Series in Macroeconomics and Finance'' in Wien, 2017, ``BISP10'' in Milan, 2017, ``ESOBE'' in Venice, 2016, ``CFENetwork'' in Seville, 2016, and the ``SIS Intermediate Meeting'' of the Italian Statistical Society in Florence, 2016.
This research used the SCSCF multiprocessor cluster system and is part of the project Venice Center for Risk Analytics (VERA) at Ca' Foscari University of Venice.
}
}

\author[1]{Monica Billio\thanks{e-mail: \href{mailto: billio@unive.it}{billio@unive.it}}}
\author[1]{Roberto Casarin\thanks{e-mail: \href{mailto: r.casarin@unive.it}{r.casarin@unive.it}}}
\author[1,2]{Matteo Iacopini\thanks{e-mail: \href{mailto: matteo.iacopini@unive.it}{matteo.iacopini@unive.it}}}
\author[3]{Sylvia Kaufmann\thanks{e-mail: \href{mailto: sylvia.kaufmann@szgerzensee.ch}{sylvia.kaufmann@szgerzensee.ch}}}
\affil[1]{Ca' Foscari University of Venice}
\affil[2]{Scuola Normale Superiore of Pisa}
\affil[3]{Study Center Gerzensee, Foundation of the Swiss National Bank}
\date{}

\begin{document}

\maketitle

\begin{abstract}
Tensor-valued data are becoming increasingly available in economics and this calls for suitable econometric tools.
We propose a new dynamic linear model for tensor-valued response variables and covariates that encompasses some well-known econometric models as special cases. Our contribution is manifold. First, we define a tensor autoregressive process (ART), study its properties and derive the associated impulse response function. Second, we exploit the PARAFAC low-rank decomposition for providing a parsimonious parametrization and to incorporate sparsity effects. We also contribute to inference methods for tensors by developing a Bayesian framework which allows for including extra-sample information and for introducing shrinking effects. We apply the ART model to time-varying multilayer networks of international trade and capital stock and study the propagation of shocks across countries, over time and between layers.
\end{abstract}
%\textbf{AMS 2000 subject classifications:} Primary 62; secondary 91B84.
%\textbf{JEL Classification:} C13, C33, C51, C53


\section{Introduction}
The increasing availability of long time series of complex-structured data, such as multidimensional tables (\cite{Balazsi15MultidimensionalPanel}, \cite{CarvWest07DynMatNormGraph}), multidimensional panel data (\cite{Bayer16ECTA_Dynamic_Demand_House}, \cite{Baltagi15Hedonic_House_price}, \cite{Davis02Multi-way_error_Panel}, \cite{Shin19Multi-dim_Heterog_Panel}), multilayer networks (\cite{Aldasoro16MultiplexNetwork}, \cite{Poledna15MultiplexNetworkBanks_SystemicRisk}), EEG (\cite{LiZhang17}), neuroimaging (\cite{Zhouetal13}) has put forward some limitations of the existing multivariate econometric models. Tensors, i.e. multidimensional arrays, are the natural class where this kind of complex data belongs.

%panel from Lazlo et al (2017):\\ housing and market retail models (\cite{Bayer16ECTA_Dynamic_Demand_House}, \cite{Baltagi15Hedonic_House_price}, \cite{Davis02Multi-way_error_Panel}), trade (\cite{Shin19Multi-dim_Heterog_Panel})

A na\"ive approach to model tensors relies on reshaping them into lower-dimensional objects (e.g., vectors and matrices) which can then be easily handled using standard multivariate statistical tools. However, mathematical representations of tensor-valued data in terms of vectors have non-negligible drawbacks, such as the difficulty of accounting for the intrinsic structure of the data (e.g., cells of a matrix representing a geographical map or pairwise relations, contiguous pixels in an image). Neglecting this information in the modelling might lead to inefficient estimation and misleading results.
Tensor-valued data entries are highly likely to depend on contiguous cells (within and between modes) and collapsing the data into a vector destroys this information. Thus, statistical approaches based on vectorization are unsuited for modelling tensor-valued data.

Tensors have been recently introduced in statistics and machine learning (e.g., \cite{Hackbusch12Tensor_book}, \cite{Kroonenberg08AppliedMultiwayDataAnalysis}) and provide a fundamental background for efficient algorithms in \textit{Big Data} handling (e.g., \cite{Cichocki14BigData_Tensor}).
However, a compelling statistical approach extending results for scalar random variables to multidimensional random objects beyond dimension 2 (i.e., matrix-valued random variables, see \cite{GuptaNagar99MatrixDistributions}) is lacking and constitutes a promising field of research.

The development of novel statistical methods able to deal directly with tensor-valued data (i.e., without relying on vectorization) is currently an open field of research in statistics and econometrics, where such kind of data is becoming increasingly available.
The main purpose of this article is to contribute to this growing literature by proposing an extension of standard multivariate econometric regression models to tensor-valued response and covariates.

%--------------------------------------------------------------------
Matrix-valued statistical models have been widely employed in time series econometrics over the past decades, especially for state space representations (\cite{HarrisonWest99BayesForecastDLM}), dynamic linear models (\cite{CarvWest07DynMatNormGraph}, \cite{Wang09Bayes_matrixNormalGraph}), Gaussian graphical models (\cite{Carvetal07HIW_Graph}), stochastic volatility (\cite{Uhlig97Bayes_VAR_SV}, \cite{Gourieroux09Wishart_AR}, \cite{Golosnoy12conditional_Wishart_AR}), classification of longitudinal datasets (\cite{Viroli11MatNorm}), models for network data (\cite{DuranteDunson14BNP_DynNet}, \cite{Zhu17Network_VAR}, \cite{Zhu19Network_Quantile_regression}) and factor models (\cite{Chen19Matrix_DynamicFactor}).

%{\color{red}
%\cite{Viroli12Mat_Reg} introduced a matrix-valued regression where both response and covariate are matrices.
%\cite{DingCook18MatrixReg} generalized the envelope method of \cite{Cook10Envelope} for achieving sparsity and increasing efficiency of the regression.
%\cite{HunkWang13binYmatX}, who defined a logistic regression model with a matrix-valued covariate.
\cite{DingCook18MatrixReg} proposed a bilinear multiplicative matrix regression model, which in vector form becomes a VAR($1$) with restrictions on the covariance matrix. The main shortcoming in using bilinear models is the difficulty in introducing sparsity. Imposing zero restrictions on a subset of the reduced form coefficients implies zero restrictions on the structural coefficients.

%--------------------------------------------------------------------
Recent papers dealing with tensor-valued data include \cite{Zhouetal13} and \cite{XuZhangetal13}, who proposed a generalized linear model to predict a scalar real or binary outcome by exploiting the tensor-valued covariate. Instead, \cite{Zhao13Tensor_BNP}, \cite{Zhao14Tensor_BNP} and \cite{ImaizumiHayashi16Tensor_BNP} followed a Bayesian nonparametric approach for regressing a scalar on tensor-valued covariate.
Another stream of the literature considers regression models with tensor-valued response and covariates. In this framework, \cite{LiZhang17} proposed a model for cross-sectional data where response and covariates are tensors, and performed sparse estimation by means of the envelope method and iterative maximum likelihood. \cite{Hoff15} exploited a multidimensional analogue of the matrix SVD (the Tucker decomposition) to define a parsimonious tensor-on-tensor regression.
%}

We propose a new dynamic linear regression model for tensor-valued response and covariates. We show that our framework admits as special cases Bayesian VAR models (\cite{Sims98BayesDynamicMultivariate}), Bayesian panel VAR models (\cite{CanovaCiccarelli04PanelVAR}) and Multivariate Autoregressive Index models (i.e. MAI, see \cite{Carrieroetal16Multivariate_AR_Index}), as well as univariate and matrix regression models. Furthermore, we exploit a suitable tensor decomposition for providing a parsimonious parametrization, thus making inference feasible in high-dimensional models.
One of the areas where these models can find application is network econometrics.

Most statistical models for network data are static (\cite{dePaula17Econometrics_Networks}), whereas dynamic models maybe more 	adequate for many applications (e.g., banking) where data on network evolution are becoming available.
Few attempts have been made to model time-varying networks (e.g., \cite{Holme12TemporalNetworks}, \cite{Kostakos09Temporal_Graphs}, \cite{Anacleto17DynamicChainGraph_NetworkTimeSeries}), and most of the contributions have focused on providing a representation and a description of temporally evolving graphs.
We provide an original study of time-varying economic and financial networks and show that our model can be successfully used to carry out impulse response analysis in this multidimensional setting.

The remainder of this paper is organized as follows. Section \ref{sec:tensor_calculus_model} provides an introduction to tensor algebra and presents the new modelling framework. Section  \ref{sec:bayesian_inference} discusses parametrization strategies and a Bayesian inference procedure. Section \ref{sec:applications} provides an empirical application and section  \ref{sec:conclusions} gives some concluding remarks. Further details and results are provided in the supplementary material.

\section{A Dynamic Tensor Model} \label{sec:tensor_calculus_model}
In this section, we present a dynamic tensor regression model and discuss some of its properties and special cases.
We review some notions of multilinear algebra which will be used in this paper, and refer the reader to Appendix \ref{sec:apdx_tensor_calculus} and the supplement for further details.

\subsection{Tensor Calculus and Decompositions}
The use of tensors is well established in physics and mechanics (e.g., see \cite{Aris12Tensor_Mechanics} and \cite{Abraham12Tensor_Physics}), but few contributions have been made beyond these disciplines. For a general introduction to the algebraic properties of tensor spaces, see \cite{Hackbusch12Tensor_book}.
Noteworthy introductions to operations on tensors and tensor decompositions are \cite{LeeChi16} and \cite{KoldaBader09}, respectively.

A $N$-order real-valued tensor is a $N$-dimensional array $\mathcal{X} = (\mathcal{X}_{i_1,\ldots,i_N}) \in \R^{I_1\times\ldots\times I_N}$ with entries $\mathcal{X}_{i_1,\ldots,i_N}$ with $i_n =1,\ldots,I_n$ and $n=1,\ldots,N$. The \textit{order} is the number of dimensions (also called modes). Vectors and matrices are examples of 1- and 2-order tensors, respectively.
In the rest of the paper we will use lower-case letters for scalars, lower-case bold letters for vectors, capital letters for matrices and calligraphic capital letters for tensors.
We use the symbol ``$:$'' to indicate selection of all elements of a given mode of a tensor.
The mode-$k$ \textit{fiber} is the vector obtained by fixing all but the $k$-th index of the tensor, i.e. the equivalent of rows and columns in a matrix. Tensor \textit{slices} and their generalizations, are obtained by keeping fixed all but two or more dimensions of the tensor.

It can be shown that the set of $N$-order tensors $\R^{I_1\times\ldots\times I_N}$ endowed with the standard addition $\mathcal{A} + \mathcal{B} = (\mathcal{A}_{i_1,\ldots,i_N} + \mathcal{B}_{i_1,\ldots,i_N})$ and scalar multiplication $\alpha \mathcal{A} = (\alpha \mathcal{A}_{i_1,\ldots,i_N})$, with $\alpha \in \R$, is a vector space.
We now introduce some operators on the set of real tensors, starting with the \textit{contracted product}, which generalizes the matrix product to tensors. The contracted product between $\mathcal{X}\in\R^{I_1 \times\ldots\times I_M}$ and $\mathcal{Y}\in\R^{J_1 \times\ldots\times J_N}$ with $I_M = J_1$, is denoted by $\mathcal{X} \times_M \mathcal{Y}$ and yields a $(M+N-2)$-order tensor $\mathcal{Z}\in\R^{I_1 \times\ldots\times I_{M-1} \times J_1 \times\ldots\times J_{N-1}}$, with entries
\begin{equation*}
\mathcal{Z}_{i_1,\ldots,i_{M-1},j_2,\ldots,j_{N}} =
(\mathcal{X}\times_M \mathcal{Y})_{i_1,\ldots,i_{M-1},j_2,\ldots,j_{N}} = 
\sum_{i_M=1}^{I_M} \mathcal{X}_{i_1,\ldots,i_{M-1},i_M} \mathcal{Y}_{i_M,j_2,\ldots,j_N}.
\label{eq_apdx_tensor_moden_tensor}
\end{equation*}
When $\mathcal{Y} = \mathbf{y}$ is a vector, the contracted product is also called \textit{mode-$M$ product}.
We define with $\mathcal{X} \bar{\times}_{N} \mathcal{Y}$ a sequence of contracted products between the $(K+N)$-order tensor $\mathcal{X} \in \R^{J_1\times \ldots \times J_K \times I_1 \times \ldots \times I_N}$ and the $(N+M)$-order tensor $\mathcal{Y} \in \R^{I_1 \times \ldots \times I_N \times H_1\times \ldots \times H_M}$. Entry-wise, it is defined as
\begin{equation*}
\big( \mathcal{X} \bar{\times}_{N} \mathcal{Y} \big)_{j_1,\ldots,j_K,h_1,\ldots,h_M} = \sum_{i_1=1}^{I_1} \ldots \sum_{i_N=1}^{I_N} \mathcal{X}_{j_1,\ldots,j_K,i_1,\ldots,i_N} \mathcal{Y}_{i_1,\ldots,i_N,h_1,\ldots,h_M}.
\end{equation*}
%The \textit{mode-$n$ product} between a $N$-order tensor $\mathcal{X}$ and a vector $\mathbf{v}\in\R^{I_n}$, is a $(N-1)$-order tensor whose entries are defined as
%\begin{equation}
%\mathcal{Y}_{i_1,\ldots,i_{n-1},i_{n+1},\ldots,i_N} = 
%(\mathcal{X} \times_n \mathbf{v})_{i_1,\ldots,i_{n-1},i_{n+1},\ldots,i_N} = 
%\sum_{i_n=1}^{I_n} \mathcal{X}_{i_1,\ldots,i_n,\ldots,i_N} \mathbf{v}_{i_n}
%\label{eq:apdx_tensor_moden_vector}
%\end{equation}
%with $\mathcal{Y}\in\R^{I_1\times\ldots,I_{n-i},I_{n+1},\ldots\times I_N}$.
Note that the contracted product is not commutative.
The \textit{outer product} $\circ$ between a $M$-order tensor $\mathcal{X}\in\R^{I_1 \times\ldots\times I_M}$ and a $N$-order tensor $\mathcal{Y}\in\R^{J_1 \times\ldots\times J_N}$ is a $(M+N)$-order tensor $\mathcal{Z}\in\R^{I_1 \times\ldots\times I_M \times J_1 \times\ldots\times J_N}$ with entries $\mathcal{Z}_{i_1,\ldots,i_M,j_1,\ldots,j_N} = 
(\mathcal{X} \circ \mathcal{Y})_{i_1,\ldots,i_M,j_1,\ldots,j_N} = 
\mathcal{X}_{i_1,\ldots,i_M} \mathcal{Y}_{j_1,\ldots,j_N}$.

Tensor decompositions allow to represent a tensor as a function of lower dimensional variables, such as matrices of vectors, linked by suitable multidimensional operations.
In this paper, we use the low-rank parallel factor (PARAFAC) decomposition, which allows to represent a $N$-order tensor in terms of a collection of vectors (called marginals). A $N$-order tensor is of rank 1 when it is the outer product of $N$ vectors.
Let $R$ be the rank of the tensor $\mathcal{X}$, that is minimum number of rank-1 tensors whose linear combination yields $\mathcal{X}$. 
The PARAFAC($R$) decomposition is rank-$R$ decomposition which represents a $N$-order tensor $\mathcal{B}$ as a finite sum of $R$ rank-$1$ tensors $\mathcal{B}_r$ defined by the outer products of $N$ vectors (called marginals) $\boldsymbol{\beta}_j^{(r)} \in\R^{I_j}$
\begin{equation}
\mathcal{B} = \sum_{r=1}^R \mathcal{B}_r = \sum_{r=1}^R \boldsymbol{\beta}_1^{(r)} \circ \ldots \circ \boldsymbol{\beta}_N^{(r)}, \qquad \mathcal{B}_r = \boldsymbol{\beta}_1^{(r)} \circ \ldots \circ \boldsymbol{\beta}_N^{(r)}.
\label{eq:PARAFAC_demposition}
\end{equation}
The \textit{mode-$n$ matricization} (or unfolding), denoted by $\mathbf{X}_{(n)} = \operatorname{mat}_{n}(\mathcal{X})$, is the operation of transforming a $N$-dimensional array $\mathcal{X}$ into a matrix. It consists in re-arranging the mode-$n$ fibers of the tensor to be the columns of the matrix $\mathbf{X}_{(n)}$, which has size $I_n\times I_{(-n)}^*$ with $I_{(-n)}^*=\prod_{i\neq n} I_i$. The mode-$n$ matricization of $\mathcal{X}$ maps the $(i_1,\ldots,i_N)$ element of $\mathcal{X}$ to the $(i_n,j)$ element of $\mathbf{X}_{(n)}$, where $j = 1+ \sum_{m\neq n} (i_m-1) \prod_{p\neq n}^{m-1} I_p$.
For some numerical examples, see \cite{KoldaBader09} and Appendix \ref{sec:apdx_tensor_calculus}. The mode-$1$ unfolding is of interest for providing a visual representation of a tensor: for example, when $\mathcal{X}$ be a 3-order tensor, its mode-$1$ matricization $\mathbf{X}_{(1)}$ is a $I_1 \times I_2 I_3$ matrix obtained by horizontally stacking the mode-$(1,2)$ slices of the tensor.
The \textit{vectorization} operator stacks all the elements in direct lexicographic order, forming a vector of length $I^*=\prod_{i} I_i$. Other orderings are possible, as long as it is consistent across the calculations.
The mode-$n$ matricization can also be used to vectorize a tensor $\mathcal{X}$, by exploiting the relationship $\vecc{\mathcal{X}} = \vecc{\mathbf{X}_{(1)}}$, where $\vecc{\mathbf{X}_{(1)}}$ stacks vertically into a vector the columns of the matrix $\mathbf{X}_{(1)}$.
Many product operations have been defined for tensors (e.g., see \cite{LeeChi16}), but here we constrain ourselves to the operators used in this work.
For the ease of notation, we will use the multiple-index summation for indicating the sum over all the corresponding indices.

\begin{remark}
Consider a $N$-order tensor $\mathcal{B} \in \R^{I_1\times\ldots\times I_N}$ with a PARAFAC(R) decomposition (with marginals $\boldsymbol{\beta}_j^{(r)}$), a $(N-1)$-order tensor $\mathcal{Y}\in\R^{I_1\times\ldots\times I_{N-1}}$ and a vector $\mathbf{x}\in\R^{I_N}$. Then
\begin{equation*}
\mathcal{Y} = \mathcal{B} \times_N \mathbf{x} \iff \vecc{\mathcal{Y}} = \mathbf{B}_{(N)}' \mathbf{x} \iff \vecc{\mathcal{Y}}' = \mathbf{x}' \mathbf{B}_{(N)}
\end{equation*}
where $\mathbf{B}_{(N)} = \sum_{r=1}^R \boldsymbol{\beta}_N^{(r)} \vecc{\boldsymbol{\beta}_1^{(r)} \circ \ldots \circ \boldsymbol{\beta}_{N-1}^{(r)}}'$.
\end{remark}

\subsection{A General Dynamic Tensor Model} \label{sec:general_model}
Let $\mathcal{Y}_t$ be a $(I_1\times\ldots\times I_N)$-dimensional tensor of endogenous variables, $\mathcal{X}_t$ a $(J_1\times\ldots\times J_M)$-dimensional tensor of covariates, and $S_y = \bigtimes_{j=1}^N \lbrace 1,\ldots, I_j \rbrace \subset \mathbb{N}^N$ and $S_x = \bigtimes_{j=1}^M \lbrace 1,\ldots, J_j \rbrace \subset \mathbb{N}^M$ sets of $n$-tuples of integers. We define the autoregressive tensor model of order $p$, ART($p$), as the system of equations
\begin{equation}
\mathcal{Y}_{\mathbf{i},t} = \mathcal{A}_{\mathbf{i},0} + \sum_{j=1}^p \sum_{\mathbf{k} \in S_y} \mathcal{A}_{\mathbf{i},\mathbf{k},j} \mathcal{Y}_{\mathbf{k},t-j} + \sum_{\mathbf{m} \in S_x} \mathcal{B}_{\mathbf{i},\mathbf{m}} \mathcal{X}_{\mathbf{m},t} + \mathcal{E}_{\mathbf{i},t}, \quad \mathcal{E}_{\mathbf{i},t} \distas{iid} \mathcal{N}(0,\sigma_\mathbf{i}^2),
\label{eq:model_general_entry}
\end{equation}
$t=1,2,\ldots$, with given initial conditions $\mathcal{Y}_{-p+1},\ldots,\mathcal{Y}_0 \in \R^{I_1\times\ldots\times I_N}$, where $\mathbf{i} = (i_1,\ldots,i_N) \in S_y$ and $\mathcal{Y}_{\mathbf{i},t}$ is the $\mathbf{i}$-th entry of $\mathcal{Y}_t$.
The general model in eq. \eqref{eq:model_general_entry} allows for measuring the effect of all the cells of $\mathcal{X}_t$ and of the lagged values of $\mathcal{Y}_t$ on each endogenous variable.

We give two equivalent compact representations of the multilinear system \eqref{eq:model_general_entry}. The first one is used for studying the stability property of the process and is obtained through the contracted product that provides a natural setting for multilinear forms, decompositions and inversions. From \eqref{eq:model_general_entry} one gets the tensor equation
\begin{equation}
\mathcal{Y}_{t} = \mathcal{A}_{0} + \sum_{j=1}^p \widetilde{\mathcal{A}}_{j} \bar{\times}_{N} \mathcal{Y}_{t-j} + \widetilde{\mathcal{B}} \bar{\times}_{M} \mathcal{X}_{t} + \mathcal{E}_{t}, \quad \mathcal{E}_{t} \distas{iid} \mathcal{N}_{I_1,\ldots,I_N}(\mathcal{O},\Sigma_1,\ldots,\Sigma_N),
\label{eq:model_general_contracted_product}
\end{equation}
where $\bar{\times}_{a,b}$ is a shorthand notation for the contracted product $\times_{a+1\ldots a+b}^{1\ldots a}$, $\widetilde{\mathcal{A}}_0$ is a $N$-order tensor of the same size as $\mathcal{Y}_t$, $\widetilde{\mathcal{A}}_j$, $j=1,\ldots,p$, are $2N$-order tensors of size $(I_1\times \ldots \times I_N \times I_1\times \ldots \times I_N)$ and $\mathcal{B}$ is a $(N+M)$-order tensor of size $(I_1\times \ldots \times I_N \times J_1\times \ldots \times J_M)$.
The error term $\mathcal{E}_t$ follows a $N$-order tensor normal distribution (\cite{Ohlson13TensorNormal}) with probability density function
\begin{equation}
f_\mathcal{E}(\mathcal{E}) = \frac{\exp\Big( -\frac{1}{2} (\mathcal{E}-\mathcal{M}) \bar{\times}_{N,0} \big( \circ_{j=1}^N \Sigma_j^{-1} \big) \bar{\times}_{N,0} (\mathcal{E}-\mathcal{M}) \Big)}{(2\pi)^{I^*/2} \prod_{j=1}^N \abs{\Sigma_j}^{I_{-j}^*/2}},
\label{eq:tensor_normal_dsitrbution}
\end{equation}
where $I^* = \prod_i I_i$ and $I_{-i}^* = \prod_{j\neq i} I_j$, $\mathcal{E}$ and $\mathcal{M}$ are $N$-order tensors of size $I_1\times\ldots\times I_N$. Each covariance matrix $\Sigma_j\in\R^{I_j\times I_j}$, $j=1,\ldots,N$, accounts for the dependence along the corresponding mode of $\mathcal{E}$.

The second representation of the ART($p$) in eq. \eqref{eq:model_general_entry} is used for developing inference.
Let $\mathcal{K}_m$ be the $(I_1\times \ldots \times I_N \times m)$-dimensional commutation tensor such that $\mathcal{K}_m^\sigma \bar{\times}_{N,0} \mathcal{K}_m= \mathbf{I}_m$, where $\mathcal{K}_m^\sigma$ is the tensor obtained by flipping the modes of $\mathcal{K}_m$.
Define the $(I_1\times\ldots\times I_N \times I^*)$-dimensional tensor $\mathcal{A}_j = \widetilde{\mathcal{A}}_j \bar{\times}_{N} \mathcal{K}_{I^*}$ and the $(I_1\times\ldots\times I_N \times J^*)$-dimensional tensor $\mathcal{B} = \widetilde{\mathcal{B}} \bar{\times}_{N} \mathcal{K}_{J^*}$, with $J^* = \prod_j J_j$. We obtain $\mathcal{A}_j \times_{N+1} \vecc{\mathcal{Y}_{t-j}} = \widetilde{\mathcal{A}}_j \bar{\times}_{N} \mathcal{Y}_{t-j}$ and the compact representation
\begin{equation}
\begin{split}
\mathcal{Y}_t & = \mathcal{A}_0 + \sum_{j=1}^p \mathcal{A}_j \times_{N+1} \vecc{\mathcal{Y}_{t-j}} + \mathcal{B} \times_{N+1} \vecc{\mathcal{X}_t} + \mathcal{E}_t, \\
\mathcal{E}_t & \distas{iid} \mathcal{N}_{I_1,\ldots,I_N}(\mathcal{O},\Sigma_1,\ldots,\Sigma_N).
\end{split}
\label{eq:model_general}
\end{equation}

Let $\T = (\R^{I_1\times\ldots\times I_N\times I_1\times\ldots\times I_N}, \bar{\times}_{N})$ be the space of $(I_1\times\ldots\times I_N\times I_1\times\ldots\times I_N)$-dimensional tensors endowed with the contracted product $\bar{\times}_{N}$.
We define the identity tensor $\mathcal{I} \in \T$ to be the neutral element of $\bar{\times}_{N}$, that is the tensor whose entries are $\mathcal{I}_{i_1,\ldots,i_N,i_{N+1},\ldots,i_{2N}} = 1$ if $i_k = i_{k+N}$ for all $k=1,\ldots,N$ and $0$ otherwise.
The inverse of a tensor $\mathcal{A} \in \T$ is the tensor $\mathcal{A}^{-1} \in \T$ satisfying $\mathcal{A}^{-1} \bar{\times}_{N} \mathcal{A} = \mathcal{A} \bar{\times}_{N} \mathcal{A}^{-1} = \mathcal{I}$.
A complex number $\lambda \in \C$ and a nonzero tensor $\mathcal{X} \in \R^{I_1\times \ldots\times I_N}$ are called eigenvalue and eigenvector of the tensor $\mathcal{A} \in \T$ if they satisfy the multilinear equation $\mathcal{A} \bar{\times}_{N} \mathcal{X} = \lambda \mathcal{X}$. We define the spectral radius $\rho(\mathcal{A})$ of $\mathcal{A}$ to be the largest modulus of the eigenvalues of $\mathcal{A}$.
We define a stochastic process to be weakly stationary if the first and second moment of its finite dimensional distributions are finite and constant in $t$.
Finally, note that it is always possible to rewrite an ART($p$) process as a ART(1) process on an augmented state space, by stacking the endogenous tensors along the first mode. Thus, without loss of generality, we focus on the case $p=1$.
We use the definition of inverse tensor, spectral radius and the convergence of power series of tensors to prove the following result.

\begin{lemma} \label{lemma:ARTp_ART1}
Every $(I_1 \times I_2 \times \ldots \times I_N \times I_1 \times I_2 \times \ldots \times I_N)$-dimensional ART($p$) process $\mathcal{Y}_t = \sum_{k=1}^p \mathcal{A}_{k} \bar{\times}_N \mathcal{Y}_{t-j} + \mathcal{E}_t$ can be rewritten as a $(pI_1 \times I_2 \times \ldots \times I_N \times pI_1 \times I_2 \times \ldots \times I_N)$-dimensional ART(1) process $\underline{\mathcal{Y}}_t = \underline{\mathcal{A}} \bar{\times}_N \underline{\mathcal{Y}}_{t-1} + \underline{\mathcal{E}}_t$.
\end{lemma}

\begin{proposition}[Stationarity] \label{proposition:ART_stationatity}
If $\rho(\widetilde{\mathcal{A}}_1) < 1$ and the process $\mathcal{X}_t$ is weakly stationary, then the ART process in eq. \eqref{eq:model_general_contracted_product}, with $p=1$, is weakly stationary and admits the representation
\begin{equation*}
\mathcal{Y}_t = (\mathcal{I} -\widetilde{\mathcal{A}}_1)^{-1} \bar{\times}_{N} \widetilde{\mathcal{A}}_0 + \sum_{k=0}^\infty \widetilde{\mathcal{A}}_1^k \bar{\times}_{N} \widetilde{\mathcal{B}} \bar{\times}_{M} \mathcal{X}_{t-k} + \sum_{k=0}^\infty \widetilde{\mathcal{A}}_1^k \bar{\times}_{N} \mathcal{E}_{t-k}.
\end{equation*}
\end{proposition}

\begin{proposition} \label{proposition:ART_stationatity_from_VAR}
The VAR($p$) in eq. \eqref{eq:model_final_vectorised} is weakly stationary if and only if the ART($p$) in eq. \eqref{eq:model_general_contracted_product} is weakly stationary.
\end{proposition}

\subsection{Parametrization}
The unrestricted model in eq. \eqref{eq:model_general} cannot be estimated, as the number of parameters greatly outmatches the available data. We address this issue by assuming a PARAFAC($R$) decomposition for the tensor coefficients, which makes the estimation feasible by reducing the dimension of the parameter space.
The models in eqq. \eqref{eq:model_general}-\eqref{eq:model_general_contracted_product} are equivalent but the assuming a PARAFAC decomposition for the coefficient tensors leads to different degrees of parsimony, as shown in the following remark.

\begin{remark}[Alternative parametrization via contracted product]
The two models \eqref{eq:model_general} and \eqref{eq:model_general_contracted_product} combined with the PARAFAC decomposition for the tensor coefficients allow for different degree of parsimony. To show this, without loss of generality, focus on the coefficient tensor $\widetilde{\mathcal{A}}_1$ (similar argument holds for $\widetilde{\mathcal{A}}_j$, $j=2,\ldots,p$ and $\widetilde{\mathcal{B}}$).
By assuming a PARAFAC($R$) decomposition for $\widetilde{\mathcal{A}}_1$ in \eqref{eq:model_general_contracted_product} and for $\mathcal{A}_1$ in \eqref{eq:model_general}, we get, respectively
\begin{align*}
\hspace*{-30pt} \widetilde{\mathcal{A}}_1 & = \sum_{r=1}^R \widetilde{\boldsymbol{\alpha}}_1^{(r)} \circ \ldots \circ \widetilde{\boldsymbol{\alpha}}_N^{(r)} \circ \widetilde{\boldsymbol{\alpha}}_{N+1}^{(r)} \circ \ldots \circ \widetilde{\boldsymbol{\alpha}}_{2N}^{(r)}, \quad
\mathcal{A}_1 = \sum_{r=1}^R \boldsymbol{\alpha}_1^{(r)} \circ \ldots \circ \boldsymbol{\alpha}_N^{(r)} \circ \boldsymbol{\alpha}_{N+1}^{(r)},
\end{align*}
The length of the vectors $\boldsymbol{\alpha}_j^{(r)}$ and $\widetilde{\boldsymbol{\alpha}}_j^{(r)}$ coincide for each $j=1,\ldots,N$. However, $\boldsymbol{\alpha}_{N+1}^{(r)}$ has length $I^*$ while $\widetilde{\boldsymbol{\alpha}}_{N+1}^{(r)},\ldots,\widetilde{\boldsymbol{\alpha}}_{2N}^{(r)}$ have length $I_1,\ldots,I_N$, respectively. Therefore, the number of free parameters in the coefficient tensor $\mathcal{A}_1$ is $R(I_1 + \ldots + I_N + \prod_{j=1}^N I_j)$, while it is $2R(I_1 + \ldots + I_N)$ for  $\widetilde{\mathcal{A}}_1$. This highlights the greater parsimony granted by the use of the PARAFAC($R$) decomposition in model \eqref{eq:model_general_contracted_product} as compared to model \eqref{eq:model_general}.
\end{remark}

\begin{remark}[Vectorization]
There is a relation between the $(I_1\times\ldots\times I_N)$-dimensional ART($p$) and a $(I_1\cdot\ldots\cdot I_N)$-dimensional VAR($p$) model. The vector form of \eqref{eq:model_general} is
\begin{align} \notag
\hspace*{-25pt}  \vecc{\mathcal{Y}_t} & = \vecc{\mathcal{A}_0} + \sum_{j=1}^p \operatorname{mat}_{N+1}(\mathcal{A}_j) \vecc{\mathcal{Y}_{t-j}} + \operatorname{mat}_{N+1}(\mathcal{B}) \vecc{\mathcal{X}_t} + \vecc{\mathcal{E}_t} \\
\hspace*{-25pt}  \mathbf{y}_t & = \boldsymbol{\alpha}_0 + \sum_{j=1}^p \mathbf{A}_{(N+1),j}' \mathbf{y}_{t-j} + \mathbf{B}_{(N+1)}' \mathbf{x}_t + \boldsymbol{\epsilon}_t, \quad \boldsymbol{\epsilon}_t \sim \mathcal{N}_{I^*}(\mathbf{0}, \Sigma_N \otimes \ldots \otimes \Sigma_1),
\label{eq:model_ARTp_vectorized}
\end{align}
where the constraint on the covariance matrix stems from the one-to-one relation between the tensor normal distribution for $\mathcal{X}$ and the distribution of its vectorization (\cite{Ohlson13TensorNormal}) given by $\mathcal{X} \sim \mathcal{N}_{I_1,\ldots,I_N}(\mathcal{M},\Sigma_1,\ldots,\Sigma_N)$ if and only if $\vecc{\mathcal{X}} \sim \mathcal{N}_{I^*}(\vecc{\mathcal{M}},\Sigma_N \otimes \ldots \otimes \Sigma_1)$. The restriction on the covariance structure for the vectorized tensor provides a parsimonious parametrization of the multivariate normal distribution, while allowing both within and between mode dependence.
Alternative parametrizations for the covariance lead to generalizations of standard models. For example, assuming an additive covariance structure results in the tensor ANOVA.
%\cite{Griffin19Lasso_Tensor_ANOVA}
This is an active field for further research.
\end{remark}

\begin{example}
For the sake of exposition, consider the model in eq. \eqref{eq:model_general}, where $p=1$, the response is a 3-order tensor $\mathcal{Y}_t\in\R^{d\times d\times d}$ and the covariates include only a constant coefficient tensor $\mathcal{A}_0$. Define by $k_{\mathcal{E}}$ the number of parameters of the noise distribution. The total number of parameters to estimate in the unrestricted case is $(d^{2N}) + k_{\mathcal{E}} = O(d^{2N})$, with $N=3$ in this example.
Instead, in a ART model defined via the mode-$n$ product in eq. \eqref{eq:model_general}, assuming a PARAFAC($R$) decomposition on $\mathcal{A}_0$ the total number of parameters is $\sum_{r=1}^R (d^N+d^N) + k_{\mathcal{E}} = O(d^N)$.
Finally, in the ART model defined by the contracted product in eq. \eqref{eq:model_general_contracted_product} with a PARAFAC($R$) decomposition on $\widetilde{\mathcal{A}}_0$ the number of parameters is $\sum_{r=1}^R Nd + k_{\mathcal{E}} = O(d)$.
A comparison of the different parsimony granted by the PARAFAC decomposition in all models is illustrated in Fig. \ref{fig:parameters}.

\begin{figure}[t]
\centering
\captionsetup{width=0.9\linewidth}
\setlength{\abovecaptionskip}{1pt}
\setlength{\tabcolsep}{4pt}
\begin{tabular}{ccc}
{\footnotesize (a) vectorized} & {\footnotesize (b) contracted product $\bar{\times}_{N}$ as \eqref{eq:model_general}} & {\footnotesize (c) contracted product $\bar{\times}_{N}$ as \eqref{eq:model_general_contracted_product}} \\[2pt]
\includegraphics[trim= 10mm 5mm 10mm 5mm,clip,height= 4.0cm,width= 4.6cm]{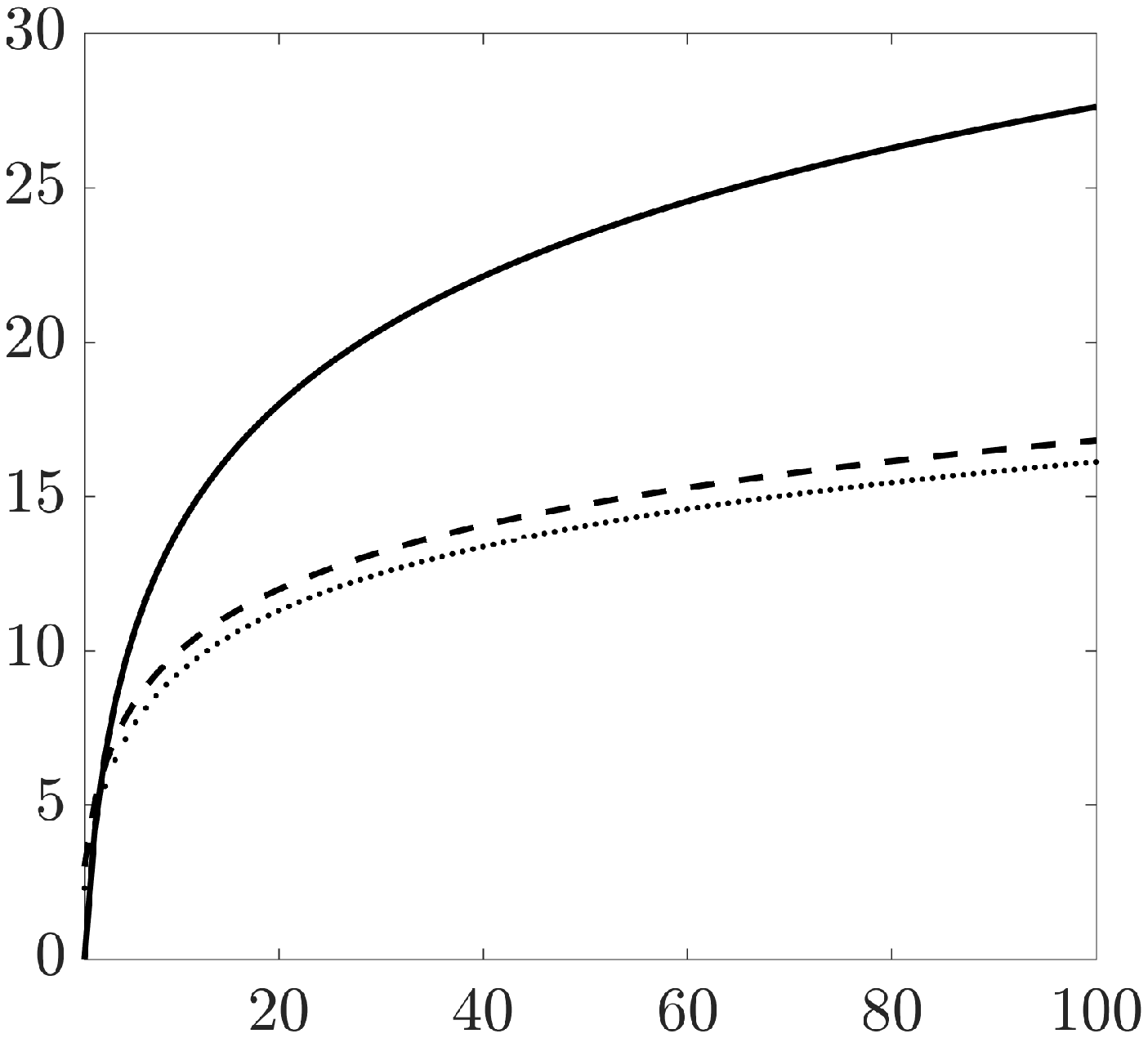} & 
\includegraphics[trim= 10mm 5mm 10mm 5mm,clip,height= 4.0cm,width= 4.6cm]{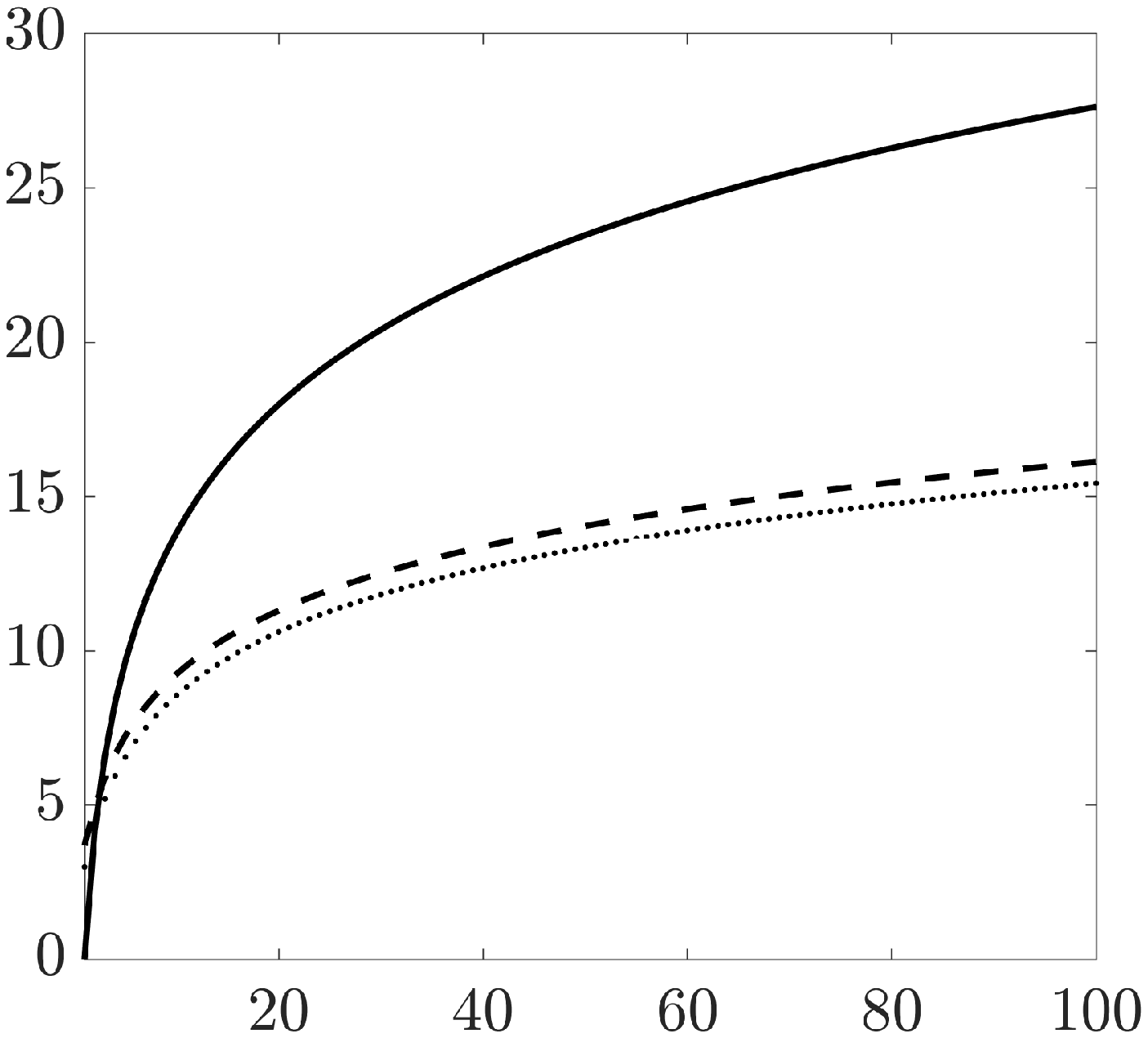} & 
\includegraphics[trim= 10mm 5mm 10mm 5mm,clip,height= 4.0cm,width= 4.6cm]{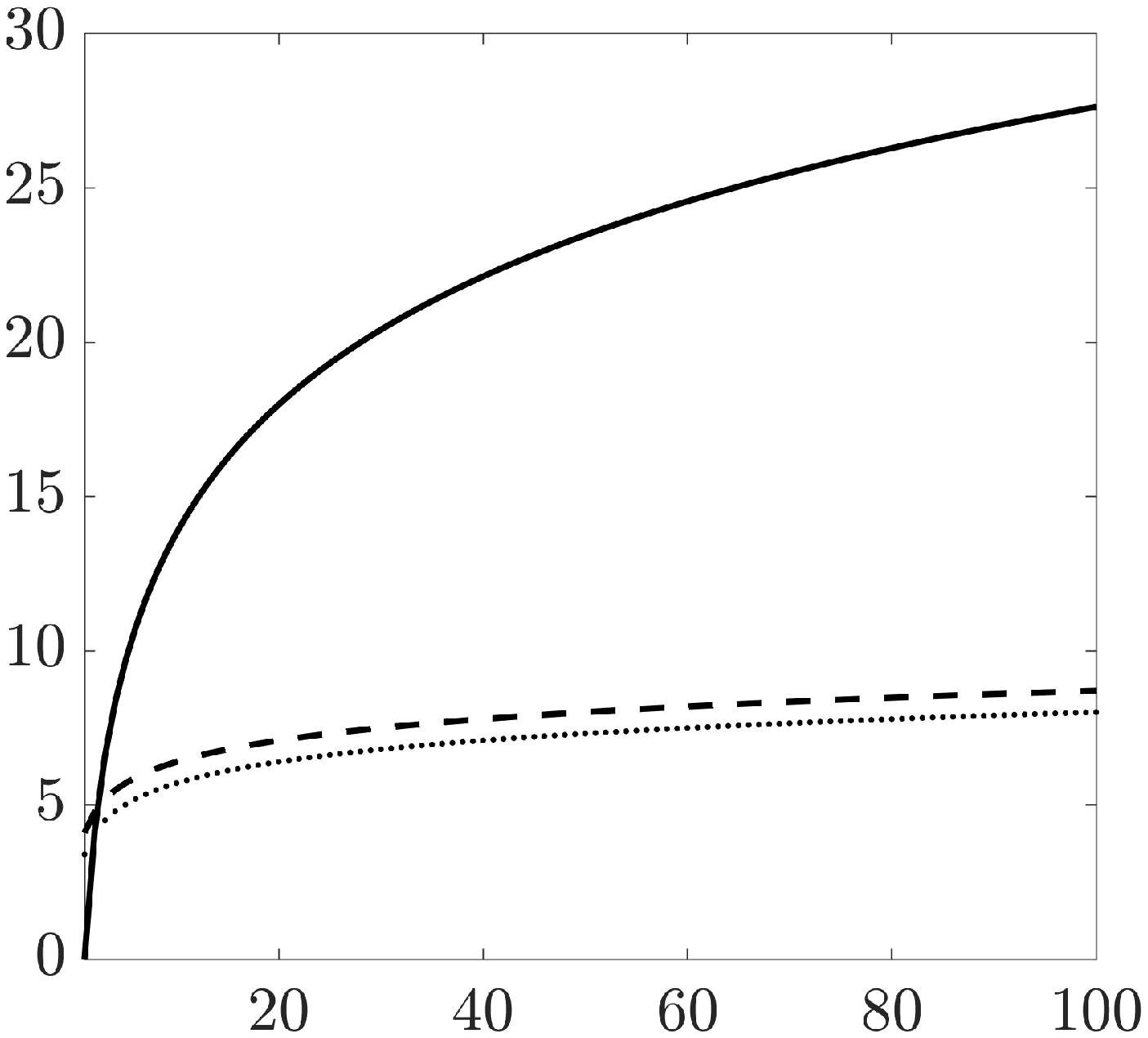}
\end{tabular}
\caption{Number of parameters in $\mathcal{A}_0$, in log-scale (\textit{vertical axis}) as function of the size $d$ of the $(d\times d\times d)$-dimensional tensor $\mathcal{Y}_t$ (\textit{horizontal axis}) in a ART(1) model. In all plots: unconstrained model (\textit{solid} line), PARAFAC($R$) parametrization with $R=10$ (\textit{dashed} line) and $R=5$ (\textit{dotted} line). Parametrizations: vectorized model (panel \textit{a}), mode-$n$ product of \protect\eqref{eq:model_general} (panel \textit{b}) and contracted product of \protect\eqref{eq:model_general_contracted_product} (panel \textit{c}).}
\label{fig:parameters}
\end{figure}
\end{example}

The structure of the PARAFAC decomposition poses an identification problem for the marginals $\boldsymbol{\beta}_j^{(r)}$, which may arise from three sources:
\begin{enumerate}[label=(\roman*)]
\item \textit{scale} identification, since $\lambda_{jr} \boldsymbol{\beta}_j^{(r)} \circ \lambda_{kr} \boldsymbol{\beta}_k^{(r)} = \boldsymbol{\beta}_j^{(r)} \circ \boldsymbol{\beta}_k^{(r)}$ for any collection $\lbrace \lambda_{jr} \rbrace_{j,r}$ such that $\prod_{j=1}^J \lambda_{jr}=1$;
\item \textit{permutation} identification, since for any permutation of the indices $\lbrace 1,\ldots,R \rbrace$ the outer product of the original vectors is equal to that of the permuted ones;
\item \textit{orthogonal transformation} identification, since $\boldsymbol{\beta}_j^{(r)}Q \circ \boldsymbol{\beta}_k^{(r)} Q = \boldsymbol{\beta}_j^{(r)} \circ \boldsymbol{\beta}_k^{(r)}$ for any orthonormal matrix $Q$.
\end{enumerate}
Note that in our framework these issues do not hamper the inference, since our object of interest is the coefficient tensor $\mathcal{B}$, which is exactly identified. The marginals $\boldsymbol{\beta}_j^{(r)}$ have no interpretation, as the PARAFAC decomposition is assumed on the coefficient tensor for the sake of providing a parsimonious parametrization.

\subsection{Important Special Cases}
The model in eq. \eqref{eq:model_general} is a generalization of several well-known econometric models, as shown in the following remarks. See the supplement for the proofs of these results.

\begin{remark}[Univariate] \label{remark:univ_reg}
If $I_i=1$ for $i=1,\ldots,N$, then model \eqref{eq:model_general} reduces to a univariate regression
\begin{equation}
y_t = \alpha_0 + \sum_{j=1}^p \alpha_j y_{t-j} + \boldsymbol{\beta}' \vecc{\mathcal{X}_t} +\epsilon_t \qquad \epsilon_t \sim \mathcal{N}(0,\sigma^2),
\label{eq:model_univ_reg}
\end{equation}
where the coefficients of \eqref{eq:model_general} become $\mathcal{A}_j = \alpha_j \in\R$, $j=0,\ldots,p$ and $\mathcal{B} = \boldsymbol{\beta} \in\R^{J^*}$.
\end{remark}

\begin{remark}[SUR] \label{remark:sur}
If $I_i=1$ for $i=2,\ldots,N$ and define by $\mathbf{1}_n$ the unit vector of length $n$, then model \eqref{eq:model_general} reduces to a Seemingly Unrelated Regression (SUR) model (\cite{Zellner62SUR})
\begin{equation}
\mathbf{y}_t = \boldsymbol{\alpha}_0 + B \times_2 \vecc{\mathcal{X}_t} + \boldsymbol{\epsilon}_t \qquad \boldsymbol{\epsilon}_t \sim\mathcal{N}_{m}(\mathbf{0},\Sigma),
\label{eq:model_sur}
\end{equation}
where $I_1=m$ and the coefficients of \eqref{eq:model_general} become $\mathcal{A}_j = 0$, $j=1,\ldots,p$, $\mathcal{A}_0 = \boldsymbol{\alpha}_0 \in \R^{m}$ and $\mathcal{B} = B \in\R^{m \times J^*}$. Note that, by definition, $B \times_2 \vecc{\mathcal{X}_t} = B \vecc{\mathcal{X}_t}$.
%with coefficients $\mathcal{A} = \boldsymbol{\alpha}\in\R^m$, $\mathcal{B} = \bar{B} \in\R^{m\times J}$ and $\mathcal{C} = C \in\R^{m\times Q}$.
\end{remark}

\begin{remark}[VARX and Panel VAR] \label{remark:var}
Consider the setup of Remark \ref{remark:sur}. If $\mathbf{z}_t = \mathbf{y}_{t-1}$, then weoobtain a VARX(1) model, with restricted covariance matrix.
Another vector of regressors $\mathbf{w}_t = \vecc{W_t} \in \R^q$ may enter the regression \eqref{eq:model_sur} pre-multiplied (along mode-$3$) by a tensor $\mathcal{D} \in \R^{m\times n\times q}$.
Therefore, model \eqref{eq:model_general} encompasses as a particular case also the panel VAR models of \cite{CanovaCiccarelli04PanelVAR}, \cite{CanovaCiccarelliOrtega07pVAR}, \cite{CanovaCiccarelli09pVAR}, provided that we make the same restriction on 	$\Sigma$.
\end{remark}

\begin{remark}[VECM] \label{remark:vecm}
The model in eq. \eqref{eq:model_general} generalises the Vector Error Correction Model (VECM) widely used in multivariate time series analysis (see \cite{EngleGranger87Cointegration_VECM}, \cite{SchotmanVanDijk91BayesUnitRoot}).
Consider a $K$-dimensional VAR(1) model
\begin{equation*}
\mathbf{y}_t = B \mathbf{y}_{t-1} + \boldsymbol{\epsilon}_t \qquad \boldsymbol{\epsilon}_t \sim \mathcal{N}_m(\mathbf{0},\Sigma).
\end{equation*}
Defining $\Delta \mathbf{y}_t = \mathbf{y}_t-\mathbf{y}_{t-1}$ and $\Pi = (B-I) = \boldsymbol{\alpha}\boldsymbol{\beta}'$, where $\boldsymbol{\alpha}$ and $\boldsymbol{\beta}$ are $K\times R$ matrices of rank $R<K$, we obtain the associated VECM
\begin{equation}
\Delta \mathbf{y}_t = \boldsymbol{\alpha}\boldsymbol{\beta}' \mathbf{y}_{t-1} + \boldsymbol{\epsilon}_t.
\label{eq:remark_vecm}
\end{equation}
This is used for studying the cointegration relations among the components of $\mathbf{y}_t$.
Since $\Pi = \boldsymbol{\alpha}\boldsymbol{\beta}' = \sum_{r=1}^R \boldsymbol{\alpha}_{:,r} \boldsymbol{\beta}_{:,r}' = \sum_{r=1}^R \tilde{\boldsymbol{\beta}}_1^{(r)} \circ \tilde{\boldsymbol{\beta}}_2^{(r)}$, we can interpret the VECM model in eq. \eqref{eq:remark_vecm} as a particular case of the model in eq. \eqref{eq:model_general} where the coefficient $\mathcal{B}$ is the matrix $\Pi = \boldsymbol{\alpha}\boldsymbol{\beta}'$. Furthermore by writing $\Pi = \sum_{r=1}^R \tilde{\boldsymbol{\beta}}_1^{(r)} \circ \tilde{\boldsymbol{\beta}}_2^{(r)}$ we can interpret this relation as a rank-R PARAFAC decomposition of $\mathcal{B}$.
Following this analogy, the PARAFAC rank corresponds to the cointegration rank, $\tilde{\boldsymbol{\beta}}_1^{(r)}$ are the mean-reverting coefficients and $\tilde{\boldsymbol{\beta}}_2^{(r)} = (\tilde{\beta}_{2,1}^{(r)},\ldots,\tilde{\beta}_{2,K}^{(r)})$ are the cointegrating vectors. See the supplement for details.
%In fact, the PARAFAC($R$) decomposition for matrices corresponds to a low rank ($R$) matrix approximation (see \cite{Eckart36LowRankMatrixApproximation}).
This interpretation opens the way to reparametrization of $\mathcal{B}$ based on tensor SVD representations, and to the application of regularization methods in the spirit of \cite{Basturk17NearBoundary_ReducedRank}. This is beyond the scope of the paper, thus we leave it for further research.
\end{remark}

\begin{remark}[MAI of \cite{Carrieroetal16Multivariate_AR_Index}] \label{remark:MAI}
The multivariate autoregressive index model (MAI) of \cite{Carrieroetal16Multivariate_AR_Index} is another special case of model \eqref{eq:model_general}. A MAI is a VAR model with a low rank decomposition imposed on the coefficient matrix, as follows
\begin{equation*}
\mathbf{y}_t = \mathbf{A}\mathbf{B}_0 \mathbf{y}_{t-1} + \boldsymbol{\epsilon}_t,
\end{equation*}
where $\mathbf{y}_t$ is a $(n\times 1)$ vector, whereas $\mathbf{A},\mathbf{B}_0$ are $(n\times R)$ and $(R\times n)$ matrices, respectively. In \cite{Carrieroetal16Multivariate_AR_Index}, the authors assumed $R=1$. This corresponds to our parametrization using $R=1$ and defining $\mathbf{A} \boldsymbol{\beta}_1^{(1)}$ and $\mathbf{B}_0' = \boldsymbol{\beta}_2^{(1)}$, which leads us to $\mathbf{A} \mathbf{B}_0 = \boldsymbol{\beta}_1^{(1)} \circ \boldsymbol{\beta}_2^{(1)}$.
\end{remark}

\begin{remark}[Tensor autoregressive model (ART)] \label{remark:tensorAR}
By removing all the covariates from eq. \eqref{eq:model_general} except the lags of the dependent variable, we obtain a tensor autoregressive model of order $p$ (or ART($p$))
\begin{equation}
\mathcal{Y}_t = \mathcal{A}_0 + \sum_{j=1}^p \mathcal{A}_j \times_{N+1} \vecc{\mathcal{Y}_{t-j}} + \mathcal{E}_t, \qquad \mathcal{E}_t \distas{iid} \mathcal{N}_{I_1,\ldots,I_N}(\mathbf{0},\Sigma_1,\ldots,\Sigma_N).
\label{eq:tensorAR}
\end{equation}
\end{remark}

Matrix autoregressive models (MAR) are another special case of \eqref{eq:model_general}, which can be obtained from eq. \eqref{eq:tensorAR} when the dependent variable is a matrix. See the supplement for an example.

\subsection{Impulse Response Analysis} \label{sec:IRF}
In this section we derive two impulse response functions (IRF) for ART models, the block Cholesky IRF and the block generalised IRF, exploiting the relationship between ART and VAR models. Without loss of generality, we focus on the ART($p$) model in eq. \eqref{eq:tensorAR}, with $p=1$ and $\mathcal{A}_0 = \mathbf{0}$, and introduce the following notation.
Let $\mathbf{y}_t = \vecc{\mathcal{Y}_t}$ and $\boldsymbol{\epsilon}_t = \vecc{\mathcal{E}_t} \sim \mathcal{N}_{I^*}(\mathbf{0},\Sigma)$ be the $(I^* \times 1)$ tensor response and noise term in vector form, respectively, where $\Sigma = \Sigma_N \otimes \ldots \otimes \Sigma_1$ is the $(I^* \times I^*)$ covariance of the model in vector form and $I^* = \prod_{k=1}^N I_k$.
Partition $\Sigma$ in blocks as
\begin{equation}
\renewcommand*{\arraystretch}{0.72}
\Sigma = \bigg( \begin{array}{c|c} A & B \\ \hline B' & C \end{array} \bigg),
\label{eq:Sigma_block}
\end{equation}
where $A$ is $n \times n$, $B$ is $n \times (I^*-n)$ and $C$ is $(I^*-n)\times (I^*-n)$. Then, denoting by $S = C - B' A^{-1} B$ the Schur complement of $A$, the LDU decomposition of $\Sigma$ is
\begin{equation*}
\renewcommand*{\arraystretch}{0.72}
\Sigma = 
\bigg( \begin{array}{c|c} I_n & \bigzero_{n,I^*-n} \\ \hline B' A^{-1} & I_{I^*-n} \end{array} \bigg)
\bigg( \begin{array}{c|c} A & \bigzero_{n,I^*-n} \\ \hline \bigzero_{n,I^*-n}' & S \end{array} \bigg)
\bigg( \begin{array}{c|c} I_n & A^{-1} B \\ \hline \bigzero_{n,I^*-n}' & I_{I^*-n} \end{array} \bigg)
= L D L'.
\label{eq:Sigma_LDU}
\end{equation*}
Hence $\Sigma$ can be block-diagonalised
\begin{equation}
\renewcommand*{\arraystretch}{0.72}
D = L^{-1} \Sigma (L')^{-1} = \bigg( \begin{array}{c|c} A & \bigzero_{n,I^*-n} \\ \hline \bigzero_{n,I^*-n}' & S \end{array} \bigg).
\label{eq:Sigma_LDU_inverse}
\end{equation}
From the Cholesky decomposition of $D$ one obtains a block Cholesky decomposition
\begin{equation*}
\renewcommand*{\arraystretch}{0.72}
\Sigma = \bigg( \begin{array}{c|c} L_A & \bigzero_{n,I^*-n} \\ \hline B' (L_A^{-1})' & L_S \end{array} \bigg) \bigg( \begin{array}{c|c} L_A' & L_A^{-1} B \\ \hline \bigzero_{n,I^*-n}' & L_S' \end{array} \bigg) = P P',
\end{equation*}
where $L_A,L_S$ are the Cholesky factors of $A$ and $S$, respectively.
Assume the vectorised ART process admits an infinite MA representation, with $\Psi_0 = I_{I^*}$ and $\Psi_i = \operatorname{mat}_{(4)}(\mathcal{B})' \Psi_{i-1}$, then using the previous results we get:
\begin{align}
\mathbf{y}_t & = \sum_{i=0}^\infty \Psi_i \boldsymbol{\epsilon}_{t-i} = \sum_{i=0}^\infty (\Psi_i L) (L^{-1} \boldsymbol{\epsilon}_{t-i}) = \sum_{i=0}^\infty (\Psi_i L) \boldsymbol{\eta}_{t-i} \quad \boldsymbol{\eta}_t \sim \mathcal{N}_{I^*}(\mathbf{0},D),
\label{eq:IRF_step1}
\end{align}
where $\boldsymbol{\eta}_t = L^{-1} \boldsymbol{\epsilon}_t$ are the block-orthogonalised shocks and $D$ is the block-diagonal matrix in eq. \eqref{eq:Sigma_LDU_inverse}.
Denote with $E_n$ the $I^* \times n$ matrix that selects $n$ columns from a pre-multiplying matrix, i.e. $D E_n$ is a matrix containing $n$ columns of $D$.
Denote with $\boldsymbol{\delta}^*$ a $n$-dimensional vector of shocks. Using the property of the multivariate Normal distribution, and recalling that the top-left block of size $n$ of $D$ is $A$, we extend the generalised IRF of \cite{Koop96Generalized_ImpulseResponse} and \cite{Pesaran98Generalized_ImpulseResponse} by defining the block generalised IRF
\begin{align}
\notag
\boldsymbol{\psi}^G(h;n) & = \mathbb{E}\big( \vecc{\mathcal{Y}_{t+h}} | \vecc{\mathcal{E}_t}' = (\boldsymbol{\delta}^{*\prime},\mathbf{0}_{I^*-n}'),\mathcal{F}_{t-1} \big) - \mathbb{E}\big( \vecc{\mathcal{Y}_{t+h}} | \mathcal{F}_{t-1} \big) \\
 & = (\Psi_h L) D E_n A^{-1} \boldsymbol{\delta}^*,
\label{eq:GIRF_new}
\end{align}
where $\mathcal{F}_t$ is the natural filtration associated to the stochastic process.
Starting from eq. \eqref{eq:IRF_step1} we derive the block Cholesky IRF (OIRF) as
\begin{align}
\notag
\boldsymbol{\psi}^O(h;n) & = \mathbb{E}\big( \vecc{\mathcal{Y}_{t+h}} | \vecc{\mathcal{E}_t}' =(\boldsymbol{\delta}^{*\prime},\mathbf{0}_{I^*-n}'),\mathcal{F}_{t-1} \big) \\ \notag
 & \quad - \mathbb{E}\big( \vecc{\mathcal{Y}_{t+h}} | \vecc{\mathcal{E}_t}' = \mathbf{0}_{I^*}', \mathcal{F}_{t-1} \big) \\
 & = (\Psi_h L) P E_n \boldsymbol{\delta}^*.
\label{eq:OIRF_new}
\end{align}
%{\color{red}
%\begin{lemma}[IRF]
%Consider a ART(1) model in vector form, and let $\boldsymbol{\delta}^*$ be a $n$-dimensional vector of shocks, $n \leq I^*$. Partition the covariance matrix of the associated VAR as
%\begin{equation*}
%\renewcommand*{\arraystretch}{0.72}
%\Sigma = \bigg( \begin{array}{c|c} A & B \\ \hline B' & C \end{array} \bigg),
%\end{equation*}
%with $A$ is a $n \times n$ matrix corresponding to the $n$ variables to be shocked. Then
%\begin{align*}
%\boldsymbol{\psi}^O(h;n) = (\Psi_h L) P E_n \boldsymbol{\delta}^*, \qquad \boldsymbol{\psi}^G(h;n) = (\Psi_h L) D E_n A^{-1} \boldsymbol{\delta}^*
%\end{align*}
%where $P$ is the Cholesky factor of $\Sigma$, $L,D$ provide the LDL decomposition $\Sigma = LDL'$ and $E_n$ is a $I^* \times n$ matrix that selects $n$ columns from a pre-multiplying matrix.
%\end{lemma}
%}
Define with $\mathbf{e}_j$ the $j$-th column of the $I^*$-dimensional identity matrix.
The impact of a shock $\delta^*$ to the $j$-th variable on all $I^*$ variables is given below in eq. \eqref{eq:GIRF_OIRF_new_j}, whereas the impact of a shock to the $j$-th variable on the $i$-th variable is given in eq. \eqref{eq:GIRF_OIRF_new_ij}.
\begin{alignat}{3}
\label{eq:GIRF_OIRF_new_j}
\boldsymbol{\psi}_j^G(h;n) & = & \Psi_h L D \mathbf{e}_j D_{jj}^{-1} \delta^*, \qquad
\boldsymbol{\psi}_j^O(h;n) & = & \Psi_h L P \mathbf{e}_j \delta^* \\
\label{eq:GIRF_OIRF_new_ij}
\psi_{ij}^G(h;n) & = & \; \mathbf{e}_i' \Psi_h L D \mathbf{e}_j D_{jj}^{-1} \delta^*, \qquad
\psi_{ij}^O(h;n) & = & \; \mathbf{e}_i' \Psi_h L P \mathbf{e}_j \delta^*.
\end{alignat}
Finally, denoting $\boldsymbol{\delta}_j = \mathbf{e}_j \delta^*$, we have the compact notation
\begin{alignat*}{3}
%\label{eq:GIRF_OIRF_new_j_compact}
\boldsymbol{\psi}_j^G(h;n) & = & \Psi_h L D D_{jj}^{-1} \boldsymbol{\delta}_j, \qquad
\boldsymbol{\psi}_j^O(h;n) & = & \Psi_h L P \boldsymbol{\delta}_j \\
%\label{eq:GIRF_OIRF_new_ij_compact}
\psi_{ij}^G(h;n) & = & \; \mathbf{e}_i' \Psi_h L D D_{jj}^{-1} \boldsymbol{\delta}_j, \qquad
\psi_{ij}^O(h;n) & = & \; \mathbf{e}_i' \Psi_h L P \boldsymbol{\delta}_j.
\end{alignat*}

\section{Bayesian Inference} \label{sec:bayesian_inference}
In this section, without loss of generality, we present the inference procedure for a special case of the model in eq. \eqref{eq:model_general}, given by
\begin{equation}
\mathcal{Y}_t = \mathcal{B} \times_4 \vecc{\mathcal{Y}_{t-1}} + \mathcal{E}_t, \qquad \mathcal{E}_t \distas{iid} \mathcal{N}_{I_1,I_2,I_3}(\mathbf{0},\Sigma_1,\Sigma_2,\Sigma_3).
\label{eq:model_final}
\end{equation}
Here $\mathcal{Y}_t$ is a $3$-order tensor response of size $I_1\times I_2\times I_3$, $\mathcal{X}_t = \mathcal{Y}_{t-1}$ and $\mathcal{B}$ is thus a $4$-order coefficient tensor of size $I_1\times I_2\times I_3\times I_4$, with $I_4 = I_1 I_2 I_3$.
This is a $3$-order \textit{tensor autoregressive model} of lag-order $1$, or ART($1$), coinciding with eq. \eqref{eq:tensorAR} for $p=1$ and $\mathcal{A}_0 = \mathbf{0}$.
The noise term $\mathcal{E}_t$ has as tensor normal distribution, with zero mean and covariance matrices $\Sigma_1,\Sigma_2,\Sigma_3$ of sizes $I_1\times I_1$, $I_2\times I_2$ and $I_3\times I_3$, respectively, accounting for the covariance along each of the three dimensions of $\mathcal{Y}_t$.
The specification of a tensor model with a tensor normal noise instead of a vector model (like a Gaussian VAR) has the advantage of being more parsimonious. By vectorising \eqref{eq:model_final}, we get the equivalent VAR
\begin{equation}
\vecc{\mathcal{Y}_t} = \mathbf{B}_{(4)}' \vecc{\mathcal{Y}_{t-1}} + \vecc{\mathcal{E}_t}, \quad \vecc{\mathcal{E}_t} \distas{iid} \mathcal{N}_{I^*}(\mathbf{0},\Sigma_3 \otimes \Sigma_2 \otimes \Sigma_1),
\label{eq:model_final_vectorised}
\end{equation}
whose covariance has a Knocker structure, which contains $(I_1(I_1+1) + I_2(I_2+1) + I_3(I_3+1))/2$ parameters (as opposed to $(I^*(I^*+1))/2$ of an unrestricted VAR) and allows for heteroskedasticity.

The choice the Bayesian approach for inference is motivated by the fact that the large number of parameters may lead to an overfitting problem, especially when the samples size is rather small. This issue can be addressed by the indirect inclusion of parameter restrictions through a suitable specification of the corresponding prior distributions.
In the unrestricted model \eqref{eq:model_final} it would be necessary to define a prior distribution on the $4$-order tensor $\mathcal{B}$. The literature on tensor-valued distributions is limited to the elliptical family (e.g., \cite{Ohlson13TensorNormal}), which includes the tensor normal and tensor $t$. Both distributions do not easily allow for the specification of restrictions on a subset of the entries of the tensor, hampering the use of standard regularization prior distributions (such as shrinkage priors).

The PARAFAC($R$) decomposition of the coefficient tensor provides a way to circumvent this issue. This decomposition allows to represent a tensor through a collection of vectors (the marginals), for which many flexible shrinkage prior distributions are available. Indirectly, this introduces \textit{a priori} sparsity on the coefficient tensor.

%By assuming a PARAFAC($R$) decomposition on the coefficient tensor, for achieving two goals: first, by reducing the parameter space this assumption makes estimation feasible; second, the decomposition transforms a multidimensional array into the outer product of vectors, we are left we the choice of a prior distribution on vectors, for which many constructions are available. In particular, we can incorporate sparsity beliefs by specifying a suitable shrinkage prior directly on the marginals of the PARAFAC. Indirectly, this introduces \textit{a priori} sparsity on the coefficient tensor.

\subsection{Prior Specification}
The choice of the prior distribution on the PARAFAC marginals is crucial for recovering the sparsity pattern of the coefficient tensor and for the efficiency of the inference.
Global-local prior distributions are based on scale mixtures of normal distributions, where the different components of the covariance matrix govern the amount of prior shrinkage.
Compared to spike-and-slab distributions (e.g., \cite{Mitchell88SpikeSlab_priors}, \cite{George97SpikeSlabPrior}, \cite{Ishwaran05SpikeSlabPrior}) which become infeasible as the parameter space grows, global-local priors have better scalability properties in high-dimensional settings. They do not provide automatic variable selection, which can nonetheless be obtained by post-estimation thresholding (\cite{Park08BayesianLasso}).

Motivated by these arguments, we define a global-local shrinkage prior for the marginals $\boldsymbol{\beta}_j^{(r)}$ of the coefficient tensor $\mathcal{B}$ following the hierarchical prior specification of \cite{GuhaniyogiDunson17BayesTensorReg} (see also \cite{BattDuns15}, \cite{ZhouBattDuns15}).
For each $\boldsymbol{\beta}_j^{(r)}$, we define a prior distributions as a scale mixture of normals centred in zero, with three components for the covariance. The global parameter $\tau$ governs the overall variance, the middle parameter $\phi_r$ defines the common shrinkage for the marginals in $r$-th component of the PARAFAC, and the local parameter $W_{j,r}= \text{diag}(\mathbf{w}_{j,r})$ drives the shrinkage of each entry of each marginal.
%We add another layer to the hierarchy by assuming each element of $\mathbf{w}_{j,r}$ is exponentially distributed with hyper-parameter $\lambda_{j,r}$ drawn from a Gamma. This is a key parameter for driving the shrinkage to zero of the marginals.
Summarizing, for $p=1,\ldots,I_j$, $j=1,\ldots,J$ ($J=4$ in eq. \eqref{eq:model_final}) and $r=1,\ldots,R$, the hierarchical prior structure\footnote{We use the shape-rate formulation for the gamma distribution.} for each vector of the PARAFAC($R$) decomposition in eq. \eqref{eq:PARAFAC_demposition} is
\begin{align}
\begin{split}
\pi(\boldsymbol{\phi})  \sim \mathcal{D}ir(\alpha \mathbf{1}_R) \quad
\pi(\tau) & \sim \mathcal{G}a(a_{\tau},b_{\tau}) \quad
\pi(\lambda_{j,r})  \sim \mathcal{G}a(a_\lambda,b_\lambda) \\
\pi(w_{j,r,p}|\lambda_{j,r}) & \sim \mathcal{E}xp (\lambda_{j,r}^2/2) \\
\pi\big( \boldsymbol{\beta}_j^{(r)} \big| W_{j,r},\boldsymbol{\phi},\tau \big) & \sim \mathcal{N}_{I_j}(\mathbf{0}, \tau \phi_r W_{j,r}),
\end{split}
\label{eq:prior_beta}
\end{align}
where $\mathbf{1}_R$ is the vector of ones of length $R$ and we assume $a_\tau = \alpha R$ and $b_\tau = \alpha R^{1/J}$.
The conditional prior distribution of a generic entry $b_{i_1,\ldots,i_J}$ of $\mathcal{B}$ is the law of a sum of product Normals\footnote{A product Normal is the distribution of the product of $n$ independent centred Normal random variables.}: it is symmetric around zero, with fatter tails than both a standard Gaussian or a standard Laplace distribution (see the supplement for further details). Note that a product Normal prior promotes sparsity due to the peak at zero.
The following result characterises the conditional prior distribution of an entry of the coefficient tensor $\mathcal{B}$ induced by the hierarchical prior in eq. \eqref{eq:prior_beta}. See the supplement for the proof.

\begin{lemma} \label{lemma:prior_entries_B}
Let $b_{ijkp} = \sum_{r=1}^R \beta_r$, where $\beta_r = \beta_{1,i}^{(r)} \beta_{2,j}^{(r)} \beta_{3,k}^{(r)} \beta_{4,p}^{(r)}$, and let $m_1=i$, $m_2=j$, $m_3=k$ and $m_4=p$. Under the prior specification in \eqref{eq:prior_beta}, the generic entry $b_{ijkp}$ of the coefficient tensor $\mathcal{B}$ has the conditional prior distribution
\begin{align*}
\pi(b_{ijkp} | \tau, \boldsymbol{\phi}, \mathbf{W}) & = p\bigg( \sum_{r=1}^R \beta_r \big| - \bigg) = p(\beta_1 | -) \ast \ldots \ast p(\beta_R | -),
% & = \sum_{r=1}^R  H_r \cdot G_{4,0}^{4,0}\Bigl( \beta_r^2 \cdot \prod_{h=1}^4 \frac{1}{2\tau \phi_r} W_{h,r}^{-1} \Bigl| \mathbf{0} \Bigr) \\
% & \propto \sum_{r=1}^R  G_{4,0}^{4,0}\Bigl( \beta_r^2 \cdot \prod_{h=1}^4 \frac{1}{2\tau \phi_r} W_{h,r}^{-1} \Bigl| \mathbf{0} \Bigr),
\end{align*}
where $\ast$ denotes convolution and
\begin{align*}
p(\beta_r | -) & = K_r \cdot G_{4,0}^{4,0}\Bigl( \beta_r^2 \prod_{h=1}^4 (2\tau \phi_r w_{h,r,m_h})^{-1} \Bigl| \mathbf{0} \Bigr),
\end{align*}
with $G_{p,q}^{m,n}(x| \textcolor{white}{a}_{\mathbf{b}}^{\mathbf{a}})$ a Meijer G-function and
\begin{align*}
G_{4,0}^{4,0}\Bigl( \beta_r^2 \prod_{h=1}^4 (2\tau \phi_r w_{h,r,m_h})^{-1} \Bigl| \mathbf{0} \Bigr) & = \frac{1}{2\pi i} \int_{c-i^\infty}^{c+i^\infty} \Bigl( \beta_r^2 \prod_{h=1}^4 (2\tau \phi_r w_{h,r,m_h})^{-1} \Bigr)^{-s} \: \mathrm{d}s \\
K_r & = (2\pi)^{-4/2} \prod_{h=1}^4 (2\tau \phi_r w_{h,r,m_h})^{-1} \, .
\end{align*}
\end{lemma}

The use of Meijer G- and Fox H-functions is not new in econometrics (e.g., \cite{Abadir97MixedNormal_Meijer-G-function}), and they have been recently used for defining prior distributions in Bayesian statistics (\cite{Andrade15H-functions_prior_posterior}, \cite{Andrade17G-Meijer_prior_posterior}).

%%%%% \Sigma_1 = \Sigma_r   \Sigma_2 = \Sigma_c %%%%%
From eq. \eqref{eq:tensor_normal_dsitrbution}, we have that the covariance matrices $\Sigma_j$ enter the likelihood in a multiplicative way, therefore separate identification of their scales requires further restrictions. \cite{Wang09Bayes_matrixNormalGraph} and \cite{Dobra15TensorKalmanFilter} adopt independent hyper-inverse Wishart prior distributions (\cite{Dawid93HyperMarkov_decomposableGraphs}) for each $\Sigma_j$, then impose the identification restriction $\Sigma_{j,11} = 1$ for $j=2,\ldots,J-1$. The hard constraint $\Sigma_j=\mathbf{I}_{I_j}$ (where $\mathbf{I}_j$ is the identity matrix of size $j$), for all but one $n$, implicitly imposes that the dependence structure within different modes is the same, but there is no dependence between modes.
We follow \cite{Hoff11SeparableCovArray_Tucker}, who suggests to introduce dependence between the Inverse Wishart prior distribution of each $\Sigma_j$ via a hyper-parameter $\gamma$ affecting their prior scale. To account for marginal dependence, we add a level of hierarchy, thus obtaining
\begin{align}
\pi(\gamma) \sim \mathcal{G}a(a_\gamma,b_\gamma) \qquad \pi(\Sigma_j | \gamma) & \sim \mathcal{IW}_{I_j}(\nu_j,\gamma\Psi_j).
%\label{eq:prior_gamma}
%\pi(\gamma) & \sim \mathcal{G}a(a_\gamma,b_\gamma) \\
%\pi(\Sigma_j | \gamma) & \sim \mathcal{IW}_{I_j}(\nu_j,\gamma\Psi_j).
\label{eq:prior_Sigmas}
\end{align}
Define $\Lambda = \lbrace \lambda_{j,r} : j=1,\ldots,J, \: r=1,\ldots,R \rbrace$ and $\mathbf{W} = \lbrace W_{j,r} : j=1,\ldots,J, \: r=1,\ldots,R \rbrace$, and let $\boldsymbol{\theta}$ denote the collection of all parameters.
The directed acyclic graph (DAG) of the prior structure is given in Fig. \ref{fig:flow_prior}.

Note that our prior specification is flexible enough to include Minnesota-type restrictions or hierarchical structures as in \cite{CanovaCiccarelli04PanelVAR}.

\begin{figure}[t]
\centering
\captionsetup{width=0.9\linewidth}
\includegraphics[trim= 0mm 0mm 0mm 0mm,clip,scale= 1.00]{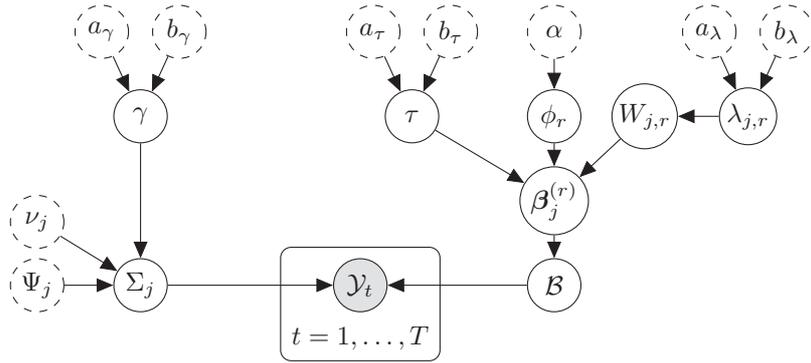}
\caption{Directed acyclic graph of the model in eq. \protect\eqref{eq:model_final} and prior structure in eqq. \protect\eqref{eq:prior_beta}-\protect\eqref{eq:prior_Sigmas}. Gray circles denote observable variables, white solid circles indicate parameters, white dashed circles indicate fixed hyperparameters. Directed edges represent the conditional independence relationships.}
\label{fig:flow_prior}
\end{figure}
%%%%%%%%%%%%%%%%%%%%%%%%%%%%%%%%%%%%%%%%%%%%%%%%%%%%%%%%%%%%

\subsection{Posterior Computation}
Define $\mathbf{Y} = \lbrace \mathcal{Y}_t \rbrace_{t=1}^T$, $I_0 = \sum_{j=1}^J I_j$, $\boldsymbol{\beta}_{-j}^{(r)} = \lbrace \boldsymbol{\beta}_i^{(r)}: i\neq j \rbrace$ and $\mathcal{B}_{-r} = \lbrace B_i : i\neq r \rbrace$, with $B_r = \boldsymbol{\beta}_1^{(r)}\circ\ldots\circ\boldsymbol{\beta}_4^{(r)}$. The likelihood function of model \eqref{eq:model_final} is
\begin{align}
\notag
 & L(\mathbf{Y} | \boldsymbol{\theta}) = \prod_{t=1}^T (2\pi)^{-\frac{I_4}{2}} \prod_{j=1}^3 \abs{\Sigma_j}^{-\frac{I_{-j}}{2}} \\
 & \; \cdot \exp\Big( -\frac{1}{2} \Sigma_2^{-1} (\mathcal{Y}_t -\mathcal{B}\times_4 \mathbf{y}_{t-1}) \times_{1\ldots 3}^{1\ldots 3} \big( \circ_{j=1}^3 \Sigma_j^{-1} \big) \times_{1\ldots 3}^{1\ldots 3} (\mathcal{Y}_t -\mathcal{B}\times_4 \mathbf{y}_{t-1}) \Big),
\label{eq:likelihood}
\end{align}
where $\mathbf{y}_{t-1} =\vecc{\mathcal{Y}_{t-1}}$.
Since the posterior distribution is not tractable in closed form, we adopt an MCMC procedure based on Gibbs sampling. The technical details of the derivation of the posterior distributions are given in Appendix \ref{sec:apdx_computational_tensor}. We articulate the sampler in three main blocks:
\begin{enumerate}[label=(\Roman*)]
\item sample the global and middle variance hyper-parameters of the marginals, from
\begin{align}
\label{eq:posterior_psi}
p(\psi_r|\mathcal{B},\mathbf{W},\alpha) & \propto \textnormal{GiG} \big( \alpha -I_0/2, \, 2b_\tau, \, 2C_r \big) \\
p(\tau|\mathcal{B},\mathbf{W},\boldsymbol{\phi}) & \propto \textnormal{GiG} \big( a_\tau -R I_0/2, \, 2b_\tau, \, 2\sum_{r=1}^R C_r/\phi_r \big),
\label{eq:posterior_tau}
\end{align}
where $C_r = \sum_{j=1}^J \boldsymbol{\beta}_j^{(r)'} W_{j,r}^{-1} \boldsymbol{\beta}_j^{(r)}$, then set $\phi_r = \psi_r/\sum_{l=1}^{R} \psi_l$. For improving the mixing, we sample $\tau$ with a Hamiltonian Monte Carlo (HMC) step (\cite{Neal11HamiltonianMC}).

\item sample the hyper-parameters of the local variance component of the marginals and the marginals themselves, from
\begin{align}
\label{eq:posterior_lambda}
& p\big( \lambda_{j,r}|\boldsymbol{\beta}_j^{(r)},\phi_r,\tau \big) \propto \mathcal{G}a \big( a_\lambda +I_j, \, b_\lambda + \norm{\boldsymbol{\beta}_j^{(r)}}_1 (\tau \phi_r)^{-1/2} \big) \\
%\frac{\norm{\boldsymbol{\beta}_j^{(r)}}_1}{\sqrt{\tau \phi_r}}
\label{eq:posterior_w}
& p\big( w_{j,r,p}|\lambda_{j,r},\phi_r,\tau,\boldsymbol{\beta}_j^{(r)} \big) \propto \textnormal{GiG}\big( 1/2, \, \lambda_{j,r}^2, \, (\beta_{j,p}^{(r)})^2/(\tau\phi_r) \big) \\
\label{eq:posterior_betas}
& p\big( \boldsymbol{\beta}_j^{(r)}|\boldsymbol{\beta}_{-j}^{(r)},\mathcal{B}_{-r},W_{j,r},\phi_r,\tau,\mathbf{Y},\Sigma_1,\ldots,\Sigma_3 \big) \propto \mathcal{N}_{I_j}(\bar{\boldsymbol{\mu}}_{\boldsymbol{\beta}_j}, \, \bar{\Sigma}_{\boldsymbol{\beta}_j}).
\end{align}

\item sample the covariance matrices and the latent scale, respectively, from
\begin{align}
\label{eq:posterior_Sigmar}
p(\Sigma_j|\mathcal{B},\mathbf{Y},\Sigma_{-j},\gamma) & \propto \mathcal{IW}_{I_j}(\nu_j + I_j, \,\gamma\Psi_j + S_j) \\
p(\gamma|\Sigma_1,\ldots,\Sigma_3) & \propto \mathcal{G}a \Big( a_\gamma + \sum_{j=1}^3 \nu_j I_j, \, b_\gamma + \sum_{j=1}^3 \textnormal{tr}(\Psi_j \Sigma_j^{-1}) \Big).
\label{eq:posterior_gamma}
\end{align}
\end{enumerate}

%%%%% Exports = DataExpLog --> Log of exports(=imports in lower part) %%%%%
\section{Application to Multilayer Dynamic Networks} \label{sec:applications}
We apply the proposed methodology to study jointly the dynamics of international trade and credit networks. The international trade network has been previously studied by several authors (e.g., \cite{Fieler11ECTA_COMTRADEdata}, \cite{Eaton02ECTA_trade}), but to the best of our knowledge, this is the first attempt to model the dynamics of two networks jointly.
The bilateral trade data come from the COMTRADE database, whereas the data on bilateral outstanding capital come from the Bank of International Settlements database.
Our sample of yearly observations for 10 countries runs from 2003 to 2016. At at each time $t$, the $3$-order tensor $\mathcal{Y}_t$ has size $(10,10,2)$ and represents a $2$-layer node-aligned network (or multiplex) with $10$ vertices (countries), where each edge is given by a bilateral trade flow or financial stock.
%The entry $(i,j,1,t)$ of $\mathcal{Y}_t$ reports the total exports of country $i$ vis-\`a-vis country $j$, in year $t$, whereas entry $(i,j,2,t)$ contains the total outstanding credit from country $i$ towards country $j$, in year $t$. The series $\lbrace \mathcal{Y}_t \rbrace_t$, $t=1,\ldots,T$, has been standardized (over the temporal dimension).
See the supplement for data description

We estimate the tensor autoregressive model in eq. \eqref{eq:model_final}, using the prior structure described in section \ref{sec:bayesian_inference}, running the Gibbs sampler for $N=100,000$ iterations after $30,000$ burn-in iterations. We retain every second draw for posterior inference.

\begin{figure}[t!h]
\setlength{\abovecaptionskip}{6pt}
\captionsetup{width=0.95\linewidth}
\centering	% ART_Bmatrix.eps		ART_Bmatrix_eigens.eps
\includegraphics[trim= 0mm 0mm 0mm 0mm,clip,height= 5.5cm,width= 6.2cm]{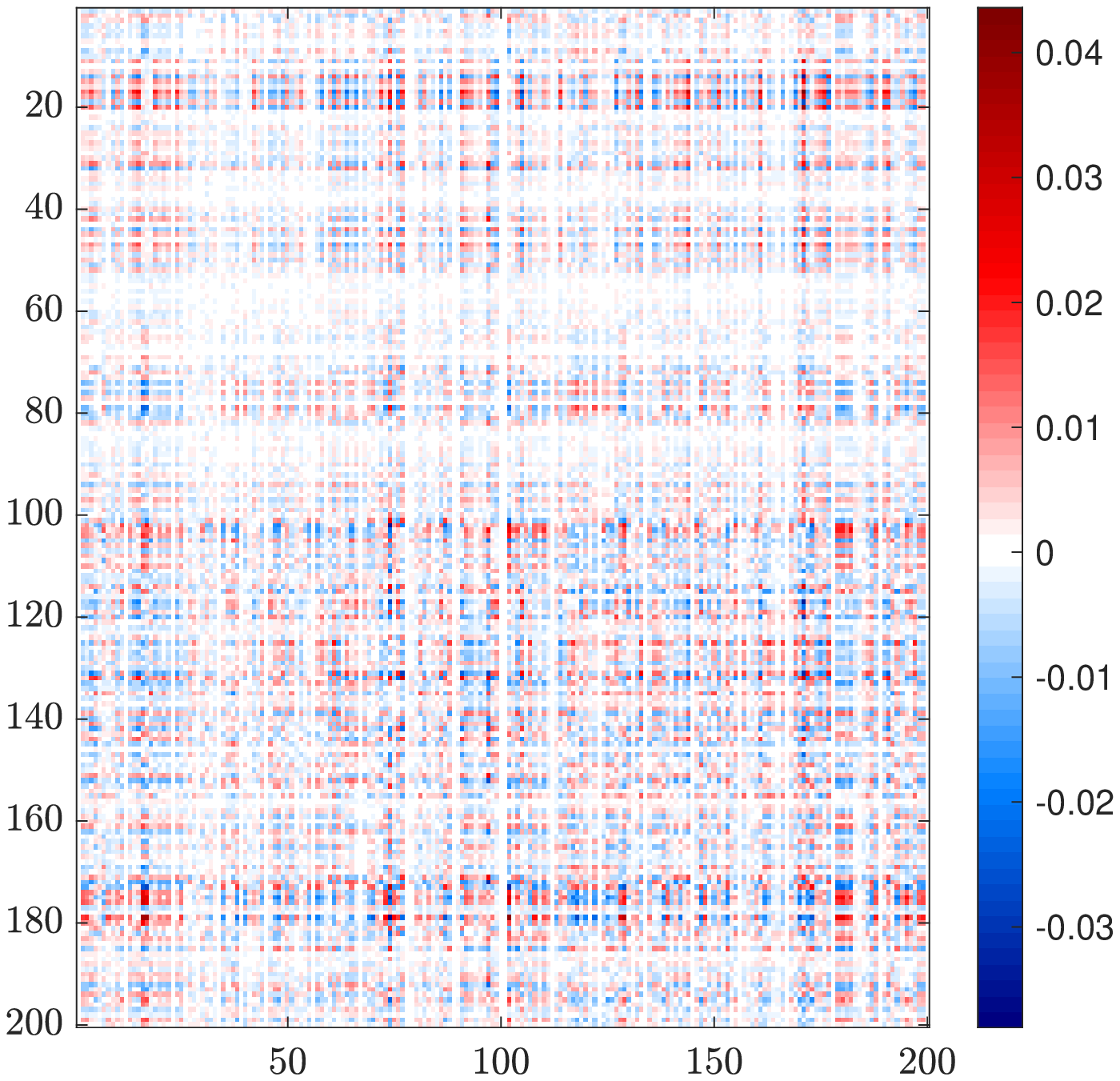} \hspace*{2ex}
\includegraphics[trim= 0mm 0mm 0mm 0mm,clip,height= 5.5cm,width= 6.0cm]{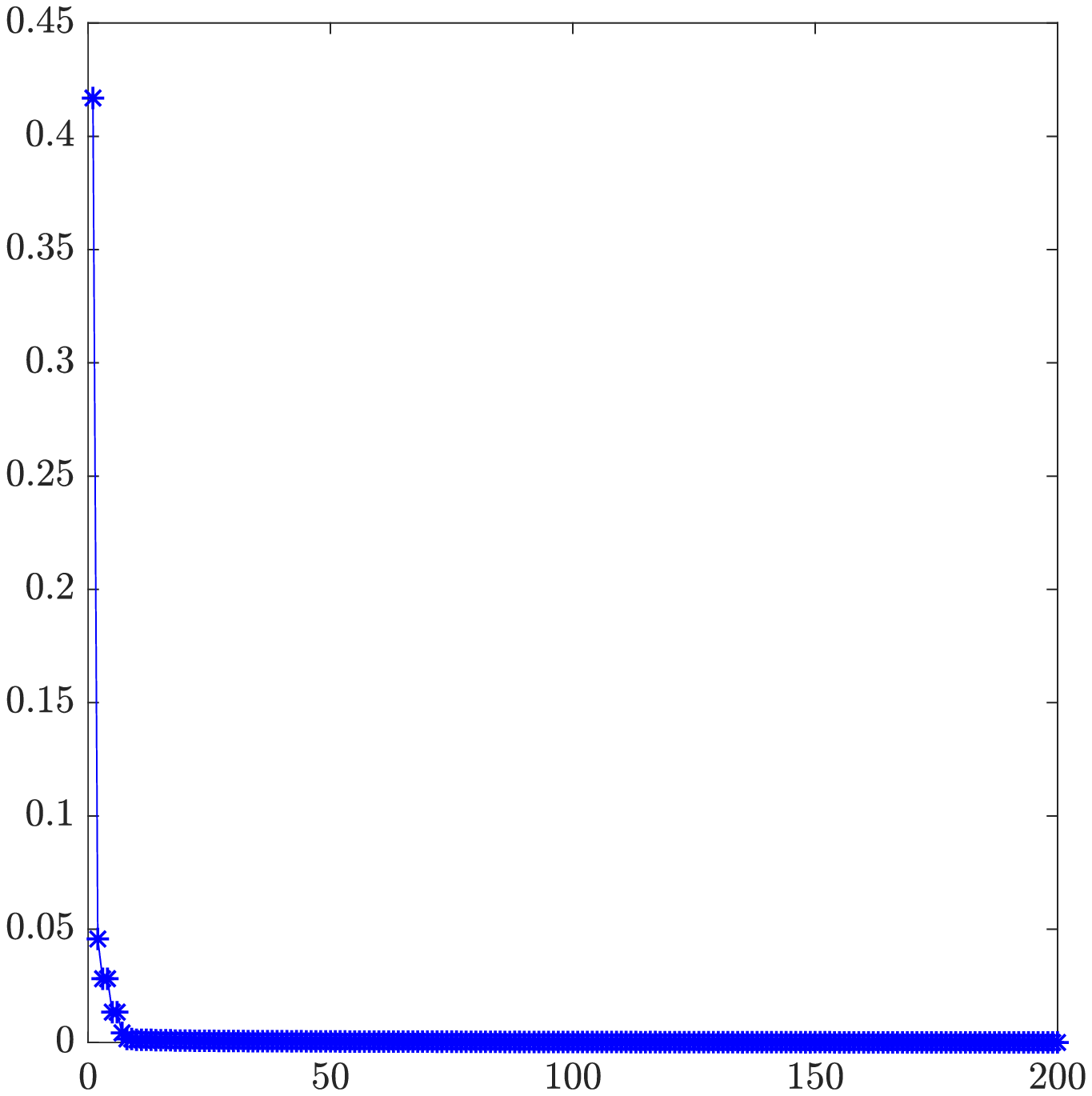}
\caption{\textit{Left:} mode-$4$ matricization of estimated coefficient tensor $\hat{B}_{(4)}$. \textit{Right:} log-spectrum of $\hat{B}_{(4)}$, decreasing order.}
\label{fig:application_ART_tensor}
\end{figure}

\begin{figure}[t!h]
\setlength{\abovecaptionskip}{6pt}
\captionsetup{width=0.95\linewidth}
\centering
\includegraphics[trim= 0mm 0mm 0mm 0mm,clip,height= 4.0cm,width= 4.2cm]{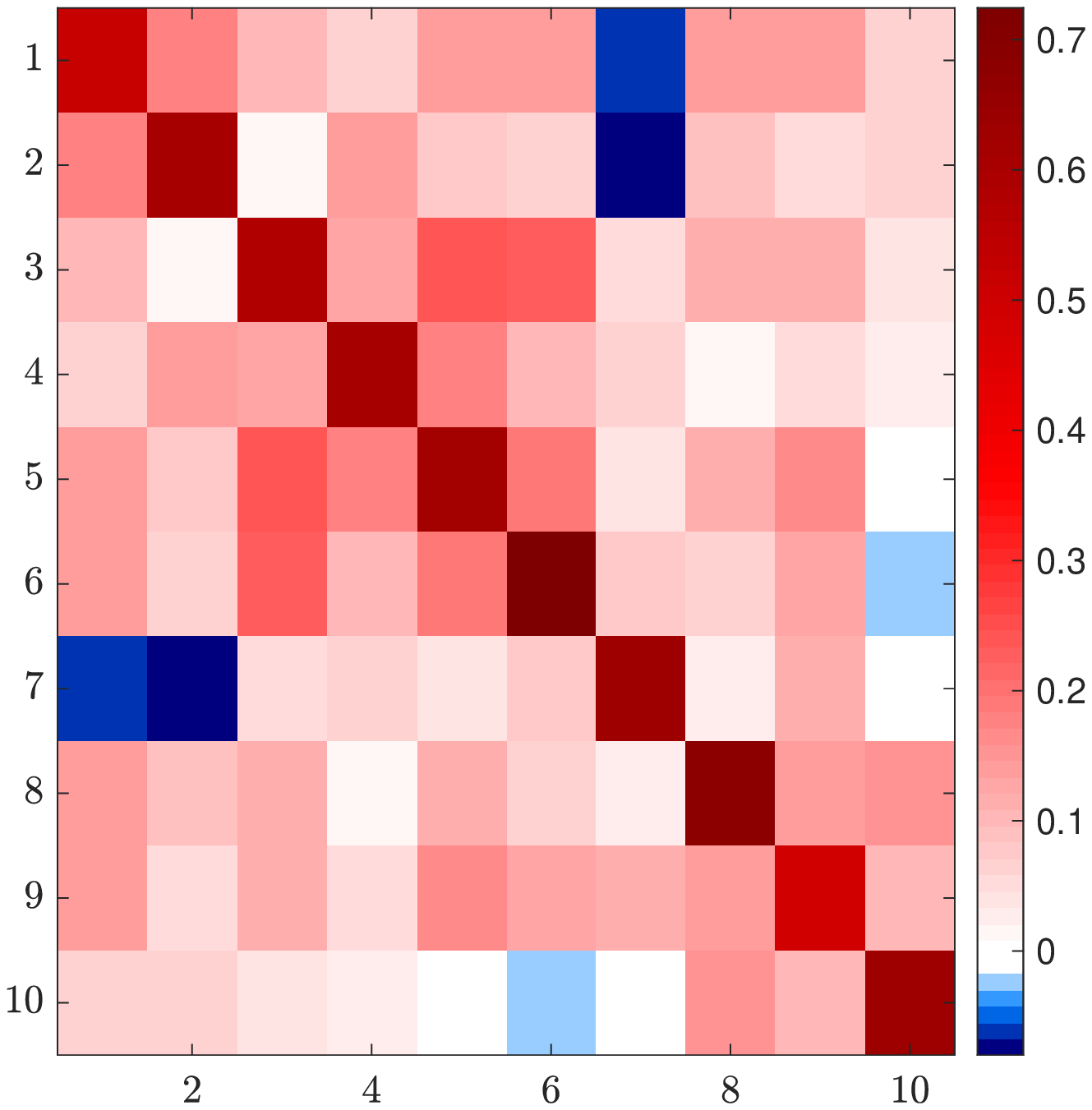} \hspace*{1ex}
\includegraphics[trim= 0mm 0mm 0mm 0mm,clip,height= 4.0cm,width= 4.2cm]{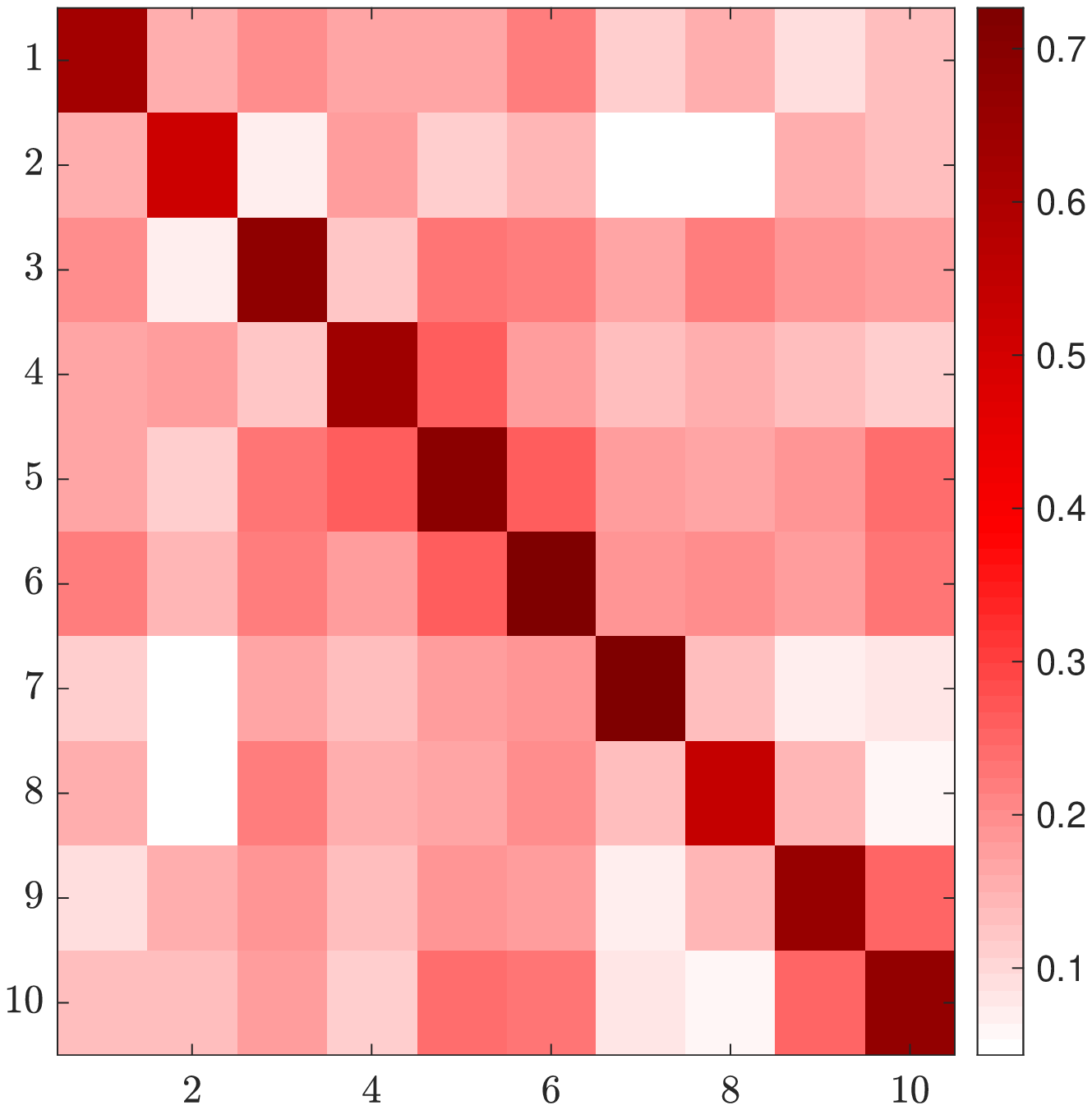} \hspace*{1ex}
\includegraphics[trim= 0mm 0mm 0mm 0mm,clip,height= 4.0cm,width= 4.2cm]{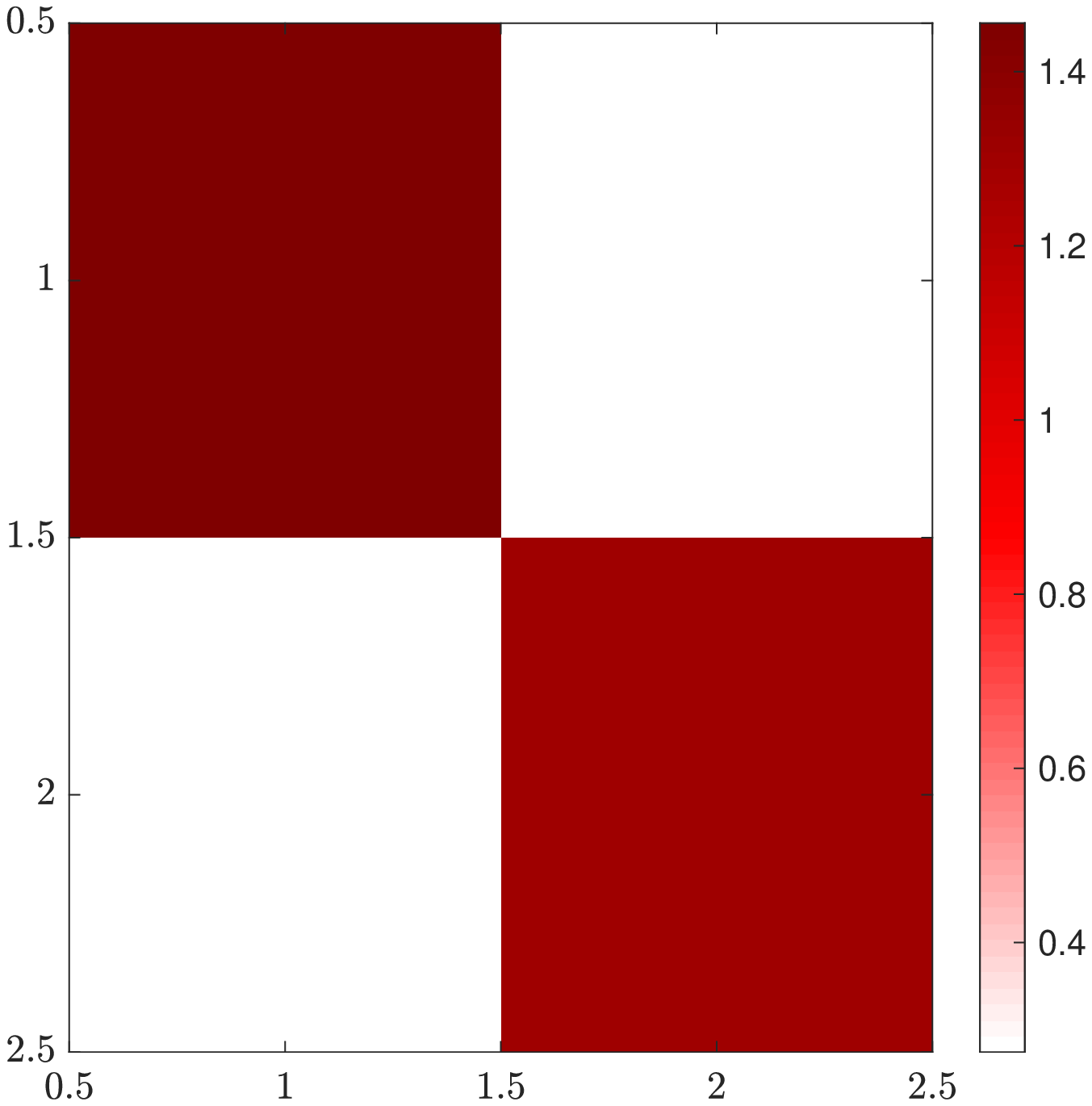}
\caption{Estimated covariance matrices: $\hat{\Sigma}_1$ (\textit{left}), $\hat{\Sigma}_2$ (\textit{center}), $\hat{\Sigma}_3$ (\textit{right}).}
\label{fig:application_ART_Sigmas}
\end{figure}

The mode-$4$ matricization of the estimated coefficient tensor, $\hat{B}_{(4)}$, is shown in the left panel of Fig. \ref{fig:application_ART_tensor}. The $(i,j)$-th entry of the matrix $\hat{B}_{(4)}$ reports the impact of the edge $j$ on edge $i$ (in vectorised form\footnote{For example, $j=21$ and $i=4$ corresponds to the coefficient of entry $\mathcal{Y}_{1,3,1,t-1}$ on $\mathcal{Y}_{4,1,1,t}$.}). The first $100$ rows/columns correspond to the edges in the first layer.
Hence, two rows of the matricized coefficient tensor are similar when two edges are affected by all the edges of the (lagged) network in a similar way, whereas two similar columns identify the situation where two edges impact the (next period) network in a similar way.
The overall distribution of the estimated entries of $\hat{B}_{(4)}$ is symmetric around zero and leptokurtic, as a consequence of the shrinkage to zero of the estimated coefficients. The right panel of Fig. \ref{fig:application_ART_tensor} shows the log-spectrum of $\hat{B}_{(4)}$. As all eigenvalues of $\hat{B}_{(4)}$ have modulus smaller than one, we conclude that the estimated ART($1$) model is stationary\footnote{It can be shown that the stationarity of the mode-$4$ matricised coefficient tensor implies stationarity of the ART($1$) process.}.
Fig. \ref{fig:application_ART_Sigmas} shows the estimated covariance matrices. In all cases, the highest values correspond to individual variances, while the estimated covariances are lower in magnitude and heterogeneous. We also find evidence of heterogeneity in the dependence structure, since $\Sigma_1$, which captures the covariance between rows (i.e., exporting and creditor countries), differs from $\Sigma_2$, which describes the covariance between columns (i.e., importing and debtor countries). With few exceptions, estimated covariances are positive.
%In particular, the covariance between exporting countries is higher on average than between importing countries.

%%%%%%%%%%%%% IMPULSE RESPONSE ANALYSIS %%%%%%%%%%%%%
% Change by $\delta \sum_i x_i$ proportional to the shares, means that
% $x_j \righarrow x_j + (\delta \sum_i x_i) \cdot (x_j / \sum_i x_i) = x_j + \delta x_j = x_j(1+\delta)$ hence $\mathbf{x} \rightarrow \mathbf{x}(1+\delta)$.
% So, the change (i.e., shock) is $\Delta = \delta \mathbf{x}$.
% If we take logarithms:  $\log(\mathbf{x}) \rightarrow \log(\mathbf{x}(1+\delta))$, thus the change (i.e., shock) is $\Delta = \log(1+\delta)$.
% Therefore, using the IRF we are observing the (decaying) path of $\log(1+\delta)$, which corresponds to $\delta \mathbf{x}$ in levels. How to recover the ``IRF on levels''?\\
% Call $k = \delta \mathbf{x}$ the shock in levels. Then $1+k/\mathbf{x} = 1+\delta$, hence $\log(1+k/\mathbf{x}) = \log(1+\delta)$. Call $y = \log(1+\delta)$, which is what we have (the decaying IRF). We can recover $k$ as $k = [\exp(y)-1] \cdot \mathbf{x}$.

After estimating the ART($1$) model \eqref{eq:model_final}, we may investigate shock propagation across the network computing generalised and orthogonalised impulse response functions presented in equations \eqref{eq:GIRF_new} and \eqref{eq:OIRF_new}, respectively.
Impulse responses allow us to analyze the propagation of shocks both across the network, within and across layers, and over time.
For illustration, we study the responses to a shock in all edges of  country, by applying  block Cholesky factorisation to $\Sigma$, in such a way that the shocked country contemporaneously affects all others and not vice-versa.\footnote{To save space, we do not report generalised IRFs, which are very similar to the ones presented.} Thus, the matrices $A$ and $C$ in eq. \eqref{eq:Sigma_block} reflect contemporaneous correlations across transactions of the shock-originating country and with transactions of all other countries, respectively. For expositional convenience, we report only statistically significant responses.

In the first analysis we consider a negative $1\%$ shock to US trade imports\footnote{That is, we allocate the shock across import originating countries to match import shares as in the last period of the sample.}. The results of the block Cholesky IRF at horizon $1$ are given in Fig. \ref{fig:application_ART_IRF_USimp}. We report the impact on the whole network (panel (a)) and, for illustrative purposes, the impact on Germany's transactions.
%\footnote{We report the detailed IRF for Germany since we consider it the most important economy of (thus, the representative country for) the EU.}.
The main findings follows.

\textit{Global effect} on the network.
The negative shock to US imports has an effect on both layers (trade and financial) of the network.
There is evidence of heterogeneous responses across countries and country-specific transactions. On average, trade flows exhibit a slight expansion in response to the shock. Switzerland is the most positively affected, both in terms of exports and imports, and  trade imports of the US show (on average) a reverted positive response one period after the shock. This reflects an oscillating impulse response. The overall average effect on the financial layer is negative, similar in magnitude to the effect on the trade layer.
More specifically, we observe that Denmark's and Sweden's exports to Switzerland, Germany and France show a contraction, whereas the effect on US', Japan's and Ireland's exports to these countries is positive. We may interpret these effects as substitution effects: The decreasing share of Denmark's and Sweden's exports to Switzerland, Germany and France is offset by an increase of US, Japanese and Irish exports. In conclusion, the dynamic model can be used for predicting possible trade creation and diversion effects (e.g., see \cite{Bikker10trade_substitution_effect}).

\textit{Local effect on Germany}.
In panel (b) of Fig. \ref{fig:application_ART_IRF_USimp} we report the response of Germany's transactions to the negative shock in US imports. The effects on imports are mixed: while Germany's imports from most other EU countries increase, imports from Sweden and Denmark decrease. Likewise, Germany's exports show heterogeneous responses, whereby exports to Switzerland react strongest (positively).
The shock of US imports does not have a significant impact on Germany's outstanding credit against most countries (except Switzerland and Japan). On the other hand, the reactions of Germany's outstanding debt reflect those on trade imports.

\textit{Local effect} on other countries.
We observe that the most affected trade transactions are those of Denmark, Japan, Ireland, Sweden and US (as exporters) vis-\'a-vis Switzerland and France (as importers).
The financial layer mirrors these effects with opposite sign, while the magnitudes are comparable.
Outstanding credit of Ireland and Japan to Switzerland, Germany and France decrease at horizon $1$. By contrast, Denmark's outstanding credit to these countries increases. Note that outstanding debt of US vis-\'a-vis almost all countries decreases after the shock.
Overall, responses to a shock on US imports at horizon $1$ are heterogeneous in sign but rather low in magnitude, whereas at horizon $2$ (plot not reported) the propagation of the shock has vanished. We interpret this as a sign of fast (and monotone) decay of the IRF.

\begin{figure}[h!t]
\setlength{\abovecaptionskip}{2pt}
\captionsetup{width=.90\linewidth}
\centering
\begin{tabular}{cc}
\multicolumn{2}{c}{(a) Network IRF at $h=1$} \\
\begin{rotate}{90} \hspace*{20pt} {\scriptsize Financial (layer 2) \hspace{40pt} Trade (layer 1)} \end{rotate} &
%\begin{rotate}{90} \hspace*{35pt} \large{$\substack{\text{Financial\hspace{80pt}Trade}\\[5pt] \hspace{5pt}\text{(layer 2)\hspace{75pt}(layer 1)}}$} \end{rotate} & 
%\includegraphics[trim= 0mm 0mm 0mm 0mm,clip,height= 8.6cm, width= 5.5cm]{IRF_USimp.eps}
\includegraphics[trim= 0mm 0mm 0mm 0mm,clip,height= 8.0cm,width= 4.5cm]{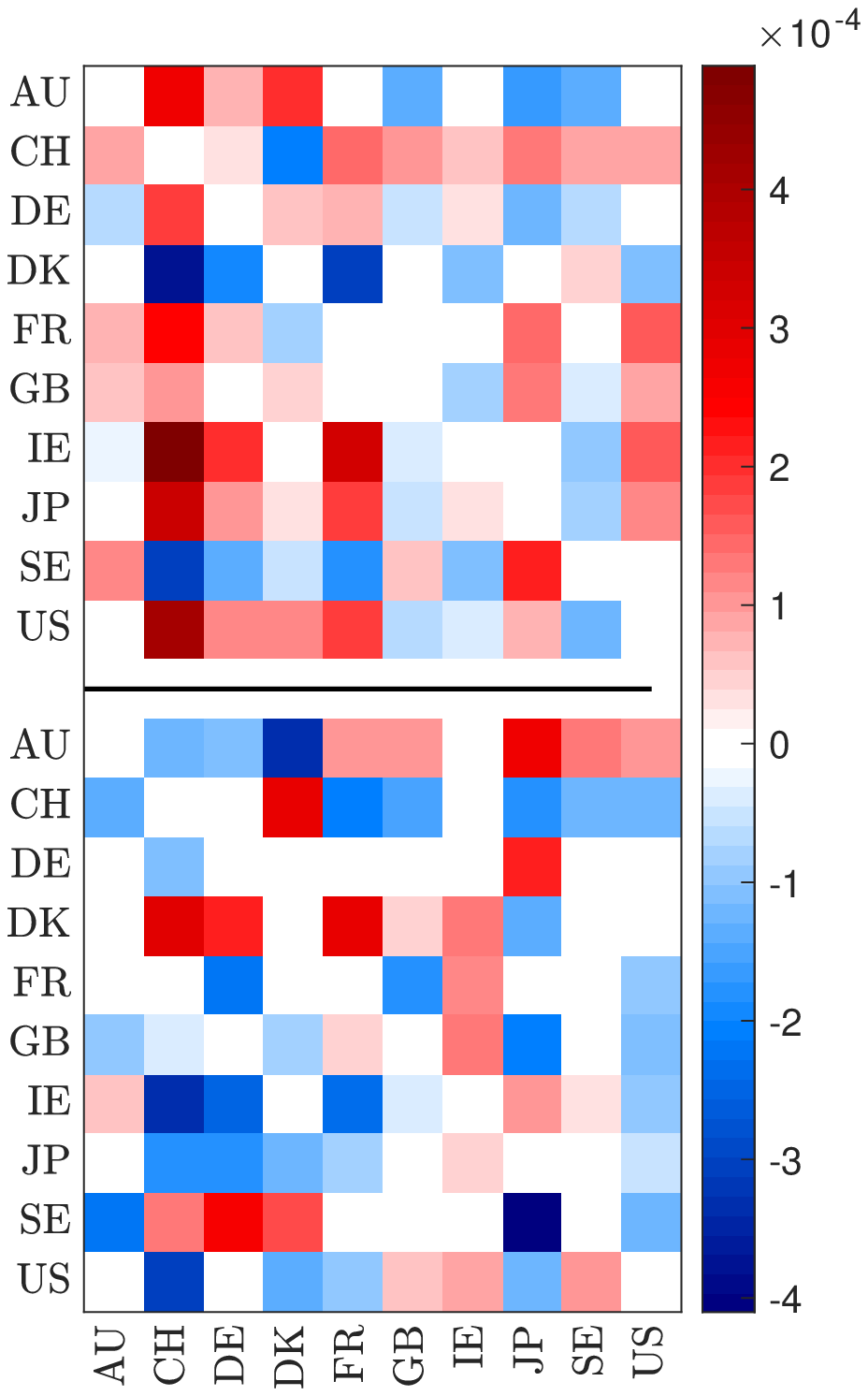}
\end{tabular}
%\hspace*{4ex}
\begin{tabular}{cc}
% & \textcolor{white}{$h=0$} \\
\multicolumn{2}{c}{(b) IRF for Germany's edges at $h=1$} \\[10pt]
\begin{rotate}{90} \hspace*{-95pt} {\scriptsize Financial (layer 2) \hspace{45pt} Trade (layer 1)} \end{rotate} &
%\begin{rotate}{90} \hspace*{-75pt} \normalsize{$\substack{\text{Financial\hspace{80pt}Trade}\\[5pt] \hspace{5pt}\text{(layer 2)\hspace{75pt}(layer 1)}}$} \end{rotate} &
\begin{tabular}{ccc}
{\footnotesize exports} & & {\footnotesize imports} \\
\includegraphics[trim= 0mm 30mm 0mm 30mm,clip,height= 3.0cm,width= 2.5cm]{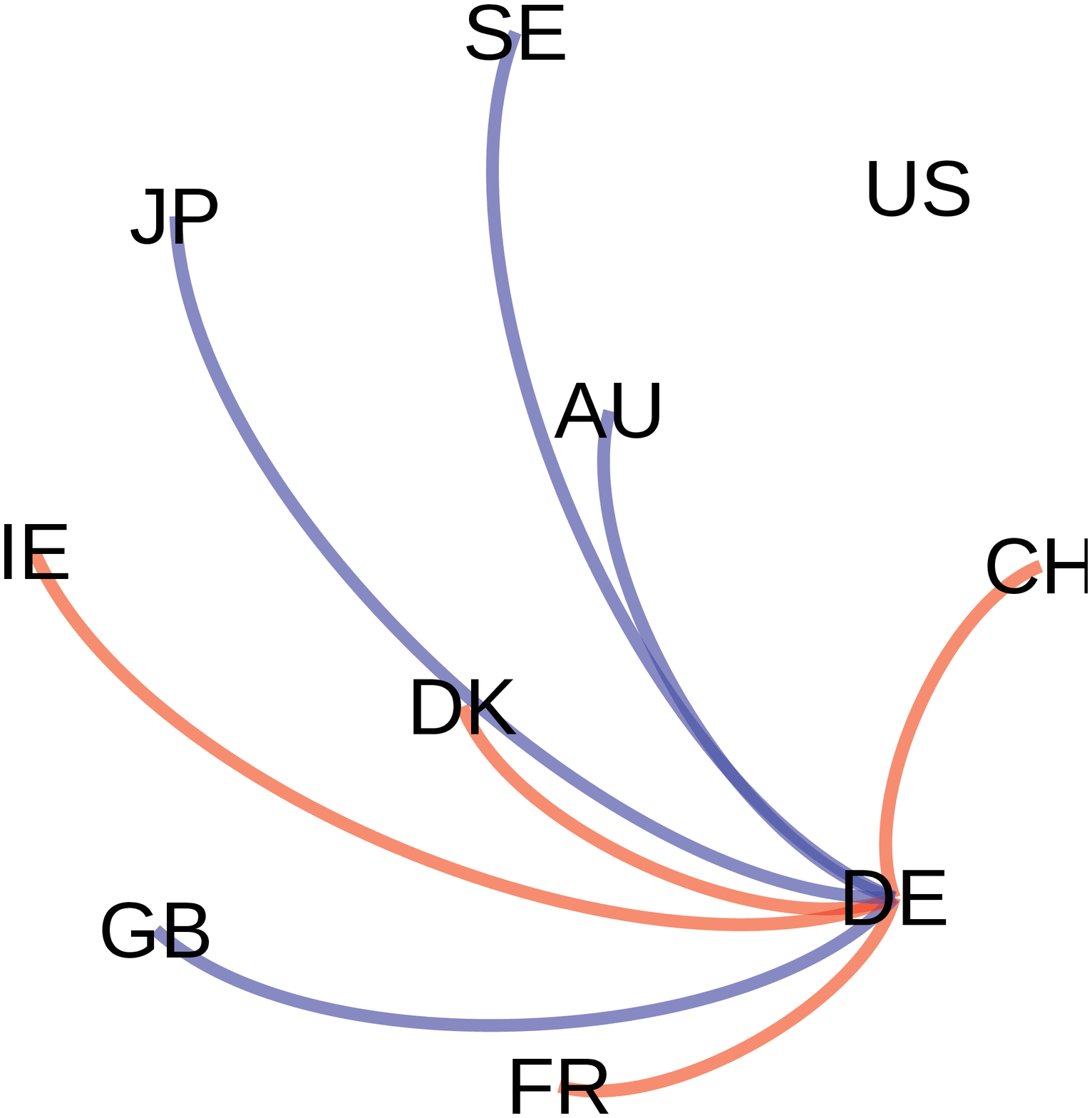} & & \includegraphics[trim= 0mm 30mm 0mm 30mm,clip,height= 3.0cm,width= 2.5cm]{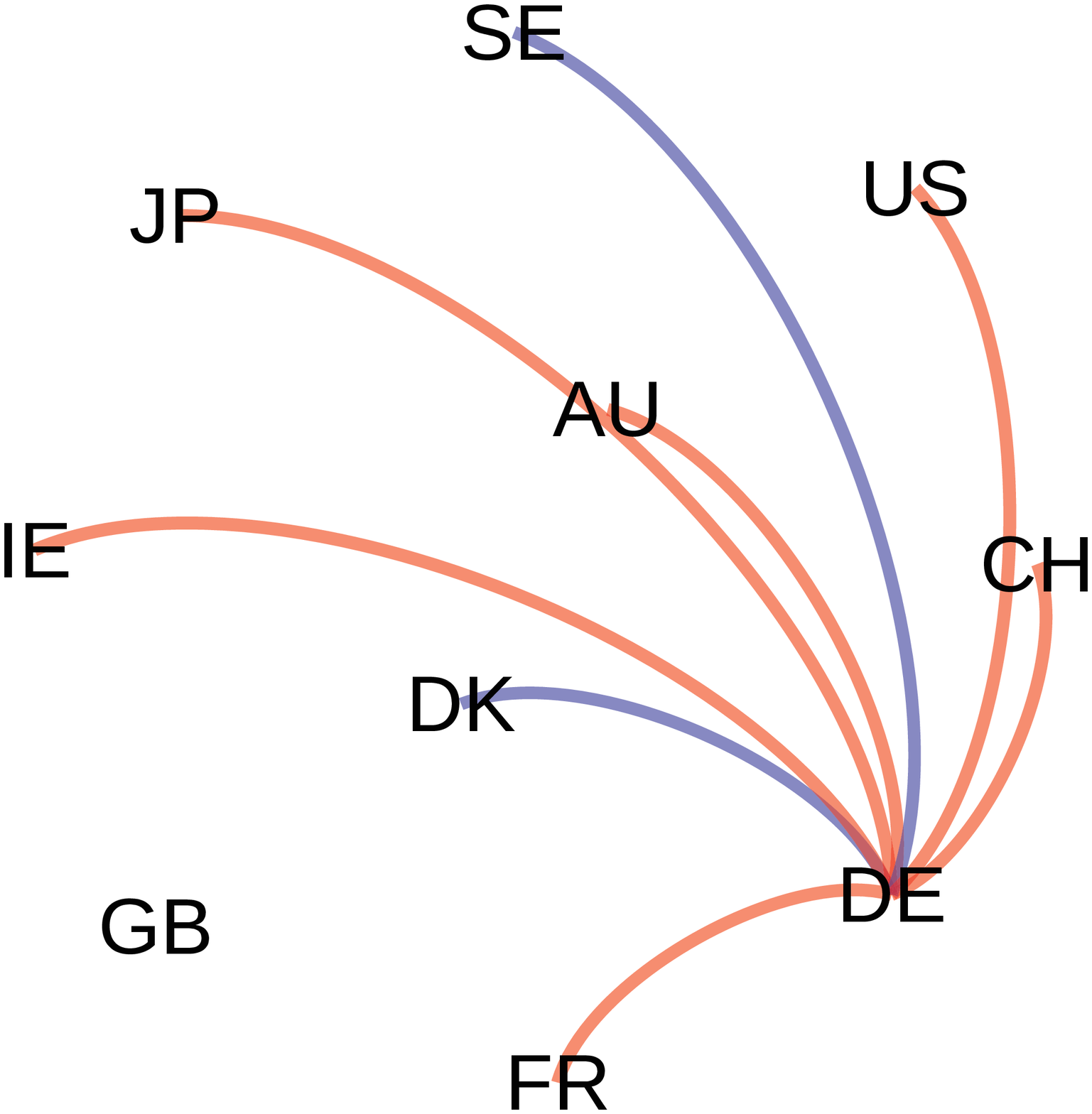} \\[-0ex]
{\footnotesize outstanding credit} & & {\footnotesize outstanding debt} \\
\includegraphics[trim= 0mm 30mm 0mm 30mm,clip,height= 3.0cm,width= 2.5cm]{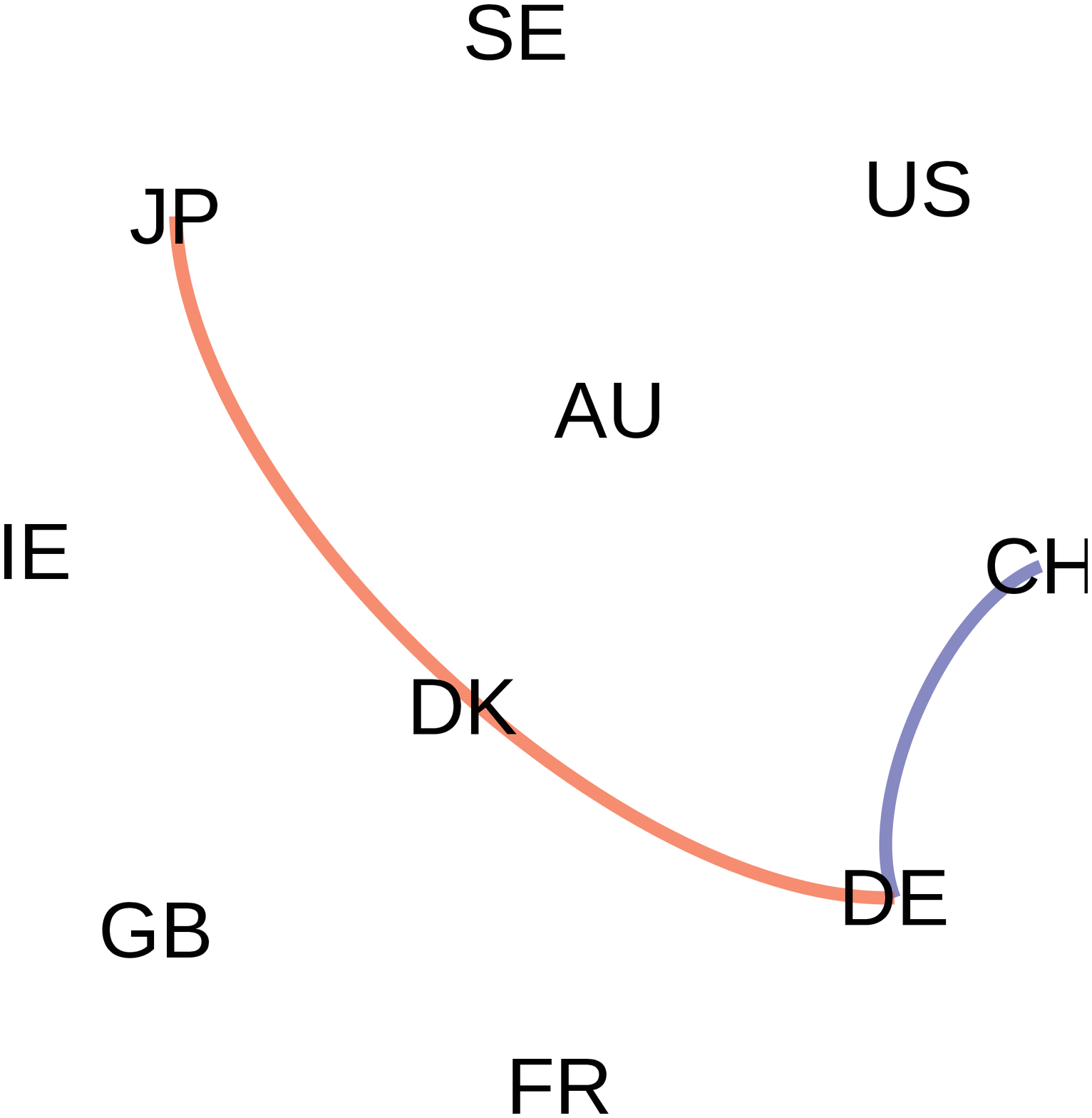} & & \includegraphics[trim= 0mm 30mm 0mm 30mm,clip,height= 3.0cm,width= 2.5cm]{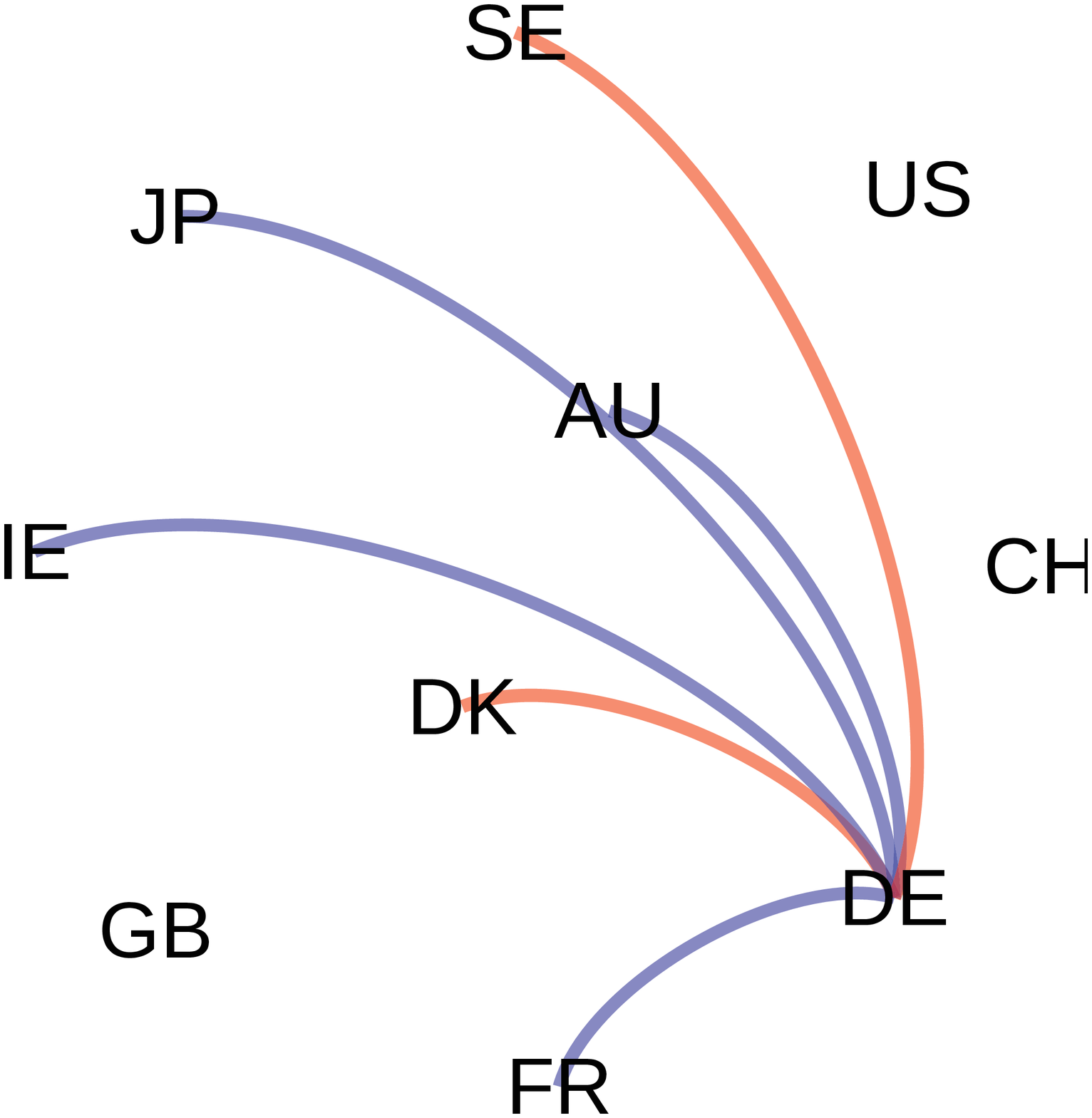} \\
\end{tabular}
\end{tabular}
\captionsetup{type=figure}
\captionof{figure}{Shock to US trade imports by -1\%. IRF at horizon $h=1$ for all (\textit{panel a}) and Germany (\textit{panel b}) financial and trade transactions. In each plot negative coefficients are in blue and positive in red.}
\label{fig:application_ART_IRF_USimp}
\end{figure}

%%%%%%%%%%%%%%%%%%%%%%%%%%%%%%%%%%%%%%%%%%%%%%%%%%%%%%%%%%%
%%%%%%%%%%%%%%% Shock to UK capital imports %%%%%%%%%%%%%%%
\begin{figure}[h!t]
\setlength{\abovecaptionskip}{2pt}
\captionsetup{width=.90\linewidth}
\centering
\begin{tabular}{cc}
\multicolumn{2}{c}{(a) Network IRF at $h=1$} \\
\begin{rotate}{90} \hspace*{20pt} {\scriptsize Financial (layer 2) \hspace{40pt} Trade (layer 1)} \end{rotate} &
%\begin{rotate}{90} \hspace*{35pt} \normalsize{$\substack{\text{Financial\hspace{80pt}Trade}\\[5pt] \hspace{5pt}\text{(layer 2)\hspace{75pt}(layer 1)}}$} \end{rotate} & 
%\includegraphics[trim= 0mm 0mm 0mm 0mm,clip,height= 8.6cm, width= 5.5cm]{IRF_GBfIm1.eps}
\includegraphics[trim= 0mm 0mm 0mm 0mm,clip,height= 8.0cm, width= 4.5cm]{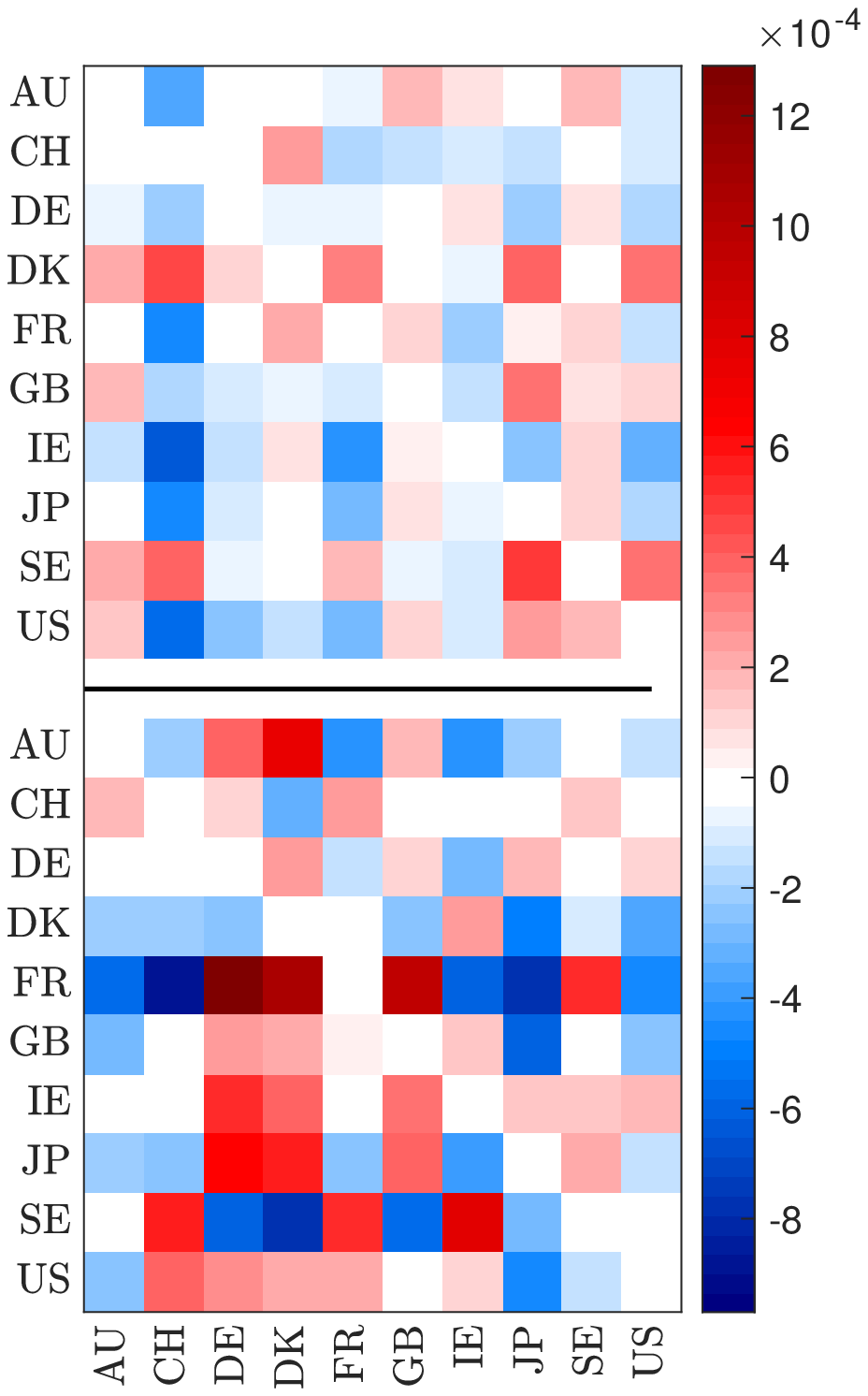}
\end{tabular}
%\hspace*{1ex}
\begin{tabular}{cc}
% & \textcolor{white}{$h=0$} \\
\multicolumn{2}{c}{(b) IRF for Germany's edges at $h=1$} \\[10pt]
\begin{rotate}{90} \hspace*{-95pt} {\scriptsize Financial (layer 2) \hspace{45pt} Trade (layer 1)} \end{rotate} &
%\begin{rotate}{90} \hspace*{-75pt} \normalsize{$\substack{\text{Financial\hspace{80pt}Trade}\\[5pt] \hspace{5pt}\text{(layer 2)\hspace{75pt}(layer 1)}}$} \end{rotate} &
\begin{tabular}{ccc}
{\footnotesize exports} & & {\footnotesize outstanding credit} \\
\includegraphics[trim= 0mm 30mm 0mm 30mm,clip,height= 3.0cm, width= 2.5cm]{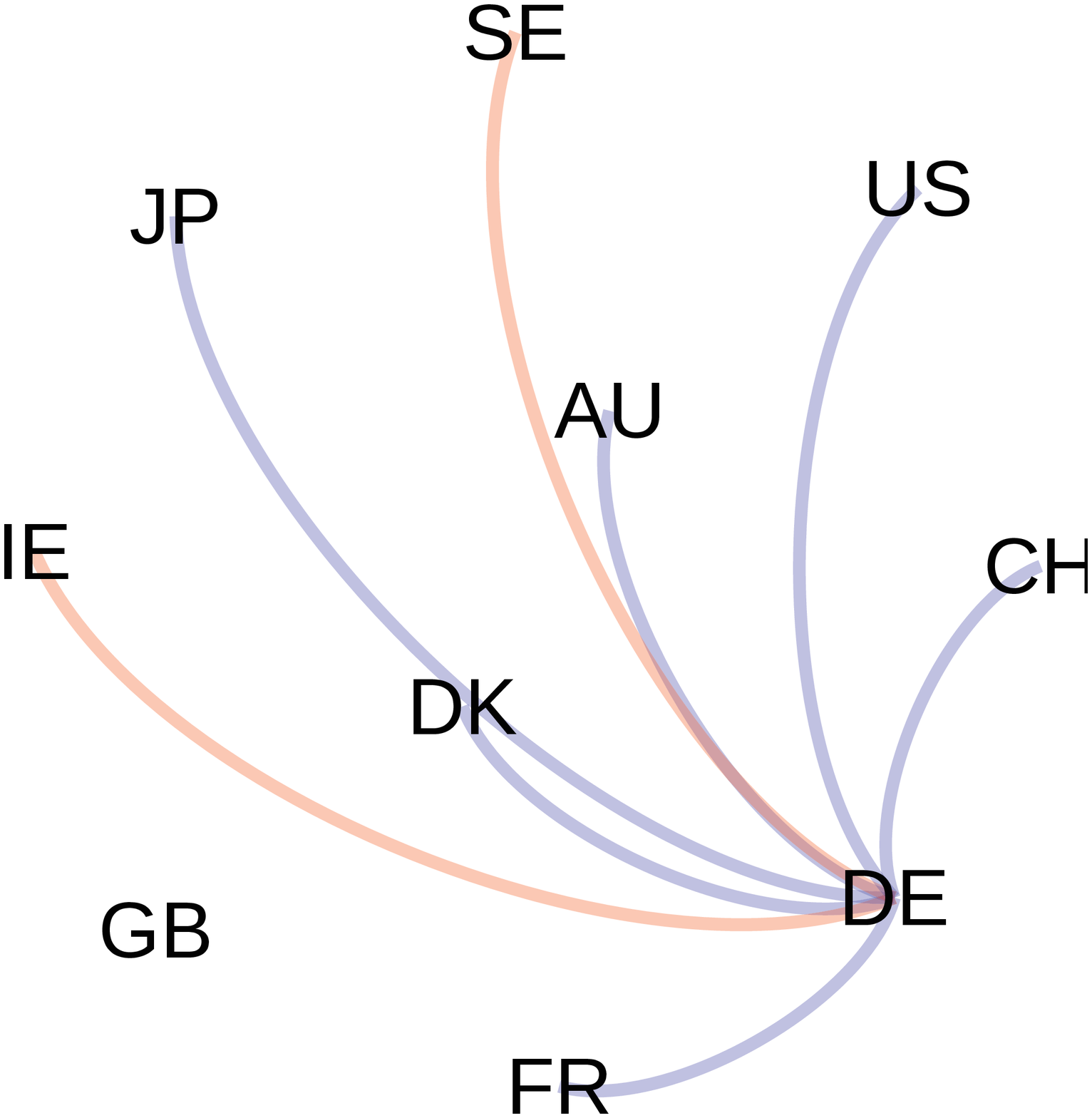} & & \includegraphics[trim= 0mm 30mm 0mm 30mm,clip,height= 3.0cm, width= 2.5cm]{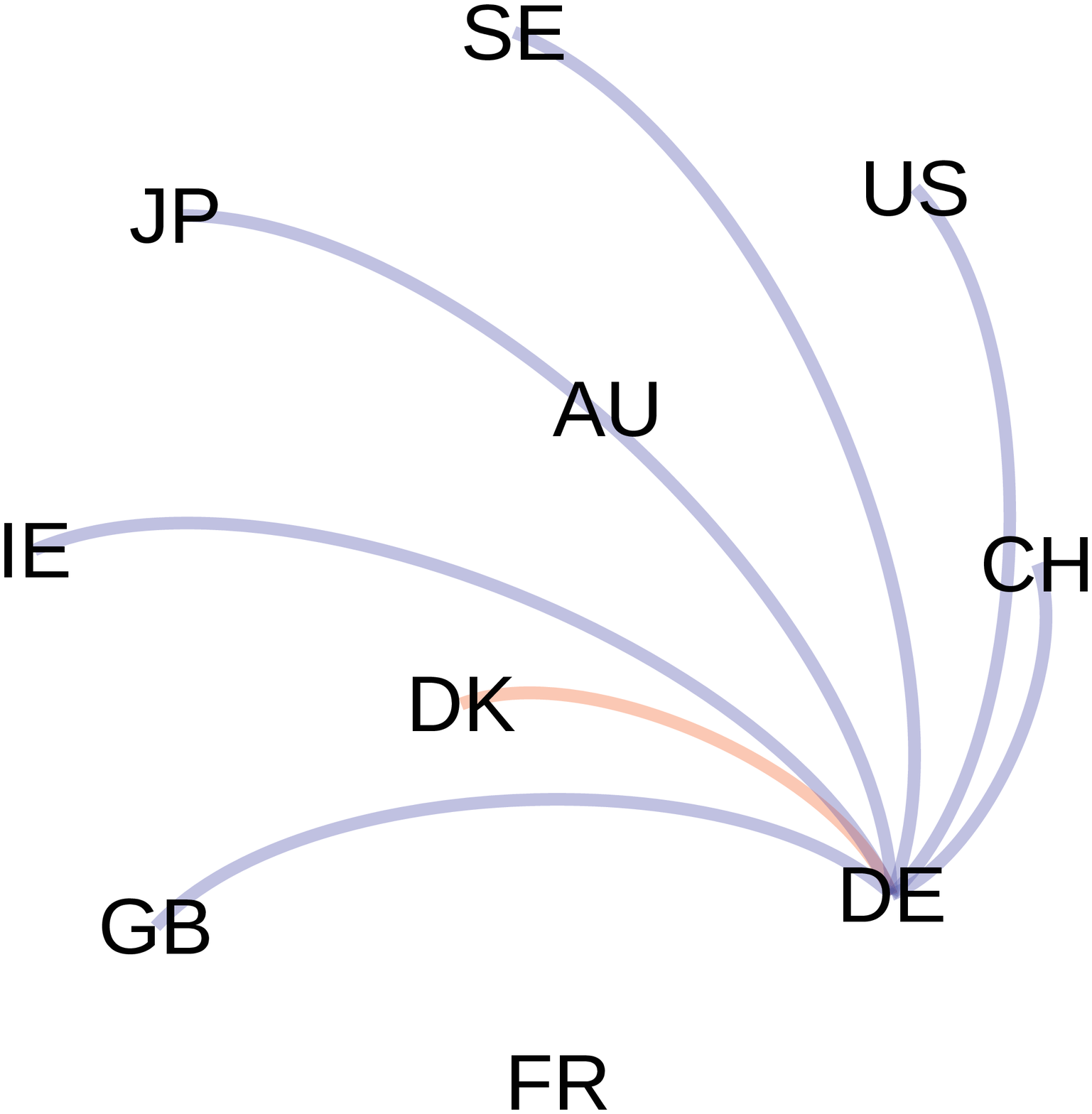} \\[-0ex]
{\footnotesize outflows} & & {\footnotesize outstanding debt} \\
\includegraphics[trim= 0mm 30mm 0mm 30mm,clip,height= 3.0cm, width= 2.5cm]{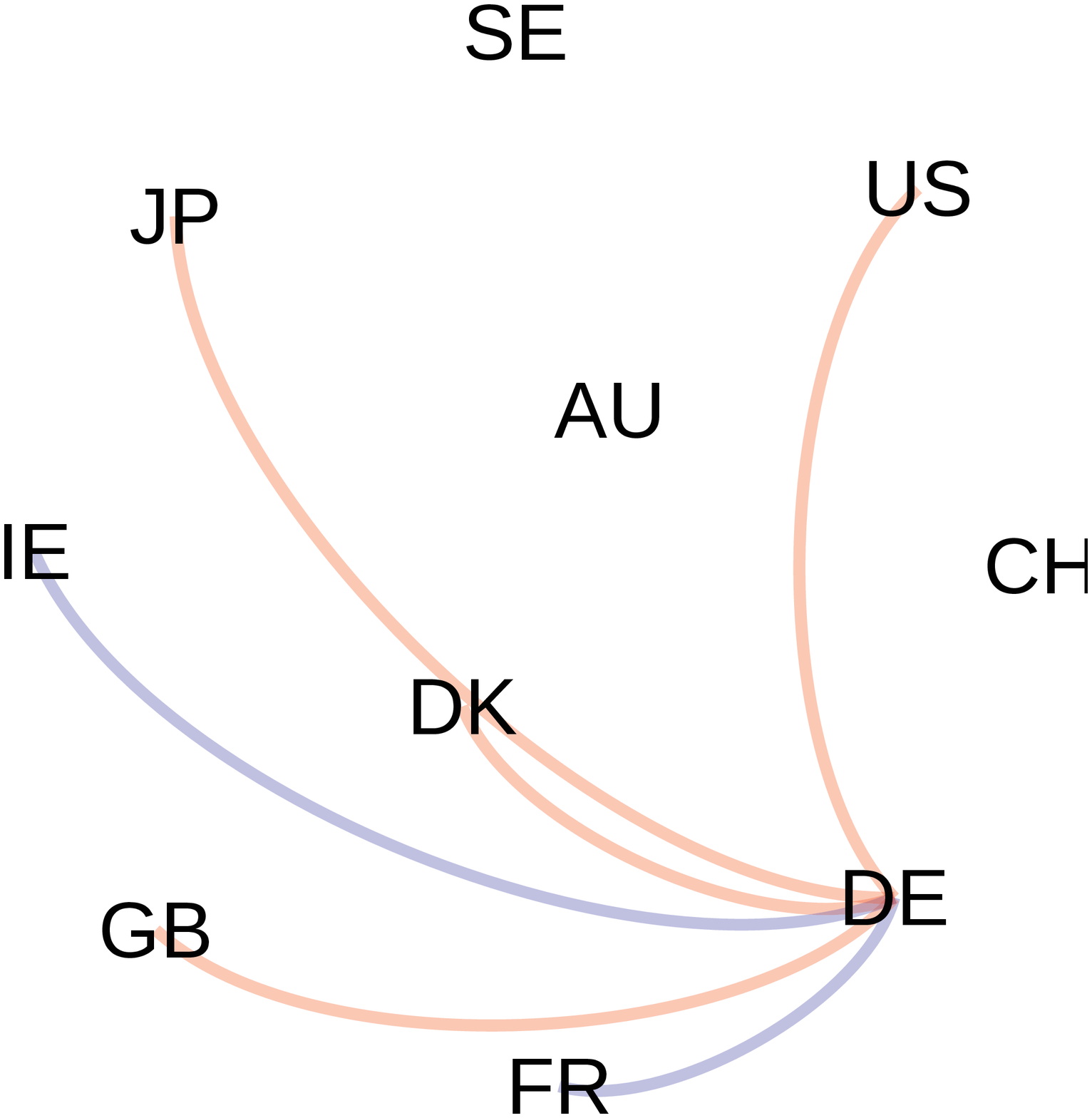} & & \includegraphics[trim= 0mm 30mm 0mm 30mm,clip,height= 3.0cm, width= 2.5cm]{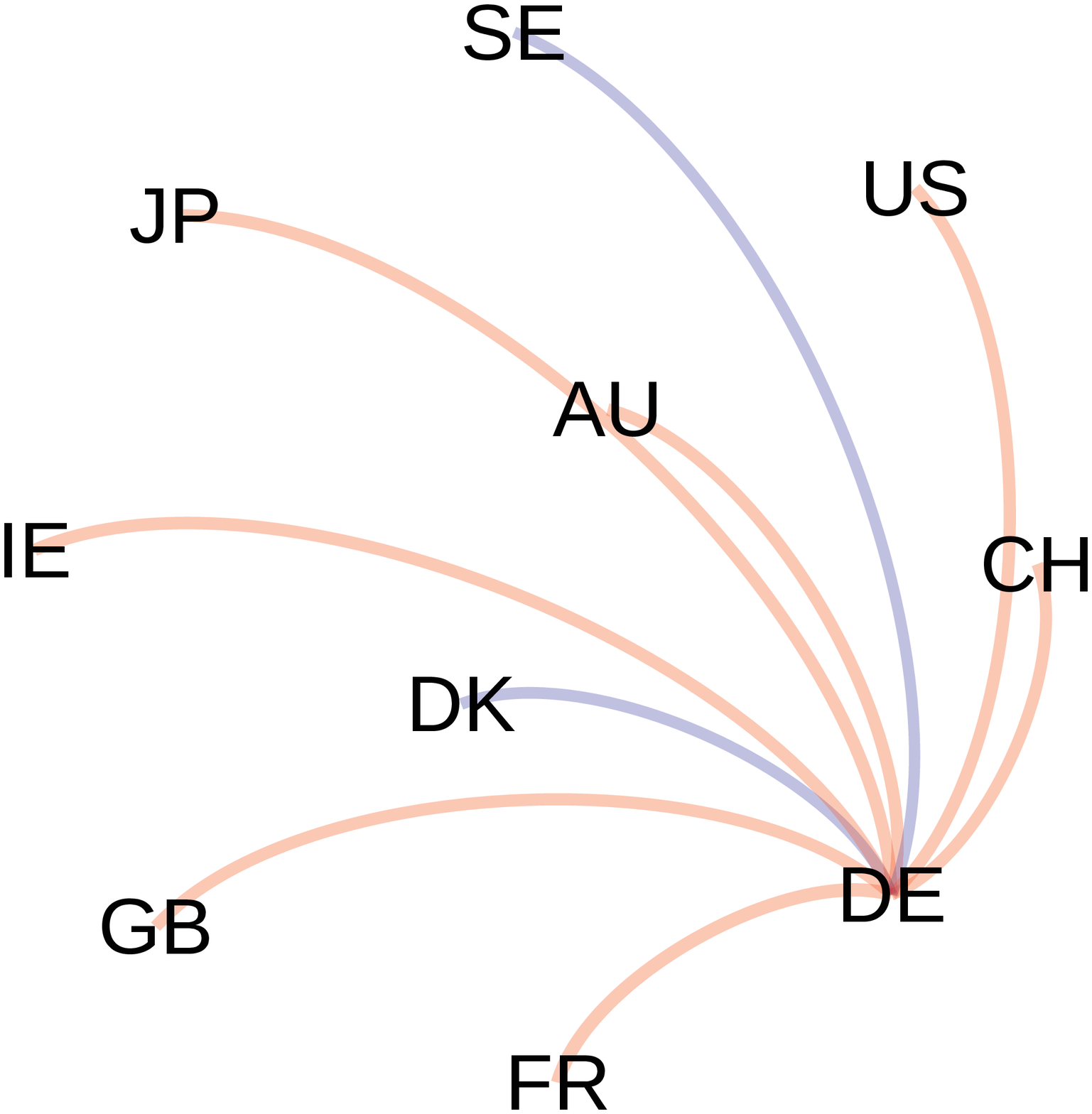} \\
\end{tabular}
\end{tabular}

\begin{tabular}{cc}
\multicolumn{2}{c}{(c) Network IRF at $h=2$} \\
\begin{rotate}{90} \hspace*{20pt} {\scriptsize Financial (layer 2) \hspace{40pt} Trade (layer 1)} \end{rotate} &
%\begin{rotate}{90} \hspace*{35pt} \normalsize{$\substack{\text{Financial\hspace{80pt}Trade}\\[5pt] \hspace{5pt}\text{(layer 2)\hspace{75pt}(layer 1)}}$} \end{rotate} & 
%\includegraphics[trim= 0mm 0mm 0mm 0mm,clip,height= 8.6cm, width= 5.5cm]{IRF_GBfIm2.eps}
\includegraphics[trim= 0mm 0mm 0mm 0mm,clip,height= 8.0cm, width= 4.5cm]{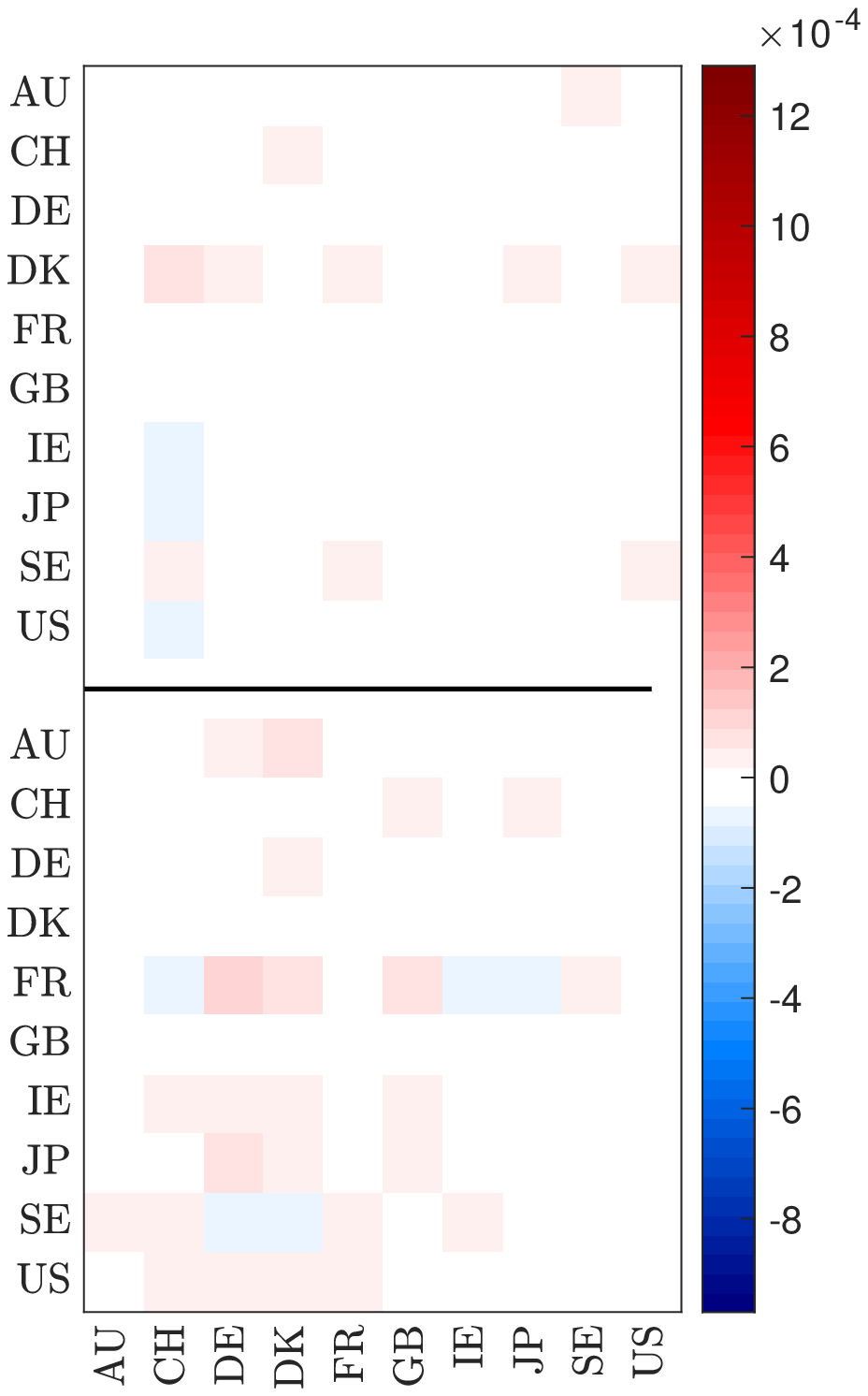}
\end{tabular}
%\hspace*{4ex}
\begin{tabular}{cc}
% & \textcolor{white}{$h=0$} \\
\multicolumn{2}{c}{(d) IRF for Germany's edges at $h=2$} \\[10pt]
\begin{rotate}{90} \hspace*{-95pt} {\scriptsize Financial (layer 2) \hspace{45pt} Trade (layer 1)} \end{rotate} &
%\begin{rotate}{90} \hspace*{-75pt} \normalsize{$\substack{\text{Financial\hspace{80pt}Trade}\\[5pt] \hspace{5pt}\text{(layer 2)\hspace{75pt}(layer 1)}}$} \end{rotate} &
\begin{tabular}{ccc}
{\footnotesize exports} & & {\footnotesize outstanding credit} \\
\includegraphics[trim= 0mm 30mm 0mm 30mm,clip,height= 3.0cm, width= 2.5cm]{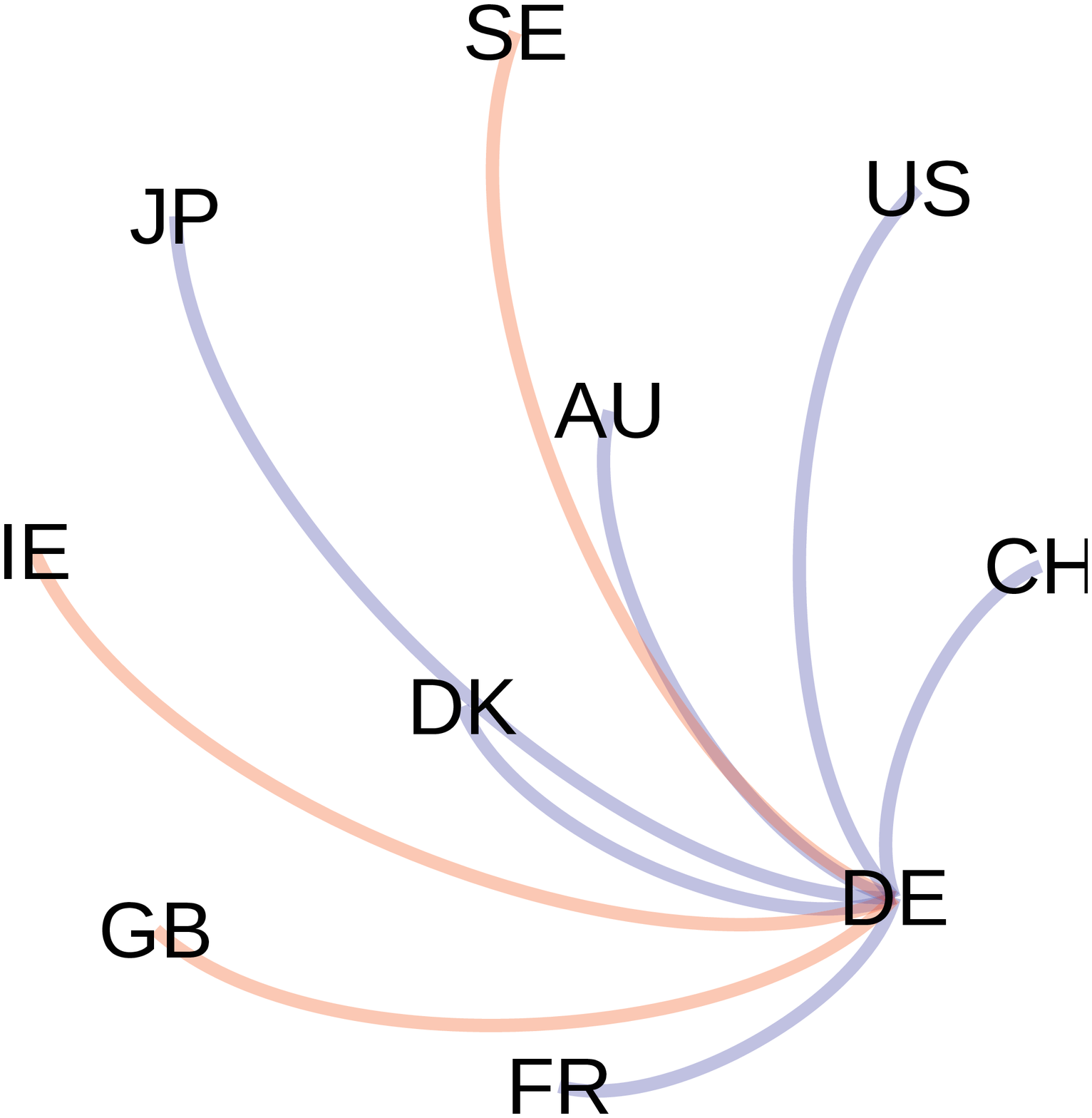} & & \includegraphics[trim= 0mm 30mm 0mm 30mm,clip,height= 3.0cm, width= 2.5cm]{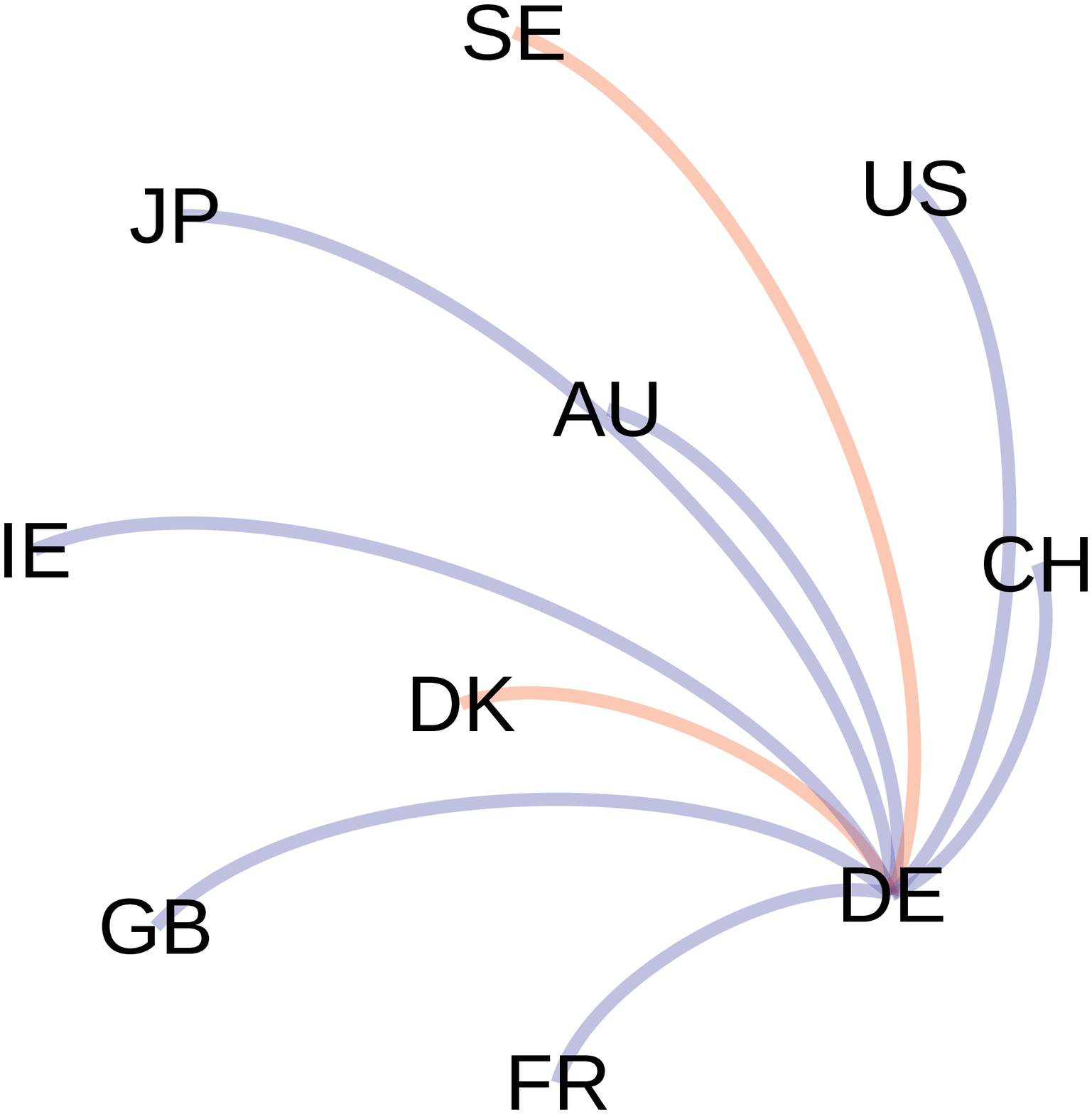} \\[-0ex]
{\footnotesize outflows} & & {\footnotesize outstanding debt} \\
\includegraphics[trim= 0mm 30mm 0mm 30mm,clip,height= 3.0cm, width= 2.5cm]{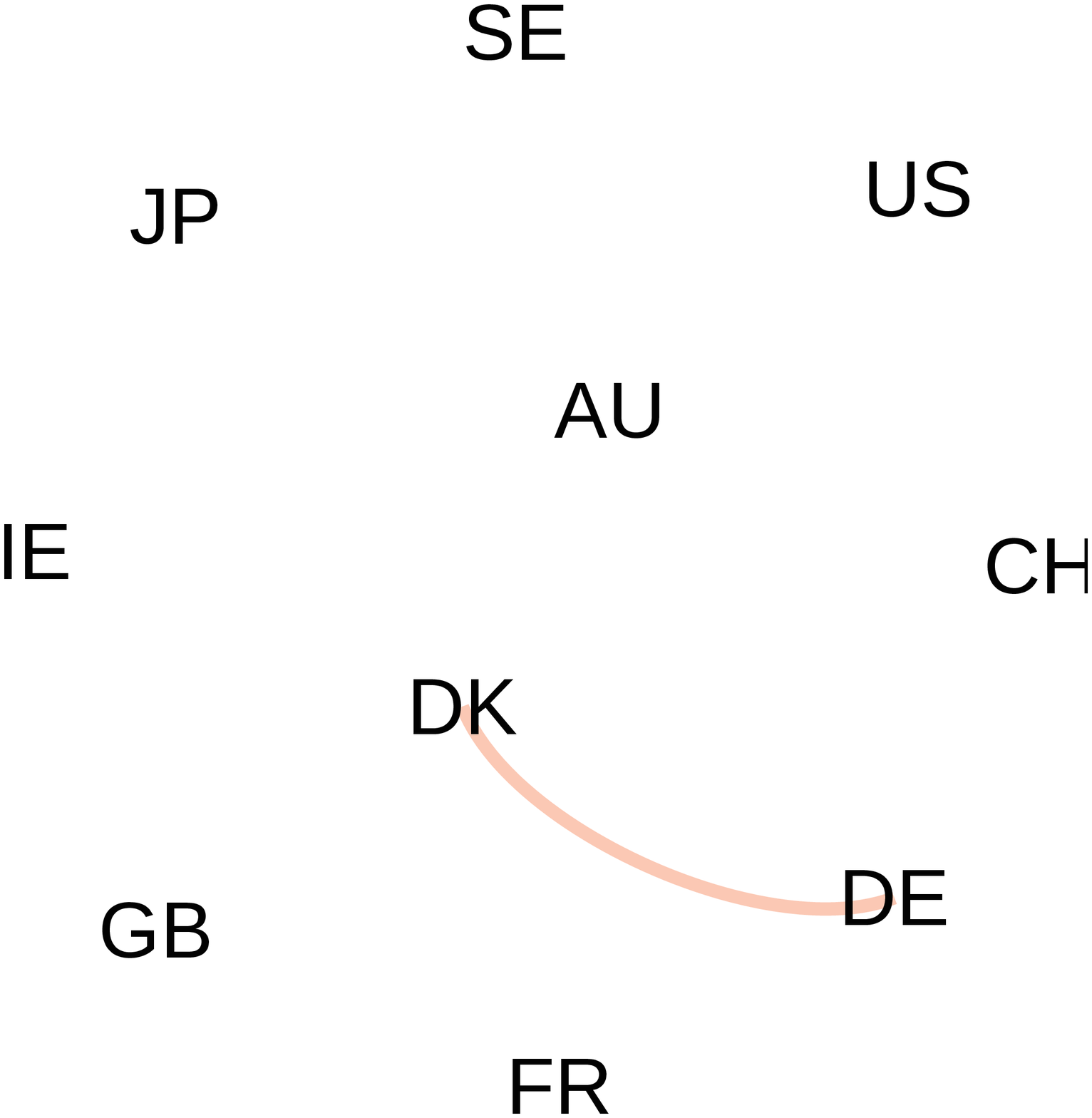} & & \includegraphics[trim= 0mm 30mm 0mm 30mm,clip,height= 3.0cm, width= 2.5cm]{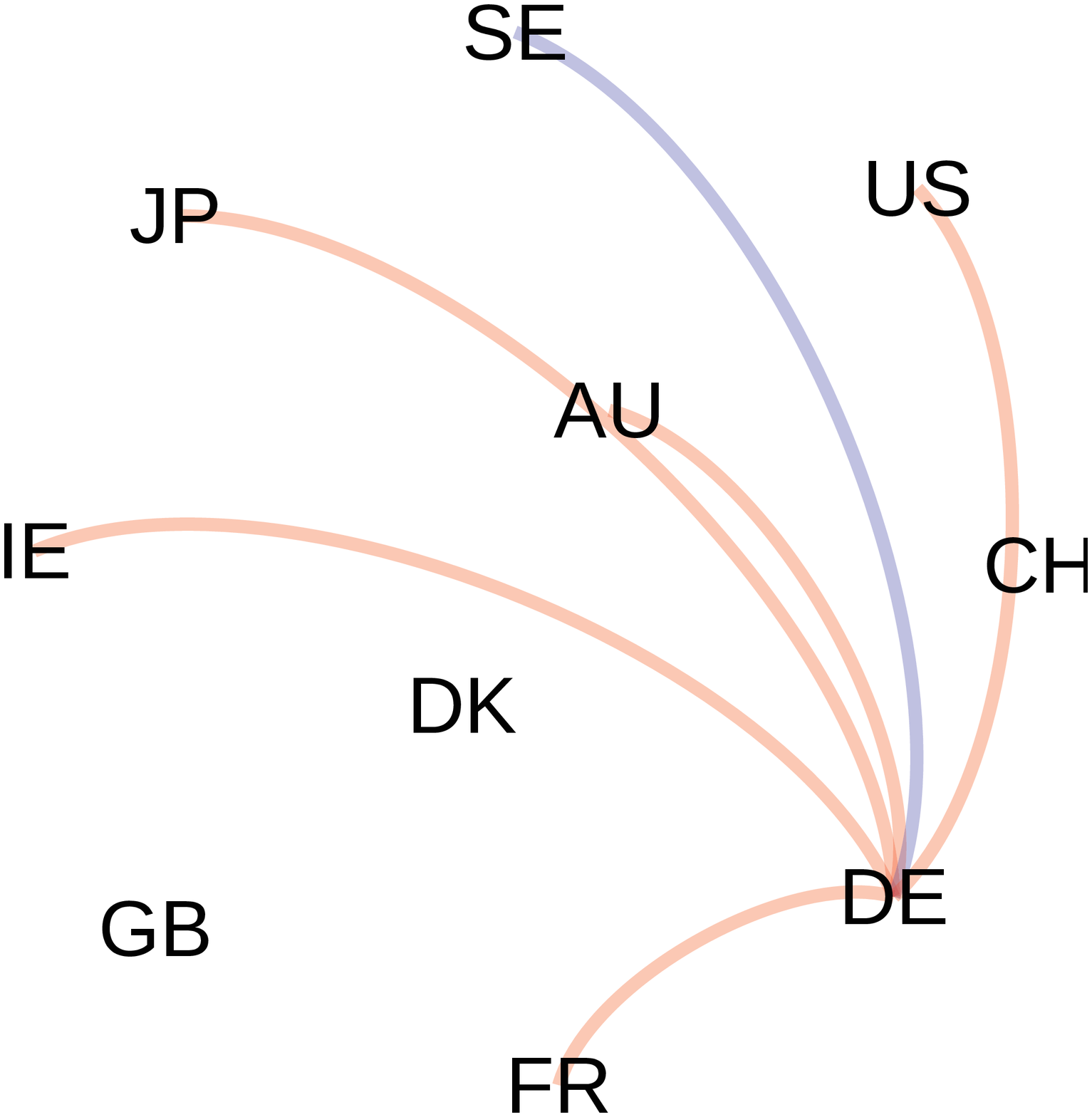} \\
\end{tabular}
\end{tabular}
\captionsetup{type=figure}
\captionof{figure}{Shock to GB capital inflows by -1\%. IRF at horizon $h=1$ for all (\textit{panel a}) and Germany (\textit{panel b}) financial and trade transactions. IRF at horizon $h=2$ for all (\textit{panel c}) and Germany (\textit{panel d}) financial and trade transactions. In each plot negative coefficients are in blue and positive in red.}
\label{fig:application_ART_IRF_UKimp1}
\end{figure}

\begin{figure}[h!t]
\setlength{\abovecaptionskip}{2pt}
\captionsetup{width=.90\linewidth}
\centering
\begin{tabular}{cc}
\multicolumn{2}{c}{(a) Network IRF at $h=1$} \\
\begin{rotate}{90} \hspace*{20pt} {\scriptsize Financial (layer 2) \hspace{40pt} Trade (layer 1)} \end{rotate} &
%\begin{rotate}{90} \hspace*{35pt} \normalsize{$\substack{\text{Financial\hspace{80pt}Trade}\\[5pt] \hspace{5pt}\text{(layer 2)\hspace{75pt}(layer 1)}}$} \end{rotate} & 
%\includegraphics[trim= 0mm 0mm 0mm 0mm,clip,height= 8.6cm, width= 5.5cm]{IRF_GBfImEx1.eps}
\includegraphics[trim= 0mm 0mm 0mm 0mm,clip,height= 8.0cm, width= 4.5cm]{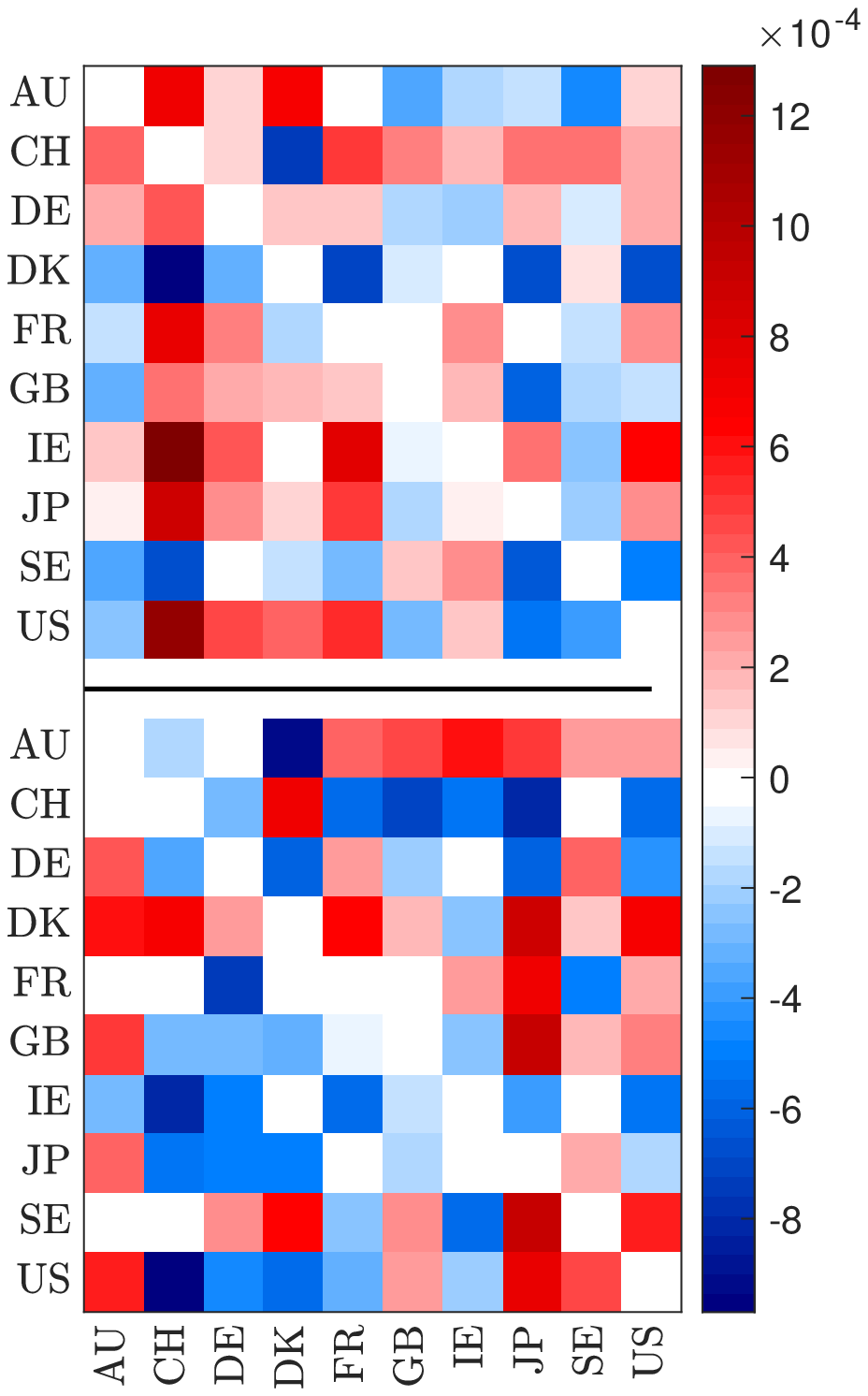}
\end{tabular}
%\hspace*{4ex}
\begin{tabular}{cc}
% & \textcolor{white}{$h=0$} \\
\multicolumn{2}{c}{(b) IRF for Germany's edges at $h=1$} \\[10pt]
\begin{rotate}{90} \hspace*{-95pt} {\scriptsize Financial (layer 2) \hspace{45pt} Trade (layer 1)} \end{rotate} &
%\begin{rotate}{90} \hspace*{-75pt} \normalsize{$\substack{\text{Financial\hspace{80pt}Trade}\\[5pt] \hspace{5pt}\text{(layer 2)\hspace{75pt}(layer 1)}}$} \end{rotate} &
\begin{tabular}{ccc}
{\footnotesize exports} & & {\footnotesize outstanding credit} \\
\includegraphics[trim= 0mm 30mm 0mm 30mm,clip,height= 3.0cm, width= 2.5cm]{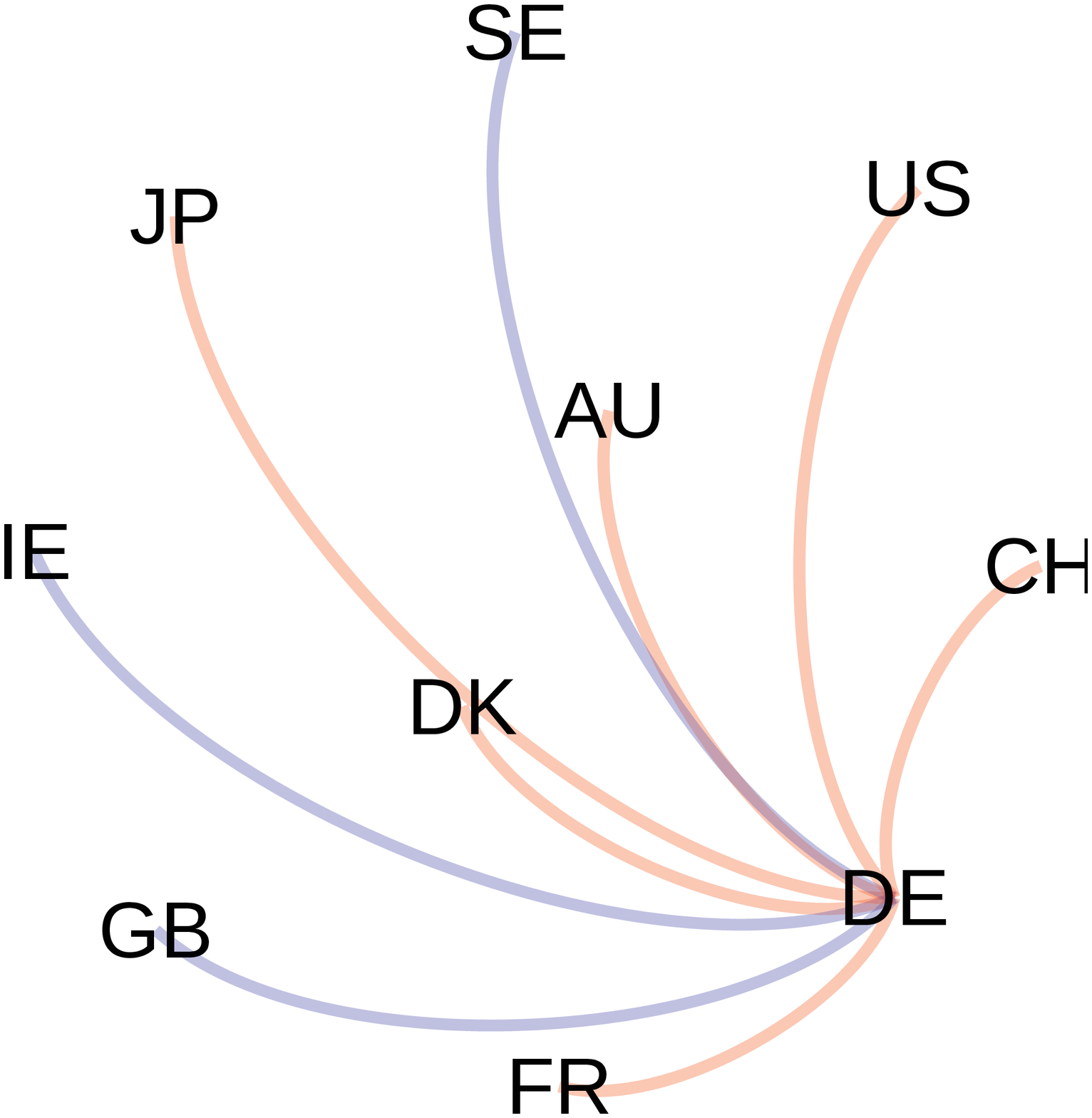} & & \includegraphics[trim= 0mm 30mm 0mm 30mm,clip,height= 3.0cm, width= 2.5cm]{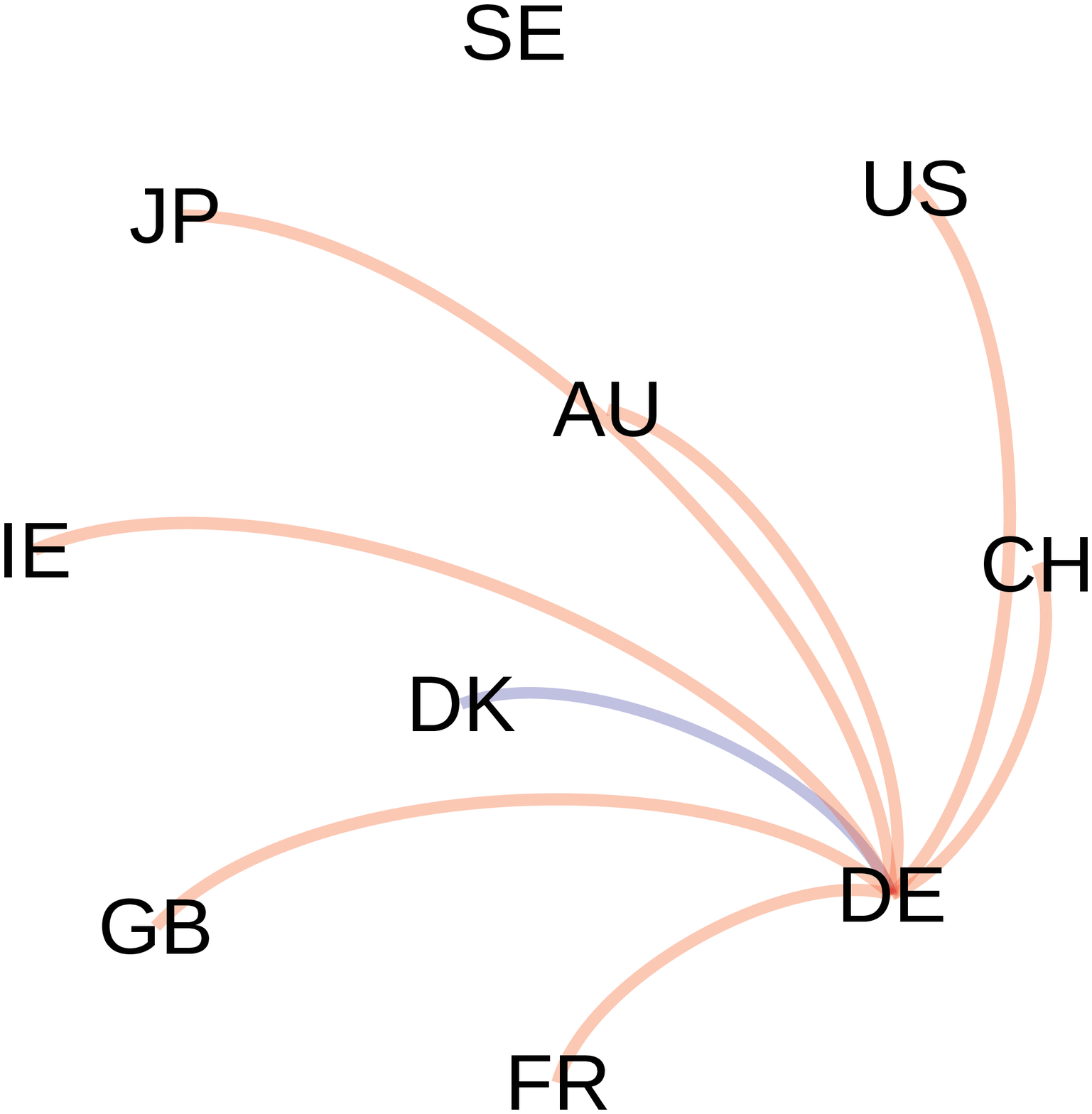} \\[-0ex]
{\footnotesize outflows} & & {\footnotesize outstanding debt} \\
\includegraphics[trim= 0mm 30mm 0mm 30mm,clip,height= 3.0cm, width= 2.5cm]{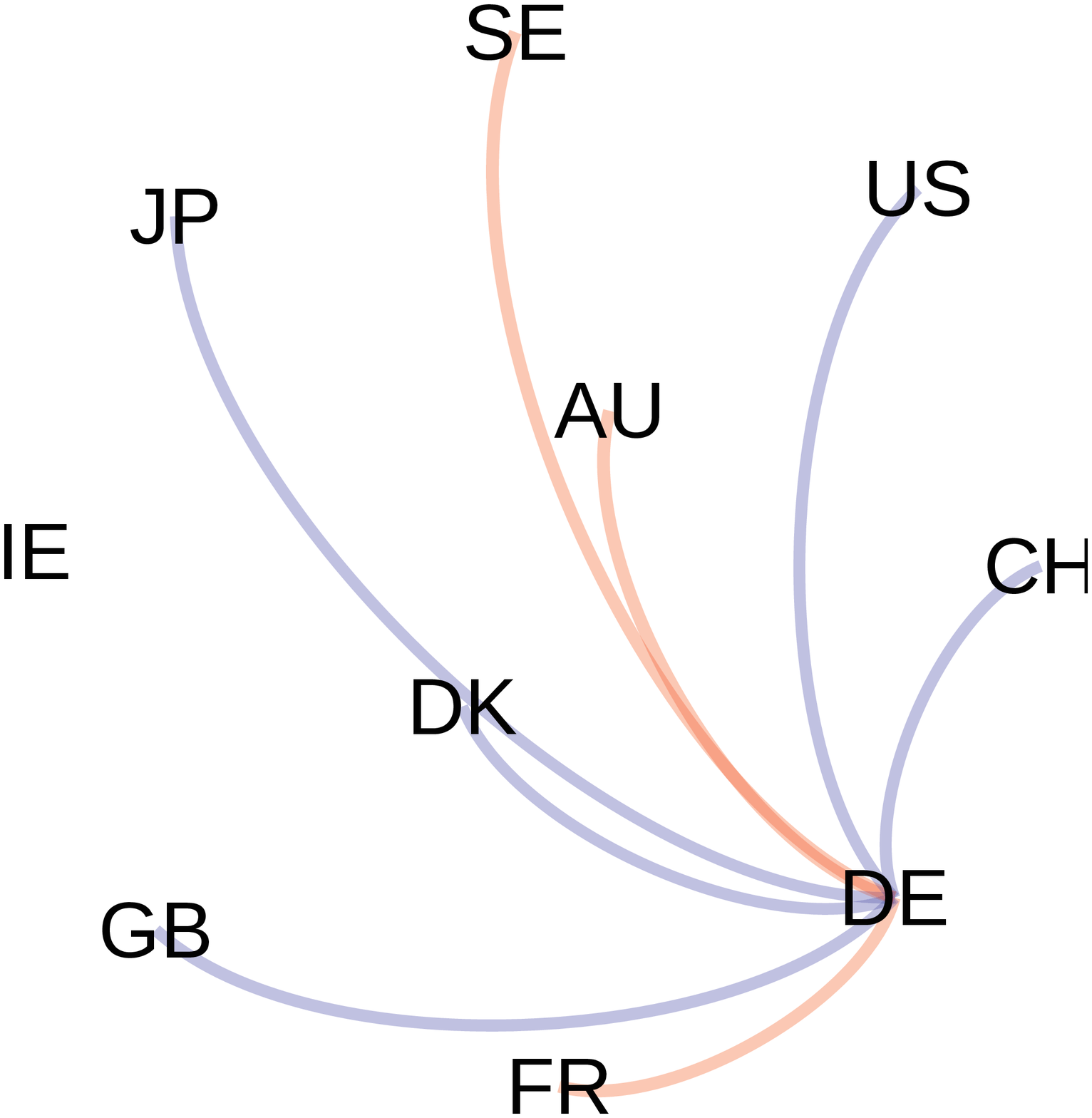} & & \includegraphics[trim= 0mm 30mm 0mm 30mm,clip,height= 3.0cm, width= 2.5cm]{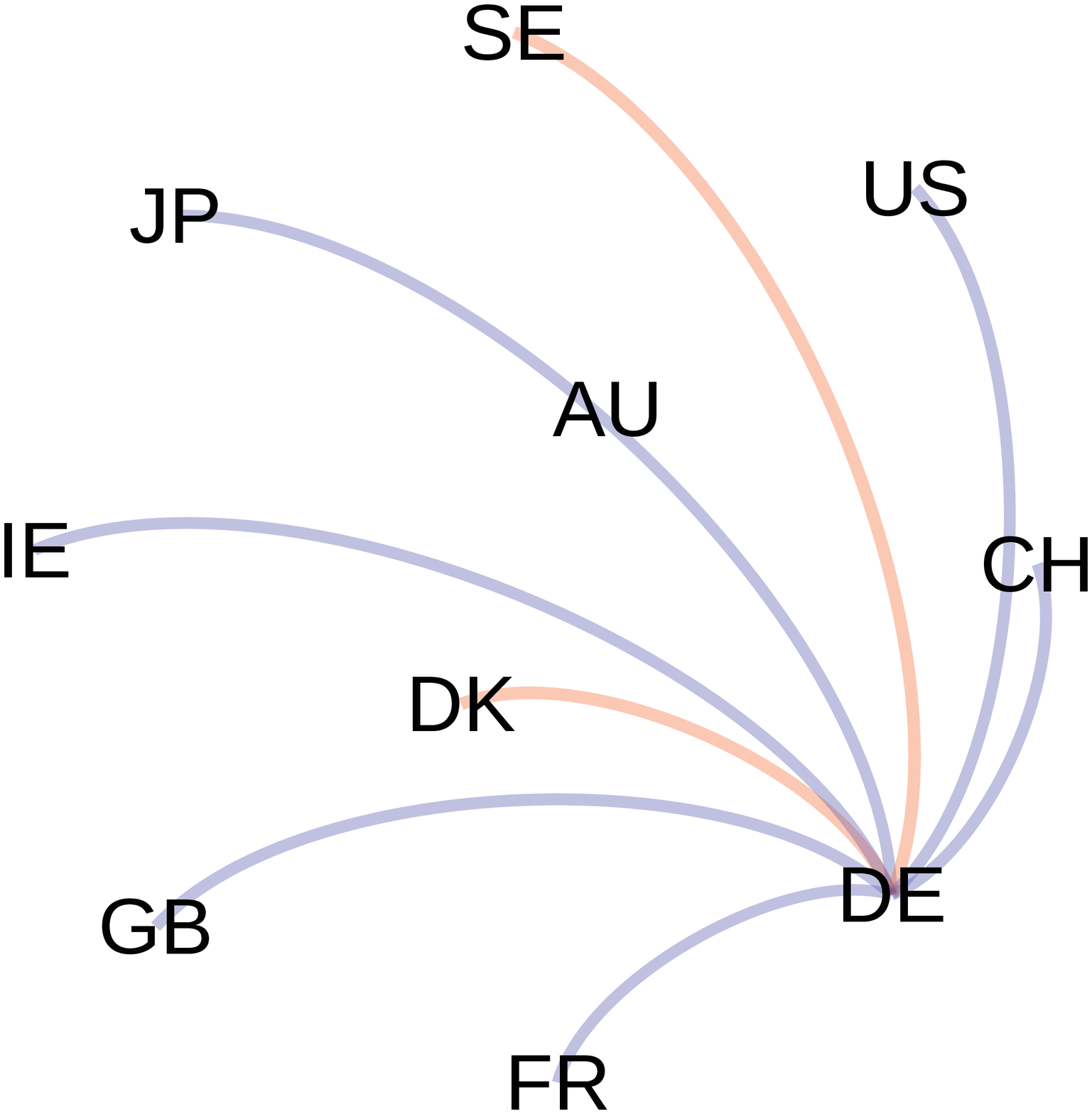} \\
\end{tabular}
\end{tabular}

\begin{tabular}{cc}
\multicolumn{2}{c}{(c) Network IRF at $h=2$} \\
\begin{rotate}{90} \hspace*{20pt} {\scriptsize Financial (layer 2) \hspace{40pt} Trade (layer 1)} \end{rotate} &
%\begin{rotate}{90} \hspace*{35pt} \normalsize{$\substack{\text{Financial\hspace{80pt}Trade}\\[5pt] \hspace{5pt}\text{(layer 2)\hspace{75pt}(layer 1)}}$} \end{rotate} & 
%\includegraphics[trim= 0mm 0mm 0mm 0mm,clip,height= 8.6cm, width= 5.5cm]{IRF_GBfImEx2.eps}
\includegraphics[trim= 0mm 0mm 0mm 0mm,clip,height= 8.0cm, width= 4.5cm]{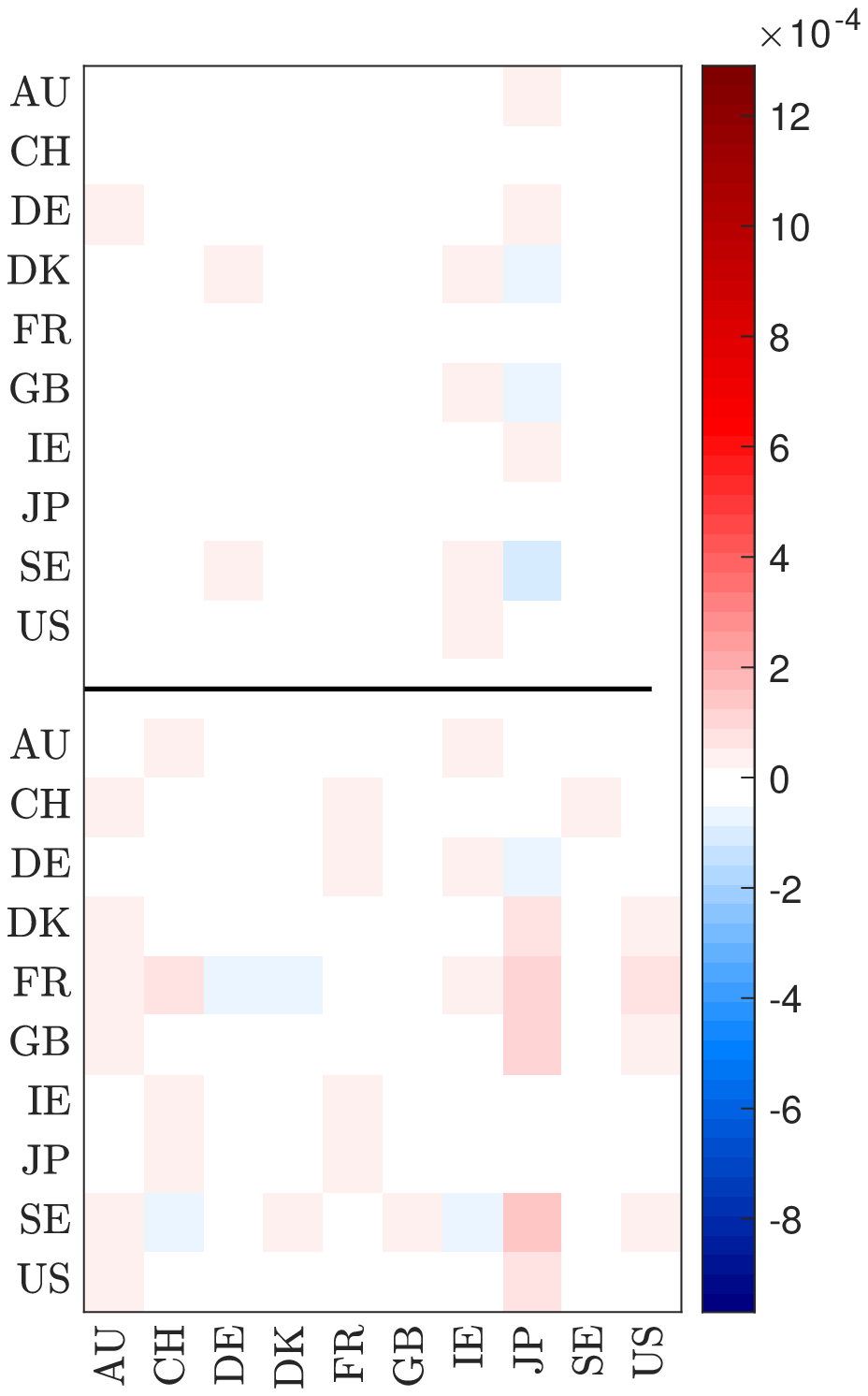}
\end{tabular}
%\hspace*{4ex}
\begin{tabular}{cc}
% & \textcolor{white}{$h=0$} \\
\multicolumn{2}{c}{(d) IRF for Germany's edges at $h=2$} \\[10pt]
\begin{rotate}{90} \hspace*{-95pt} {\scriptsize Financial (layer 2) \hspace{45pt} Trade (layer 1)} \end{rotate} &
%\begin{rotate}{90} \hspace*{-75pt} \normalsize{$\substack{\text{Financial\hspace{80pt}Trade}\\[5pt] \hspace{5pt}\text{(layer 2)\hspace{75pt}(layer 1)}}$} \end{rotate} &
\begin{tabular}{ccc}
{\footnotesize exports} & & {\footnotesize outstanding credit} \\
\includegraphics[trim= 0mm 30mm 0mm 30mm,clip,height= 3.0cm, width= 2.5cm]{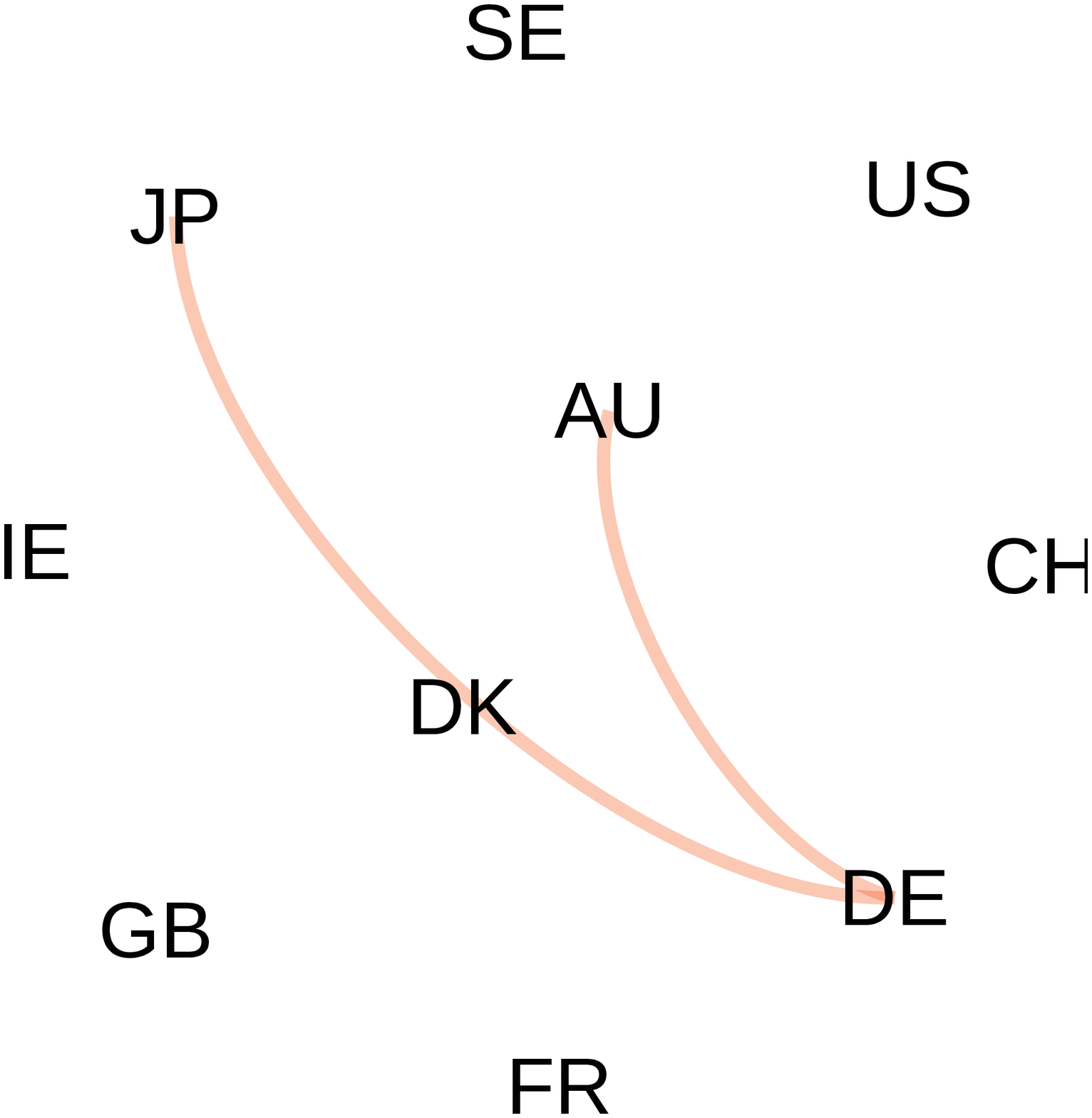} & & \includegraphics[trim= 0mm 30mm 0mm 30mm,clip,height= 3.0cm, width= 2.5cm]{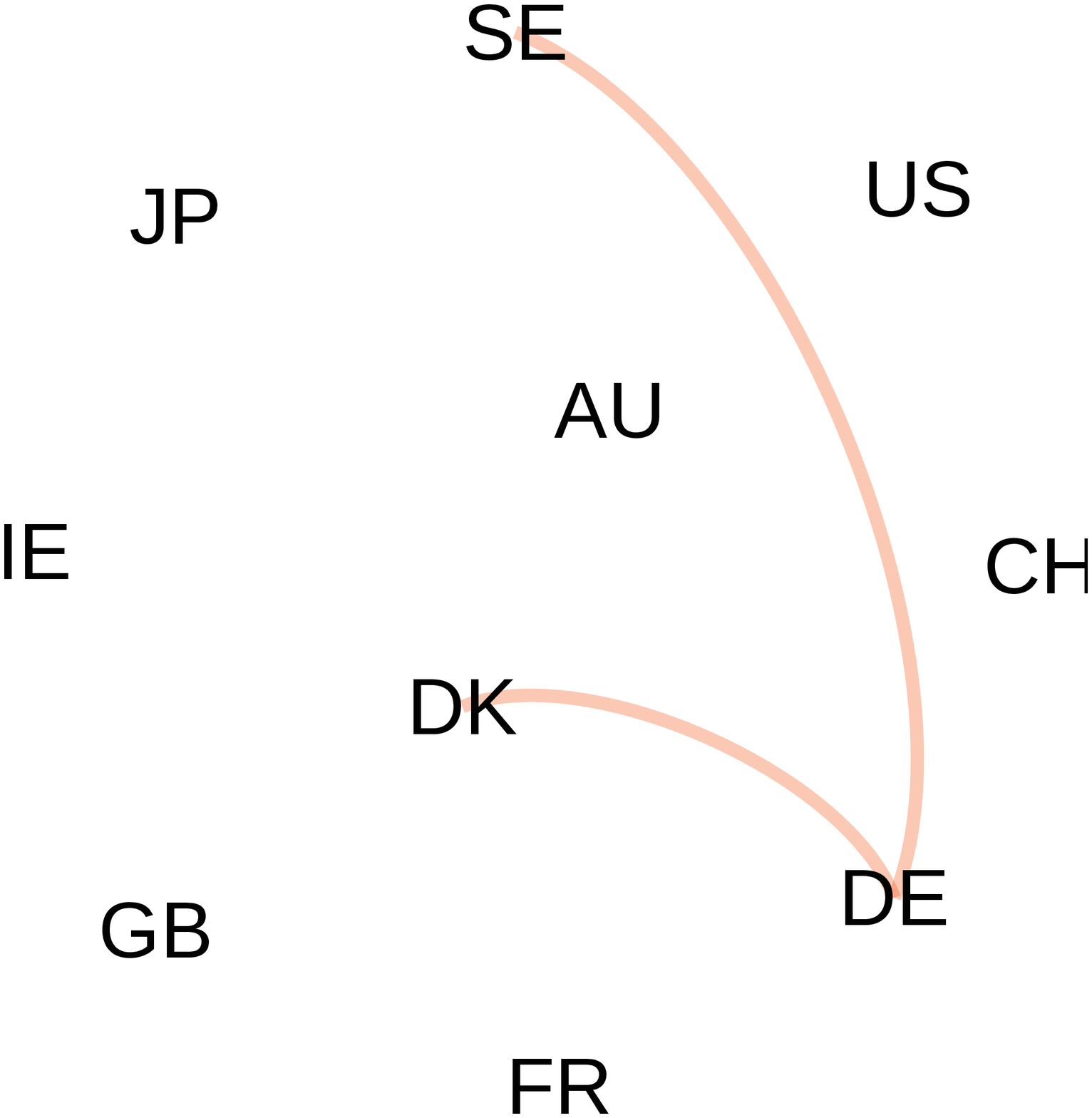} \\[-0ex]
{\footnotesize outflows} & & {\footnotesize outstanding debt} \\
\includegraphics[trim= 0mm 30mm 0mm 30mm,clip,height= 3.0cm, width= 2.5cm]{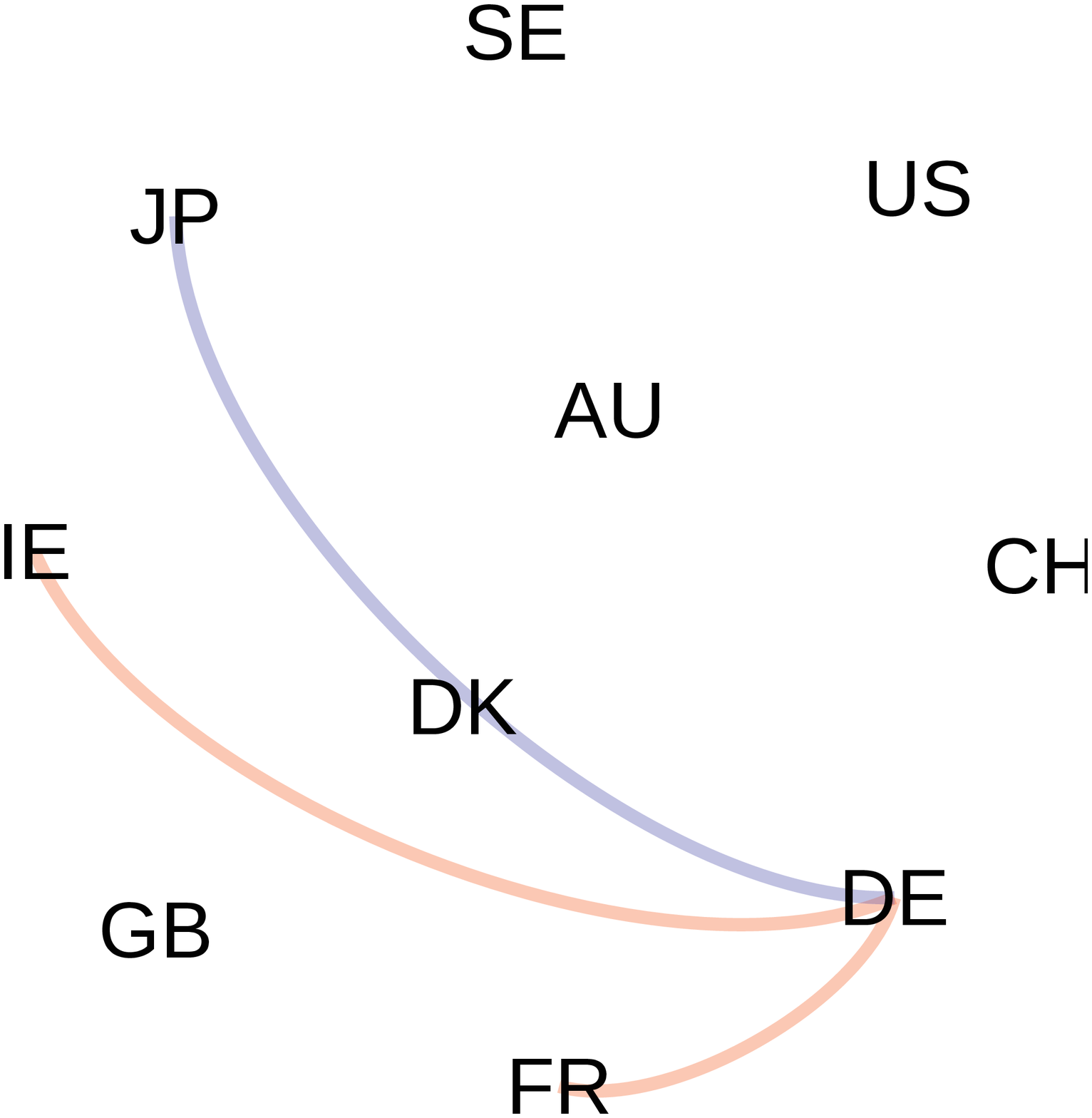} & & \includegraphics[trim= 0mm 30mm 0mm 30mm,clip,height= 3.0cm, width= 2.5cm]{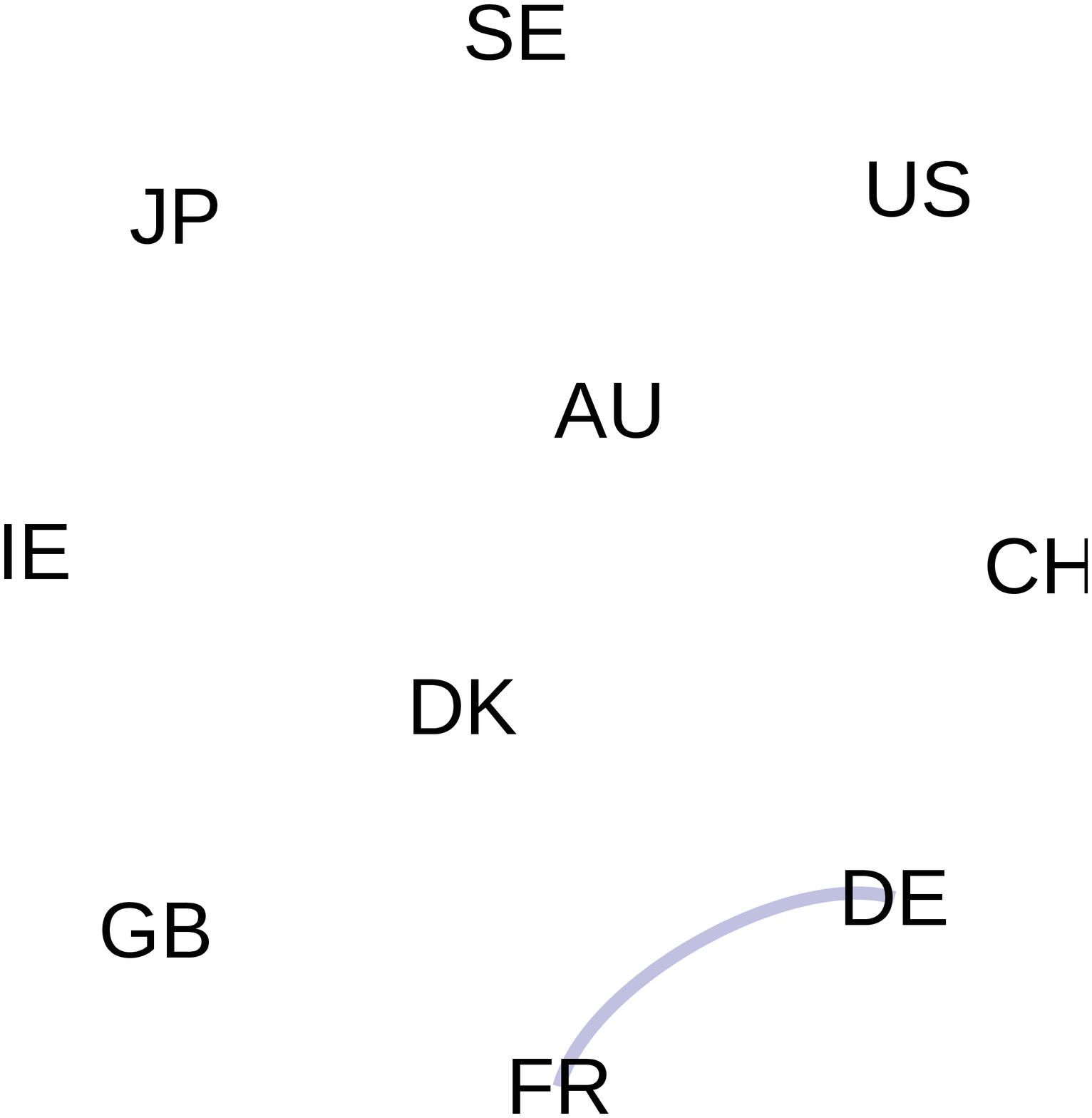}
\end{tabular}
\end{tabular}
\captionsetup{type=figure}
\captionof{figure}{Shock to GB capital inflows by -1\% and outflows by +1\%. IRF at horizon $h=1$ for all (\textit{panel a}) and Germany (\textit{panel b}) financial and trade transactions. IRF at horizon $h=2$ for all (\textit{panel c}) and Germany (\textit{panel d}) financial and trade transactions. In each plot negative coefficients are in blue and positive in red.}
\label{fig:application_ART_IRF_UKimpexp1}
\end{figure}

Fig. \ref{fig:application_ART_IRF_UKimp1} shows the block Cholesky IRF at horizon $h=1,2$, resulting from a negative $1\%$ shock to GB's outstanding debt\footnote{Again, the shock is allocated across countries to reflect country-specific shares of the last period in the sample.}.
The main findings follow.

\textit{Global effect} on the network.
We observe heterogeneous effects across countries. Effects on the trade layer at horizon $1$ are equally heterogeneous, but smaller in magnitude compared with the financial layer.

\textit{Local effect on Germany}.
Compared with other countries, the shock has smaller effects on Germany's trade. The negative shock to GB's outstanding debt has a negative impact on Germany's exports and imports to all countries but Ireland and Sweden for exports and Denmark for imports. Germany's outstanding credit increases vis-\'a-vis Denmark, GB, Japan and US. Germany's outstanding debt increases against all countries but Denmark and Sweden, in particular against France, Japan and Ireland.
At horizon $2$ responses are not reverted, but  nearly all effects turn insignificant, providing evidence of monotone and fast decay of the IRFs.

\textit{Local effect} on other countries.
On the trade layer at horizon $1$, we observe a positive response in Denmark's exports and on average a negative response of Switzerland's, Ireland's and Japan's exports.
France and Sweden are the most affected countries on the financial layer: The increase in outstanding credit of France towards Germany, Denmark and GB is counterbalanced by a reduction in Sweden's outstanding credit towards the same countries. We observe reverse effects concerning France's and Sweden's outstanding credit towards Switzerland and Ireland.
Finally, Ireland's outstanding credit reacts positively towards most other countries.

Compared with responses to the shock to US imports, the persistence of a negative shock to GB's outstanding debt is slightly stronger, see impulse responses at horizon $2$ in Fig. \ \ref{fig:application_ART_IRF_UKimp1}. The decay is monotonic. However, the speed of decay is heterogeneous across countries. For some countries, there are small effects at horizon 2, while for others the effects are completely wiped already.
Overall, we do not find evidence of a relation between the size of a country in terms of exports or outstanding credit and the persistence in the impulse response. At the most, persistence seems determined by the origin of the shock, the effects of a financial shock being more persistent than those of a trade shock.

Finally, in Fig. \ref{fig:application_ART_IRF_UKimpexp1} we plot the block Cholesky IRF, respectively, at horizon $h=1,2$, resulting from a $1\%$ negative shock to GB's outstanding debt coupled with a $1\%$ positive shock to GB's outstanding credit. The main findings follow.

\textit{Global effect} on the network.
The results remarkably differ from the previous ones (see Fig. \ref{fig:application_ART_IRF_UKimp1}). The responses to this simultaneous shock in GB's outstanding debt and credit are larger, in particular in the trade layer.
However, already at horizon $2$ responses are nearly fully decayed.
The results in Fig. \ref{fig:application_ART_IRF_UKimp1} and Fig. \ref{fig:application_ART_IRF_UKimpexp1} suggest that an increase in GB's outstanding credit has an overall positive effect on trade, stimulating export/import activities of most other countries.

\textit{Local effect on Germany}.
One period after the shock, we observe an overall positive effect on German exports, the exception being towards GB, Ireland and Sweden. Imports react mostly positively. Imports from US and Ireland react most, while those from Denmark react negatively.
The responses of Germany's outstanding debt vis-\'a-vis most countries but Denmark and Sweden are negative, especially against France.
At horizon $2$ Germany's responses have nearly faded away, suggesting a rapid monotone decay of the shock's effect.

\textit{Local effect} on other countries.
In particular, the reactions of Switzerland's imports and outstanding debt are strikingly different from the previous case, compare with Fig. \ref{fig:application_ART_IRF_UKimp1}. Imports from US and Ireland, and to a lesser extent from France and Austria, are strongly boosted, while those from Denmark and Sweden decrease strongly.
Moreover, we note that Japan's outstanding debt increases significantly against most countries. We interpret this as a signal for Japan's attractiveness for foreign capital. Compared with the previous exercise, France's financial responses are now mostly insignificant, or of opposite sign.
Finally, the reactions of GB's exports and outstanding credit are heterogeneous, the latter ones being larger in absolute magnitude.

\section{Conclusions} \label{sec:conclusions}
We defined a new statistical framework for dynamic tensor regression. It is a generalisation of many models frequently used in time series analysis, such as VAR, panel VAR, SUR and matrix regression models. The PARAFAC decomposition of the tensor of regression coefficients allows to reduce the dimension of the parameter space but also permits to choose flexible multivariate prior distributions, instead of multidimensional ones. Overall, this allows to encompass sparsity beliefs and to design efficient algorithm for posterior inference.

%We tested the Gibbs sampler algorithm on synthetic matrix-valued datasets with matrices of different sizes, obtaining good results in terms of both the estimation of the true value of the parameter and the efficiency.

The proposed methodology has been used for analysing the temporal evolution of the international trade and financial network, and the investigation has been complemented with an impulse response analysis. We have found evidence of (i) wide heterogeneity in the sign and magnitude of the estimated coefficients; (ii) stationarity of the network process.
The impulse response analysis has highlighted the role of network topology in shock propagation across countries and over time. Irrespective of its origin, any shock is found to propagate between layers, but financial shocks are more persistent than those on international trade. Moreover, we we do not find evidence of a relation between the size of a country, expressed by the total trade or capital exports, and the persistence its response to a shock.
Finally, we have found evidence of substitution effects in response to the shocks, meaning that pairs of countries experience opposite effects from a shock to another country. In conclusion, our dynamic model can be used for predicting possible trade creation and diversion effects.

%\section*{Acknowledgements}
%This research has benefited from the use of the Scientific Computation System of Ca' Foscari University of Venice (SCSCF) for the computational for the implementation of the estimation procedure.

\section*{Supplementary Material}
Supplementary material including background results on tensors, the derivation of the posterior, simulation experiments and the description of the data is available online\footnote{\url{https://matteoiacopini.github.io/docs/BiCaIaKa_Supplement.pdf}}.

%%%%%%%%%%%%%%%%%%%%%%%%%%%%%%%%%%%%%%%%%%%%%%%%%%%
% BIBLIOGRAPHY - plain apalike siam chicago ecta
\bibliographystyle{plain}
%\bibliography{refsDynTens}

%%%%%%%%%%%%%%%%%%%%%%%%%%%%%%%%%%%%%%%%%%%%%%%%%%%

\appendix
\renewcommand{\theequation}{\Alph{section}.\arabic{equation}} % equations as (A.1)...

\setcounter{equation}{0}
\section{Background Material on Tensor Calculus} \label{sec:apdx_tensor_calculus}
This appendix provides the main tools used in the paper. See the supplement for further results and details.
A $N$-order tensor is an element of the tensor product of $N$ vector spaces. Since there exists a isomorphism between two vector spaces of dimensions $N$ and $M<N$, it is possible to define a one-to-one map between their elements, that is, between a $N$-order tensor and a $M$-order tensor. %We call this tensor reshaping and give its formal definition below.

\begin{definition}[Tensor reshaping]
Let $V_1,\ldots,V_N$ and $U_1,\ldots,U_M$ be vector subspaces $V_n, U_m \subseteq \R$ and $\mathcal{X} \in \R^{I_1\times \ldots \times I_N} = V_1 \otimes \ldots \otimes V_N$ be a $N$-order real tensor of dimensions $I_1,\ldots,I_N$. Let $(\mathbf{v}_1,\ldots,\mathbf{v}_N)$ be a canonical basis of $\R^{I_1\times\ldots\times I_N}$ and let $\Pi_S$ be the projection defined as
\begin{align*}
\Pi_S : & V_1 \otimes \ldots \otimes V_N \rightarrow V_{s_1} \otimes \ldots \otimes V_{s_k} \\
 & \mathbf{v}_1 \otimes \ldots \otimes \mathbf{v}_N \mapsto \mathbf{v}_{s_1} \otimes \ldots \otimes \mathbf{v}_{s_k}
\end{align*}
with $S = \lbrace s_1,\ldots,s_k \rbrace \subset \lbrace 1,\ldots,N \rbrace$. Let $(S_1,\ldots,S_M)$ be a partition of $\lbrace 1,\ldots,N \rbrace$. The $(S_1,\ldots,S_M)$ tensor reshaping of $\mathcal{X}$ is defined as $\mathcal{X}_{(S_1,\ldots,S_M)} = (\Pi_{S_1}\mathcal{X}) \otimes \ldots \otimes (\Pi_{S_M}\mathcal{X}) = U_1 \otimes \ldots \otimes U_M$.
%\begin{align*}
%\mathcal{X}_{(S_1,\ldots,S_M)} & = (\Pi_{S_1}\mathcal{X}) \otimes \ldots \otimes (\Pi_{S_M}\mathcal{X}) = U_1 \otimes \ldots \otimes U_M.
%%\mathcal{X}_{(S_1,\ldots,S_M)} & = (\Pi_{S_1}\mathcal{X}) \otimes \ldots \otimes (\Pi_{S_M}\mathcal{X}) \\
%% & \quad \in \Big( \bigotimes_{s\in S_1} V_s \Big) \otimes \ldots \otimes \Big( \bigotimes_{s\in S_M} V_s \Big) \\
%% & = U_1 \otimes \ldots \otimes U_M.
%\end{align*}
The mapping is an isomorphism between $V_1 \otimes \ldots \otimes V_N$ and $U_1 \otimes \ldots \otimes U_M$.
\end{definition}

The matricization is a particular case of reshaping a $N$-order tensor into a $2$-order tensor, by choosing a mapping between the tensor modes and the rows and columns of the resulting matrix, then permuting the tensor and reshaping it, accordingly.
\begin{definition}[Matricization]
Let $\mathcal{X}$ be a $N$-order tensor with dimensions $I_1,\ldots,I_N$. Let the ordered sets $\mathscr{R} = \lbrace r_1,\ldots,r_L \rbrace$ and $\mathscr{C} = \lbrace c_1,\ldots,c_M \rbrace$ be a partition of $\mathbf{N} = \lbrace 1,\ldots,N \rbrace$.
%and let $I_{\mathbf{N}}= \lbrace I_1,\ldots,I_N \rbrace$.
The matricized tensor is defined by
\begin{equation*}
\operatorname{mat}_{\mathscr{R}, \mathscr{C}}(\mathcal{X}) = \mathbf{X}_{(\mathscr{R}, \mathscr{C})} \in \R^{J\times K}, \qquad J= \prod_{n\in \mathscr{R}} I_n, \quad K= \prod_{n\in \mathscr{C}} I_n \, .
\end{equation*}
Indices of $\mathscr{R},\mathscr{C}$ are mapped to the rows and the columns, respectively, and
\begin{equation*}
\big( \mathbf{X}_{(\mathscr{R} \times \mathscr{C})} \big)_{j,k} = \mathcal{X}_{i_1,i_2,\ldots,i_N},  \quad  j= 1+\sum_{l=1}^L \Big( (i_{r_l}-1) \prod_{l'=1}^{l-1} I_{r_l'} \Big),  \quad  k= 1+\sum_{m=1}^M \Big( (i_{c_m}-1) \prod_{m'=1}^{m-1} I_{c_m'} \Big).
\end{equation*}
\end{definition}

\noindent The \textit{inner product} between two $(I_1\times\ldots\times I_N)$-dimensional tensors $\mathcal{X},\mathcal{Y}$ is defined as
\begin{equation*}
\langle \mathcal{X}, \mathcal{Y} \rangle = \sum_{i_1=1}^{I_1}\ldots\sum_{i_N=1}^{I_N} \mathcal{X}_{i_1,\ldots,i_N} \mathcal{Y}_{i_1,\ldots,i_N} 
%= \mathcal{X} \bar{\times}_N \mathcal{Y}.
\label{eq:apdx_tensor_scal_prod}
\end{equation*}
%The Tucker decomposition is a higher-order generalization of the Principal Component Analysis (PCA): a tensor $\mathcal{B}\in\R^{I_1\times\ldots\times I_N}$ is decomposed into the product (along the corresponding modes) of a ``core'' tensor $\mathcal{G}\in\R^{g_1\times\ldots\times g_N}$ and factor matrices $A^{(m)}\in\R^{I_m\times J_m}$, $m=1,\ldots,N$
%\begin{equation*}
%\mathcal{B} = \mathcal{G} \times_1 A^{(1)} \times_2 \ldots \times_N A^{(N)} = \sum_{i_1=1}^{g_1}\ldots \sum_{i_N=1}^{g_N} \mathcal{G}_{i_1,\ldots,i_N} \mathbf{a}_{i_1}^{(1)} \circ \ldots \circ \mathbf{a}_{i_N}^{(N)}
%\label{eq:tucker_decomposition_1}
%\end{equation*}
%where $\mathbf{a}_{i_l}^{(m)}\in\R^{g_m}$ is the $m$-th column of the matrix $A^{(m)}$. Each entry of the tensor is
%\begin{equation*}
%\mathcal{B}_{j_1,\ldots,j_N} = \sum_{i_1=1}^{g_1} \ldots \sum_{i_N=1}^{g_N} \mathcal{G}_{i_1,\ldots,i_N} \cdot A_{i_1,j_1}^{(1)} \cdots A_{i_N,j_N}^{(N)}
%\label{eq:tucker_decomposition_2}
%\end{equation*}
The PARAFAC($R$) decomposition (e.g., see \cite{KoldaBader09}), is rank-$R$ decomposition which represents a tensor $\mathcal{B}\in\R^{I_1\times\ldots\times I_N}$ as a finite sum of $R$ rank-$1$ tensors obtained as the outer products of $N$ vectors (called marginals) $\boldsymbol{\beta}_j^{(r)} \in\R^{I_j}$
\begin{equation*}
\mathcal{B} = \sum_{r=1}^R \mathcal{B}_r = \sum_{r=1}^R \boldsymbol{\beta}_1^{(r)} \circ \ldots \circ \boldsymbol{\beta}_J^{(r)}.
\end{equation*}

\begin{lemma}[Contracted product -- some properties] \label{lemma:contracted_properties}
Let $\mathcal{X} \in \R^{I_1\times \ldots\times I_N}$ and $\mathcal{Y} \in \R^{J_1\times \ldots\times J_N \times J_{N+1}\times\ldots\times J_{N+P}}$. Let $(\mathscr{S}_1,\mathscr{S}_2)$ be a partition of $\lbrace 1,\ldots,N+P \rbrace$, where $\mathscr{S}_1 = \lbrace 1,\ldots,N \rbrace$, $\mathscr{S}_2 = \lbrace N+1,\ldots,N+P \rbrace$. It holds:
\begin{enumerate}[label=(\roman*)]
\item if $P=0$ and $I_n = J_n$, $n=1,\ldots,N$, then $\mathcal{X} \bar{\times}_N \mathcal{Y} = \langle \mathcal{X}, \mathcal{Y} \rangle = \vecc{\mathcal{X}}' \cdot \vecc{\mathcal{Y}}$.

\item if $P>0$ and $I_n = J_n$ for $n=1,\ldots,N$, then
\begin{align*}
\mathcal{X} \bar{\times}_N \mathcal{Y} & = \vecc{\mathcal{X}} \times_1
\mathcal{Y}_{(\mathscr{S}_1,\mathscr{S}_2)} \; \; \in\R^{j_1\times\ldots\times j_P} \\
\mathcal{Y} \bar{\times}_N \mathcal{X} & = \mathcal{Y}_{(\mathscr{S}_1,\mathscr{S}_2)} \times_1 \vecc{\mathcal{X}} \; \; \in\R^{j_1\times\ldots\times j_P}.
\end{align*}

\item let $\mathscr{R}=\lbrace 1,\ldots,N \rbrace$ and $\mathscr{C}=\lbrace N+1,\ldots,2N \rbrace$. If $P=N$ and $I_n = J_n = J_{N+n}$, $n=1,\ldots,N$, then
\begin{equation*}
\mathcal{X} \bar{\times}_N \mathcal{Y} \bar{\times}_N \mathcal{X} = \vecc{\mathcal{X}}' \mathbf{Y}_{(\mathscr{R}, \mathscr{C})} \vecc{\mathcal{X}}.
\end{equation*}

\item let $M=N+P$, then $\mathcal{X} \circ \mathcal{Y} = \underline{\mathcal{X}} \bar{\times}_1 \underline{\mathcal{Y}}^T$, where $\underline{\mathcal{X}},\underline{\mathcal{Y}}$ are $(I_1\times\ldots\times I_N\times 1)$- and $(J_1\times\ldots\times J_M\times 1)$-dimensional tensors, respectively, given by $\underline{\mathcal{X}}_{:,\ldots,:,1} = \mathcal{X}$, $\underline{\mathcal{Y}}_{:,\ldots,:,1} = \mathcal{Y}$ and $\underline{\mathcal{Y}}^T_{j_1,\ldots,j_M,j_{M+1}} = \underline{\mathcal{Y}}_{j_{M+1},j_M,\ldots,j_1}$.
\end{enumerate}
\end{lemma}
\begin{proof}
Case (i). By definition of contracted product and tensor scalar product
\begin{align*}
\mathcal{X} \bar{\times}_N \mathcal{Y} & = \sum_{i_1=1}^{I_1} \ldots \sum_{i_N=1}^{I_N} \mathcal{X}_{i_1,\ldots,i_N} \mathcal{Y}_{i_1,\ldots,i_N} = \langle \mathcal{X}, \mathcal{Y} \rangle = \vecc{\mathcal{X}}' \cdot \vecc{\mathcal{Y}}.
\end{align*}
Case (ii). Define $I^* = \prod_{n=1}^N I_n$ and $k=1+\sum_{j=1}^N (i_j-1) \prod_{m=1}^{j-1} I_m$. By definition of contracted product and tensor scalar product
\begin{align*}
\mathcal{X} \bar{\times}_N \mathcal{Y} & = \sum_{i_1=1}^{I_1} \ldots \sum_{i_N=1}^{I_N} \mathcal{X}_{i_1,\ldots,i_N} \mathcal{Y}_{i_1,\ldots,i_N,j_{N+1},\ldots,j_{N+P}} = \sum_{k=1}^{I^*} \mathcal{X}_{k} \mathcal{Y}_{k,j_{N+1},\ldots,j_{N+P}} \, .
\end{align*}
Note that the one-to-one correspondence established by the mapping between $k$ and $(i_1,\ldots,i_N)$ corresponds to that of the vectorization of a $(I_1\times\ldots\times I_N)$-dimensional tensor. It also corresponds to the mapping established by the tensor reshaping of a $(N+P)$-order tensor with dimensions $I_1,\ldots,I_N,J_{N+1},\ldots,J_{N+P}$ into a $(P+1)$-order tensor with dimensions $I^*,J_{N+1},\ldots,J_{N+P}$. Let $\mathscr{S}_1 = \lbrace 1,\ldots,N \rbrace$, then
\begin{align*}
\mathcal{X} \bar{\times}_N \mathcal{Y} = \sum_{i_1=1}^{I_1} \ldots \sum_{i_N=1}^{I_N} \mathcal{X}_{i_1,\ldots,i_N} \mathcal{Y}_{i_1,\ldots,i_N,:,\ldots,:} = \sum_{s_1=1}^{|\mathscr{S}_1|} \mathbf{x}_{s_1} \bar{\mathcal{Y}}_{s_1,:,\ldots,:}
%\vecc{\mathcal{X}} \times_1 \mathcal{Y}_{(S, N+1,\ldots,N+P)}.
\end{align*}
where $\bar{\mathcal{Y}} = \operatorname{reshape}_{(\mathscr{S}_1,N+1,\ldots,N+P)}(\mathcal{Y})$. Following the same approach, and defining $\mathscr{S}_2 = \lbrace N+1,\ldots,N+P \rbrace$, we obtain the second part of the result.\\
Case (iii). We follow the same strategy adopted in case b). Let $\mathbf{x}=\vecc{\mathcal{X}}$, $S_1 = \lbrace 1,\ldots,N\rbrace$ and $S_2 = \lbrace N+1,\ldots,N+P \rbrace$, such that $(S-1,S_2)$ is a partition of $\lbrace 1,\ldots,N+P \rbrace$. Let $k,k'$ be defined as in case b). Then
\begin{align*}
\hspace*{-30pt}  & \mathcal{X} \bar{\times}_N \mathcal{Y} \bar{\times}_N \mathcal{X} = \sum_{i_1=1}^{I_1} \ldots \sum_{i_N=1}^{I_N} \sum_{i_1'=1}^{I_1} \ldots \sum_{i_N'=1}^{I_N} \mathcal{X}_{i_1,\ldots,i_N} \mathcal{Y}_{i_1,\ldots,i_N,i_1',\ldots,i_N'} \mathcal{X}_{i_1',\ldots,i_N'} \\
\hspace*{-30pt}  & = \sum_{k=1}^{I^*} \sum_{i_1'=1}^{I_1} \ldots \sum_{i_N'=1}^{I_N} \mathbf{x}_{k} \mathcal{Y}_{k,i_1',\ldots,i_N'} \mathcal{X}_{i_1',\ldots,i_N'}  = \sum_{k=1}^{I^*} \sum_{k'=1}^{I^*} \mathbf{x}_{k} \mathcal{Y}_{k,k'} \mathbf{x}_{k'} = \vecc{\mathcal{X}}' \mathcal{Y}_{(S_1,S_2)} \vecc{\mathcal{X}}.
\end{align*}
Case (iv).
%We use the fact that every $N$-order tensor can be interpreted as a $(N+1)$-order tensor, whose last mode has length $1$. Note that, when $N=1$ this is equivalent to interpret a $I_1$-dimensional vector as a $(I_1\times 1)$ matrix. By the definition of outer and contracted product, we have that $\mathcal{Z} = \mathcal{X} \circ \mathcal{Y}$ and $\underline{\mathcal{Z}} = \underline{\mathcal{X}} \bar{\times}_1 \underline{\mathcal{Y}}^T$ are $(I_1\times\ldots\times I_N\times J_1\times\ldots\times J_M)$-dimensional tensors. We now show that they coincide.
Let $\mathbf{i} = (i_1,\ldots,i_N)$ and $\mathbf{j} = (j_1,\ldots,j_M)$ be two multi-indexes. By the definition of outer and contracted product we get $(\mathcal{X} \circ \mathcal{Y})_{\mathbf{i},\mathbf{j}} = \underline{\mathcal{X}}_{\mathbf{i},1} \underline{\mathcal{Y}}_{1,\mathbf{j}} = (\underline{\mathcal{X}} \bar{\times}_1 \underline{\mathcal{Y}}^T)_{\mathbf{i},\mathbf{j}}$.
Therefore, with a slight abuse of notation, we use $\underline{\mathcal{Y}} = \mathcal{Y}$ and write $\mathcal{Y} \circ \mathcal{Y} = \mathcal{Y} \bar{\times}_1 \mathcal{Y}^T$, when the meaning of the products is clear form the context.
\end{proof}

%The next result gives the relation between the matricization of a tensor resulting from the outer product of matrices and the Kronecker product.
\begin{lemma}[Kronecker - matricization] \label{lemma:matricize_outer_kronecker}
Let $X_n$ be a $I_n\times I_n$ matrix, for $n=1,\ldots,N$, and let $\mathcal{X} = X_1 \circ \ldots \circ X_N$ be the $(I_1\times\ldots\times I_N\times I_1\times\ldots\times I_N)$-dimensional tensor obtained as the outer product of the matrices $X_1,\ldots,X_N$. Let $(\mathscr{S}_1,\mathscr{S}_2)$ be a partition of $I_\mathbf{N} = \lbrace 1,\ldots,2N \rbrace$, where $\mathscr{S}_1 = \lbrace 1,\ldots,N \rbrace$ and $\mathscr{S}_2 = \lbrace N+1,\ldots,N \rbrace$. Then $\mathcal{X}_{(\mathscr{S}_1,\mathscr{S}_2)} = \mathbf{X}_{(\mathscr{R}, \mathscr{C})} = (X_N \otimes \ldots \otimes X_1)$.
\end{lemma}
\begin{proof}
Use the pair of indices $(i_n,i_n')$ for the entries of the matrix $X_n$, $n=1,\ldots,N$. By definition of outer product $(X_1 \circ \ldots \circ X_N)_{i_1,\ldots,i_N,i_1',\ldots,i_N'} = (X_1)_{i_1,i_1'} \cdot \ldots \cdot (X_N)_{i_N,i_N'}$.
By definition of matricization, $\mathcal{X}_{(\mathscr{S}_1,\mathscr{S}_2)} = \mathbf{X}_{(\mathscr{R}, \mathscr{C})}$. Moreover $(\mathcal{X}_{(\mathscr{S}_1,\mathscr{S}_2)})_{h,k} = \mathcal{X}_{i_1,\ldots,i_{2N}}$ with $h = \sum_{p=1}^N (i_{S_{1,p}} -1) \prod_{q=1}^{p-1} J_{S_{1,p}}$ and $k = \sum_{p=1}^N (i_{S_{2,p}} -1) \prod_{q=1}^{p-1} J_{S_{2,p}}$.
By definition of the Kronecker product, the entry $(h',k')$ of $(X_N \otimes \ldots \otimes X_1)$ is $(X_N \otimes \ldots \otimes X_1)_{h',k'} = (X_N)_{i_N',i_N'} \cdot \ldots \cdot (X_1)_{i_1,i_1'}$, where $h' = \sum_{p=1}^N (i_{S_{1,p}} -1) \prod_{q=1}^{p-1} J_{S_{1,p}}$ and $k' = \sum_{p=1}^N (i_{S_{2,p}} -1) \prod_{q=1}^{p-1} J_{S_{2,p}}$.
Since $h=h'$ and $k=k'$ and the associated elements of $\mathcal{X}_{(\mathscr{S}_1,\mathscr{S}_2)}$ and $(X_N \otimes \ldots \otimes X_1)$ are the same, the result follows.
\end{proof}

\begin{lemma}[Outer product and vectorization] \label{lemma:outer_product_vec}
Let $\boldsymbol{\alpha}_1,\ldots,\boldsymbol{\alpha}_n$ be vectors such that $\boldsymbol{\alpha}_i$ has length $d_i$, for $i=1,\ldots,n$. Then, for each $j=1,\ldots,n$, it holds
\begin{align*}
\vecc{\mathop{\circ}_{i=1}^n \boldsymbol{\alpha}_i} & = \mathop{\otimes}_{i=1}^n \boldsymbol{\alpha}_{n-i+1} = \bigl( \boldsymbol{\alpha}_n \otimes \ldots \otimes \boldsymbol{\alpha}_{j+1} \otimes \mathbf{I}_{d_j} \otimes \boldsymbol{\alpha}_{j-1} \otimes \ldots \otimes \boldsymbol{\alpha}_1 \bigr) \boldsymbol{\alpha}_j.
\end{align*}
\end{lemma}
\begin{proof}
The result follows from the definitions of vectorisation operator and outer product. For $n=2$, the result follows directly from
\begin{equation*}
\vecc{\boldsymbol{\alpha}_1 \circ \boldsymbol{\alpha}_2} = \vecc{\boldsymbol{\alpha}_1 \boldsymbol{\alpha}_2'} = \boldsymbol{\alpha}_2 \otimes \boldsymbol{\alpha}_1 = (\boldsymbol{\alpha}_2 \otimes \mathbf{I}_{d_1})\boldsymbol{\alpha}_1 = (\mathbf{I}_{d_2} \otimes \boldsymbol{\alpha}_1)\boldsymbol{\alpha}_2.
\end{equation*}
For $n > 2$ consider, without loss of generality, $n=3$ (an analogous proof holds for $n>3$). Then, from the definitions of outer product and Kronecker product we have
\begin{align*}
 & \hspace*{-30pt}   \vecc{\boldsymbol{\alpha}_1 \circ \boldsymbol{\alpha}_2 \circ \boldsymbol{\alpha}_3} = \\
 & \hspace*{-30pt}   = (\boldsymbol{\alpha}_1' \cdot \alpha_{2,1} \alpha_{3,1}, \ldots, \boldsymbol{\alpha}_1' \cdot \alpha_{2,d_2} \alpha_{3,1}, \boldsymbol{\alpha}_1' \cdot \alpha_{2,1} \alpha_{3,2}, \ldots, \boldsymbol{\alpha}_1' \cdot \alpha_{2,d_2} \alpha_{3,2}, \ldots, \boldsymbol{\alpha}_1' \cdot \alpha_{2,d_2} \alpha_{3,d_3})' \\
 & \hspace*{-30pt}   = \boldsymbol{\alpha}_3 \otimes \boldsymbol{\alpha}_2 \otimes \boldsymbol{\alpha}_1 = (\boldsymbol{\alpha}_3 \otimes \boldsymbol{\alpha}_2 \otimes \mathbf{I}_{d_1})\boldsymbol{\alpha}_1 = (\boldsymbol{\alpha}_3 \otimes \mathbf{I}_{d_2} \otimes \boldsymbol{\alpha}_1)\boldsymbol{\alpha}_2 = (\mathbf{I}_{d_3} \otimes \boldsymbol{\alpha}_2 \otimes \boldsymbol{\alpha}_1)\boldsymbol{\alpha}_3.
\end{align*}
\end{proof}

\begin{proof}[Proof of Proposition \ref{proposition:ART_stationatity}]
Denote with $L$ the lag operator, s.t. $L \mathcal{Y}_t = \mathcal{Y}_{t-1}$, by properties of the contracted product in Lemma \ref{lemma:contracted_properties}, case (iv), we get $(\mathcal{I} -\widetilde{\mathcal{A}}_1 L) \bar{\times}_N \mathcal{Y}_t = \widetilde{\mathcal{A}}_0 + \widetilde{\mathcal{B}} \bar{\times}_M \mathcal{X}_t + \mathcal{E}_t$.
We apply to both sides the operator $(\mathcal{I} + \widetilde{\mathcal{A}}_1 L + \widetilde{\mathcal{A}}_1^2 L^2 + \ldots + \widetilde{\mathcal{A}}_1^{t-1} L^{t-1})$, take $t\to \infty$, and get
\begin{align*}
\lim_{t\to \infty} (\mathcal{I}-\widetilde{\mathcal{A}}_1^{t} L^{t}) \bar{\times}_N \mathcal{Y}_t = \Big( \sum_{k=0}^\infty \widetilde{\mathcal{A}}_1^k L^k \Big) \bar{\times}_N (\widetilde{\mathcal{A}}_0 + \widetilde{\mathcal{B}} \bar{\times}_M \mathcal{X}_t + \mathcal{E}_t).
\end{align*}
From \cite{Behera19DrazinInverse_tensor_even}, if $\rho(\widetilde{\mathcal{A}}_1) < 1$ and $\mathcal{Y}_0$ is finite a.s., then $\lim_{t\to \infty} \widetilde{\mathcal{A}}_1^{t} \bar{\times}_N \mathcal{Y}_{0} = \mathcal{O}$ and the operator $\sum_{k=0}^\infty \widetilde{\mathcal{A}}_1^k L^k$ applied to a sequence $\mathcal{Y}_t$ s.t. $|\mathcal{Y}_{\mathbf{i},t}| < c$ a.s. $\forall \, \mathbf{i}$ converges to the inverse operator $(\mathcal{I} -\widetilde{\mathcal{A}}_1 L)^{-1}$. By the properties of the contracted product we get
\begin{align*}
\mathcal{Y}_t & = \sum_{k=0}^\infty \widetilde{\mathcal{A}}_1^k \bar{\times}_N (L^k \widetilde{\mathcal{A}}_0) + \sum_{k=0}^\infty (\widetilde{\mathcal{A}}_1^k \bar{\times}_N \widetilde{\mathcal{B}}) \bar{\times}_M (L^k \mathcal{X}_t) + \sum_{k=0}^\infty \widetilde{\mathcal{A}}_1^k \bar{\times}_N (L^k \mathcal{E}_t) \\
 & = (\mathcal{I} - \widetilde{\mathcal{A}}_1 L)^{-1} \bar{\times}_N \widetilde{\mathcal{A}}_0 + \sum_{k=0}^\infty \widetilde{\mathcal{A}}_1^k \bar{\times}_N \widetilde{\mathcal{B}} \bar{\times}_M \mathcal{X}_{t-k} + \sum_{k=0}^\infty \widetilde{\mathcal{A}}_1^k \bar{\times}_N \mathcal{E}_{t-k} \, .
\end{align*}
%The stationarity of the process follows from $\sum_{k=0}^\infty \norm{\widetilde{\mathcal{A}}_1}^k < +\infty$, where $\norm{\cdot}$ is the Frobenious norm induced by the inner product.
From the assumption $\mathcal{E}_t \distas{iid} \mathcal{N}_{I_1,\ldots,I_N}(\mathcal{O},\Sigma_1,\ldots,\Sigma_N)$, we know that $\E(\mathcal{Y}_t) = \mathcal{Y}_0$, which is finite. Consider the auto-covariance at lag $h \geq 1$. From Lemma \ref{lemma:contracted_properties}, we have $\E\big( \big( \mathcal{Y}_t -\E(\mathcal{Y}_t) \big) \circ \big( \mathcal{Y}_{t-h} -\E(\mathcal{Y}_{t-h}) \big) \big) = \E\big( \mathcal{Y}_t \circ \mathcal{Y}_{t-h} \big) = \E\big( \mathcal{Y}_t \bar{\times}_1 \mathcal{Y}_{t-h}^T \big)$.
Using the infinite moving average representation for $\mathcal{Y}_t$, we get
\begin{align*}
\hspace*{-30pt}
\E\Big( \mathcal{Y}_t \bar{\times}_1 \mathcal{Y}_{t-h}^T \Big)  
% & = \E\Big( \big( \sum_{k=0}^\infty\mathcal{A}^k \bar{\times}_N \mathcal{E}_{t-k} \big) \bar{\times}_1 \big( \sum_{k=0}^\infty\mathcal{A}^k \bar{\times}_N \mathcal{E}_{t-k-h} \big)^T \Big) \\
 & = \E\Big( \big( \sum_{k=0}^{h-1} \mathcal{A}^k \bar{\times}_N \mathcal{E}_{t-k} + \sum_{k=0}^\infty \mathcal{A}^{k+h} \bar{\times}_N \mathcal{E}_{t-k-h} \big) \bar{\times}_1 \big( \sum_{k=0}^\infty \mathcal{A}^k \bar{\times}_N \mathcal{E}_{t-k-h} \big)^T \Big) \\
% & = \E\Big( \big( \sum_{k=0}^\infty\mathcal{A}^{k+h} \bar{\times}_N \mathcal{E}_{t-k-h} \big) \bar{\times}_1 \big( \sum_{k=0}^\infty \mathcal{A}^k \bar{\times}_N \mathcal{E}_{t-k-h} \big)^T \Big) \\
 & = \E\Big( \big( \sum_{k=0}^\infty\mathcal{A}^{k+h} \bar{\times}_N \mathcal{E}_{t-k-h} \big) \bar{\times}_1 \big( \sum_{k=0}^\infty \mathcal{E}_{t-k-h}^T \bar{\times}_N (\mathcal{A}^T)^k \big) \Big),
\end{align*}
where we used the assumption of independence of $\mathcal{E}_t, \mathcal{E}_{t-h}$, for any $h \geq 0$, and the fact that $(\mathcal{X} \bar{\times}_N \mathcal{Y})^T = (\mathcal{Y}^T \bar{\times}_N \mathcal{X}^T)$. Using $\E(\mathcal{E}_t) = \mathcal{O}$ and linearity of expectation and of the contracted product we get
\begin{align*}
% \E\Big( \big( \sum_{k=0}^\infty\mathcal{A}^{k+h} \bar{\times}_N \mathcal{E}_{t-k-h} \big) \bar{\times}_1 \big( \sum_{k=0}^\infty \mathcal{E}_{t-k-h}^T \bar{\times}_N (\mathcal{A}^T)^k \big) \Big)
\E\Big( \mathcal{Y}_t \bar{\times}_1 \mathcal{Y}_{t-h}^T \Big)
% & = \E\Big( \sum_{k=0}^\infty \big( \mathcal{A}^{k+h} \bar{\times}_N \mathcal{E}_{t-k-h} \big) \bar{\times}_1 \big( \mathcal{E}_{t-k-h}^T \bar{\times}_N (\mathcal{A}^T)^k \big) \Big) \\
 & = \sum_{k=0}^\infty \mathcal{A}^{k+h} \bar{\times}_N \E\Big( \mathcal{E}_{t-k-h} \bar{\times}_1 \mathcal{E}_{t-k-h}^T \Big) \bar{\times}_N (\mathcal{A}^T)^k \\
 & = \sum_{k=0}^\infty \mathcal{A}^{k+h} \bar{\times}_N \boldsymbol{\Sigma} \bar{\times}_N (\mathcal{A}^T)^k 
 = \mathcal{A}^{h} \bar{\times}_N (\mathcal{I}-\mathcal{A} \bar{\times}_N \boldsymbol{\Sigma} \bar{\times}_N \mathcal{A}^T)^{-1},
\end{align*}
where $\E( \mathcal{E}_{t-k-h} \bar{\times}_1 \mathcal{E}_{t-k-h}^T) = \E( \mathcal{E}_{t-k-h} \circ \mathcal{E}_{t-k-h}) = \boldsymbol{\Sigma} = \Sigma_1 \circ \ldots \circ \Sigma_N$.
%Note that $\rho(\mathcal{A}) = \rho(\mathcal{A}^T)$.
From the assumption $\rho(\mathcal{A}) < 1$ it follows that the above series converges to a finite limit, which is independent from $t$, thus proving that the process is weakly stationary.
\end{proof}

\begin{proof}[Proof of Proposition \ref{proposition:ART_stationatity_from_VAR}]
From Theorem 3.2, Corollary 3.3 of \cite{Brazell13Solving_MultilinearSystem}, we know that $\T$ is a group (called tensor group) and that the matricization operator $\operatorname{mat}_{1:N,1:N}$ is an isomorphism between $\T$ and the linear group of square matrices of size $I^* = \prod_{n=1}^N I_n$.
Therefore, there exists a one-to-one relationship between the two eigenvalue problems $\mathcal{A} \bar{\times}_N \mathcal{X} = \lambda \mathcal{X}$ and $A\mathbf{x} = \widetilde{\lambda} \mathbf{x}$, where $A = \operatorname{mat}_{1:N,1:N}(\mathcal{A})$. In particular, $\lambda = \widetilde{\lambda}$ and $\mathbf{x} = \vecc{\mathcal{X}}$.
Consequently, $\rho(A) = \rho(\mathcal{A})$ and the result follows for $p=1$ from the fact that $\rho(A) < 1$ is a sufficient condition for the VAR(1) stationarity Proposition 2.1 of \cite{Lutkepohl05VAR_book}.
Since any VAR($p$) and ART($p$) processes can be rewritten as VAR(1) and ART(1), respectively, on an augmented state space, the result follows for any $p \geq 1$.
\end{proof}

\begin{proof}[Proof of Lemma \ref{lemma:ARTp_ART1}]
Consider a ART($p$) process with $\mathcal{Y}_t \in \R^{I_1\times\ldots\times I_N}$ and $p \geq 1$. We define the $(pI_1 \times I_2 \times \ldots \times I_N)$-dimensional tensors $\underline{\mathcal{Y}}_t$ and $\underline{\mathcal{E}}_t$ as $\underline{\mathcal{Y}}_{(k-1)I_1+1:kI_1,:,\ldots,:,t} = \mathcal{Y}_{t-k}$ and $\underline{\mathcal{E}}_{(k-1)I_1+1:kI_1,:,\ldots,:,t} = \mathcal{E}_{t-k}$, for $k=0,\ldots,p$, respectively.
%\begin{equation*}
%\underline{\mathcal{Y}}_{(k-1)I_1+1:kI_1,:,\ldots,:,t} = \mathcal{Y}_{t-k}, \qquad
%\underline{\mathcal{E}}_{(k-1)I_1+1:kI_1,:,\ldots,:,t} = \mathcal{E}_{t-k}, \quad k=0,\ldots,p
%\end{equation*}
Define the $(pI_1 \times I_2 \times \ldots \times I_N \times pI_1 \times I_2 \times \ldots \times I_N)$-dimensional tensor $\underline{\mathcal{A}}$ as $\underline{\mathcal{A}}_{(1:I_1,:,\ldots,:,(k-1)I_1+1:kI_1,:,\ldots,:} = \mathcal{A}_{k}$, for $k=1,\ldots,p$, $\underline{\mathcal{A}}_{(kI_1+1:(k+1)I_1,:,\ldots,:,(k-1)I_1+1:kI_1,:,\ldots,:} = \mathcal{I}$, for $k=1,\ldots,p-1$
%\begin{align*}
%\underline{\mathcal{A}}_{(1:I_1,:,\ldots,:,(k-1)I_1+1:kI_1,:,\ldots,:} & = \mathcal{A}_{k}, \quad k=1,\ldots,p \\
%\underline{\mathcal{A}}_{(kI_1+1:(k+1)I_1,:,\ldots,:,(k-1)I_1+1:kI_1,:,\ldots,:} & = \mathcal{I}, \quad k=1,\ldots,p-1
%\end{align*}
and $0$ elsewhere.
Using this notation, we can rewrite the $(I_1 \times I_2 \times \ldots \times I_N)$-dimensional ART($p$) process $\mathcal{Y}_t = \sum_{k=1}^p \mathcal{A}_{k} \bar{\times}_N \mathcal{Y}_{t-j} + \mathcal{E}_t$ as the $(pI_1 \times I_2 \times \ldots \times I_N)$-dimensional ART(1) process $\underline{\mathcal{Y}}_t = \underline{\mathcal{A}} \bar{\times}_N \underline{\mathcal{Y}}_{t-1} + \underline{\mathcal{E}}_t$.
\end{proof}

\section{Computational Details} \label{sec:apdx_computational_tensor}
This appendix shows the derivation of the results. See the supplement for details.

\subsection{Full conditional distribution of $\phi_r$}
Define $C_r = \sum_{j=1}^J \boldsymbol{\beta}_j^{(r)'} W_{j,r}^{-1} \boldsymbol{\beta}_j^{(r)}$ and note that, since $\sum_{r=1}^R \phi_r =1$, it holds $\sum_{r=1}^R b_{\tau} \tau \phi_r = b_{\tau} \tau$. The posterior full conditional distribution of $\boldsymbol{\phi}$, integrating out $\tau$, is
\begin{align*}
 & p(\boldsymbol{\phi}|\mathcal{B},\mathbf{W}) \propto \pi(\boldsymbol{\phi}) \int_0^{+\infty} p(\mathcal{B}|\mathbf{W},\boldsymbol{\phi},\tau) \pi(\tau) \mathrm{d}\tau \\
 & \; \propto \prod_{r=1}^R \phi_r^{\alpha-1} \int_0^{+\infty} \Big( \prod_{r=1}^R \prod_{j=1}^J (\tau \phi_r)^{-I_j/2} \exp\Big( -\frac{1}{2\tau \phi_r} \boldsymbol{\beta}_j^{(r)'} W_{j,r}^{-1} \boldsymbol{\beta}_j^{(r)} \Big) \Big) \tau^{a_\tau-1} e^{-b_\tau \tau} \mathrm{d}\tau \\
% & \propto \prod_{r=1}^R \phi_r^{\alpha-1-\frac{I_0}{2}} \int_0^{+\infty} \tau^{a_\tau -1 -\frac{Rd_0}{2}} \exp\Big( -b_\tau \tau -\sum_{r=1}^R \frac{C_r}{2\tau \phi_r} \Big) \mathrm{d}\tau \\
 & \; \propto \int_0^{+\infty} \Big( \prod_{r=1}^R \phi_r^{\alpha-\frac{I_0}{2}-1} \Big) \tau^{\big( \alpha R -\frac{RI_0}{2} \big) -1} \exp\Big( -\sum_{r=1}^R \Big( \frac{C_r}{2\tau \phi_r} +b_\tau \tau \phi_r \Big) \Big) \mathrm{d}\tau
\end{align*}
where the integrand is the kernel of the GiG for $\psi_r=\tau \phi_r$ in eq. \eqref{eq:posterior_psi}. Then, by renormalizing, $\phi_r = \psi_r / \sum_{l=1}^R \psi_l$.

\subsection{Full conditional distribution of $\tau$}
The posterior full conditional distribution of $\tau$ is
\begin{align*}
p(\tau|\mathcal{B},\mathbf{W},\boldsymbol{\phi}) & \propto \tau^{a_\tau -1} e^{-b_\tau \tau} \Big( \prod_{r=1}^R (\tau \phi_r)^{-\frac{I_0}{2}} \exp\Big( -\dfrac{1}{2\tau \phi_r} \sum_{j=1}^4 \boldsymbol{\beta}_j^{(r)'} (W_{j,r})^{-1} \boldsymbol{\beta}_j^{(r)} \Big) \Big) \\
 & \propto \tau^{a_\tau -\frac{R I_0}{2} -1} \exp\Big( -b_\tau \tau - \tau^{-1} \sum_{r=1}^R \frac{C_r}{\phi_r} \Big),
\end{align*}
which is the kernel of the GiG in eq. \eqref{eq:posterior_tau}.

\subsection{Full conditional distribution of $\lambda_{j,r}$}
The full conditional distribution of $\lambda_{j,r}$, integrating out $W_{j,r}$, is
\begin{align*}
p(\lambda_{j,r}|\boldsymbol{\beta}_j^{(r)},\phi_r,\tau) & \propto \lambda_{j,r}^{a_\lambda -1} e^{-b_\lambda \lambda_{j,r}} \prod_{p=1}^{I_j} \frac{\lambda_{j,r}}{2\sqrt{\tau \phi_r}} \exp\Big( -\frac{\abs{\beta_{j,p}^{(r)}}}{ (\lambda_{j,r}/\sqrt{\tau \phi_r})^{-1} } \Big) \\
 & \propto \lambda_{j,r}^{(a_\lambda+I_j)-1} \exp\Big( -\Big( b_\lambda +\frac{\norm{\boldsymbol{\beta}_j^{(r)}}_1}{\sqrt{\tau \phi_r}} \Big) \lambda_{j,r} \Big),
\end{align*}
which is the kernel of the Gamma in eq. \eqref{eq:posterior_w}.

\subsection{Full conditional distribution of $w_{j,r,p}$}
The posterior full conditional distribution of $w_{j,r,p}$ is
\begin{align*}
p(w_{j,r,p}|\boldsymbol{\beta}_j^{(r)}, \lambda_{j,r},\phi_r,\tau) & \propto w_{j,r,p}^{-\frac{1}{2}} \exp\Big( -\frac{\beta_{j,p}^{(r)^2} w_{j,r,p}^{-1}}{2\tau \phi_r} \Big) \exp\Big( -\frac{\lambda_{j,r}^2 w_{j,r,p}}{2} \Big) \\
 & \propto w_{j,r,p}^{-\frac{1}{2}} \exp\Big( -\frac{\lambda_{j,r}^2}{2}w_{j,r,p} -\frac{\beta_{j,p}^{(r)^2}}{2\tau \phi_r} w_{j,r,p}^{-1} \Big),
\end{align*}
which is the kernel of the GiG in eq. \eqref{eq:posterior_w}.

\subsection{Full conditional distributions of PARAFAC marginals}
Consider the model in eq. \eqref{eq:model_final}, it holds
\begin{align*}
\vecc{\mathcal{Y}_t} & = \vecc{ \mathcal{B}_{-r} \times_4 \, \mathbf{x}_t } + \vecc{ \mathcal{B}_r \times_4 \, \mathbf{x}_t } + \vecc{\mathcal{E}_t},
\end{align*}
with $\vecc{\mathcal{B}_r \times_4 \, \mathbf{x}_t} = \vecc{ \boldsymbol{\beta}_1^{(r)} \circ \boldsymbol{\beta}_2^{(r)} \circ \boldsymbol{\beta}_3^{(r)} } \cdot \mathbf{x}_{t}' \boldsymbol{\beta}_4^{(r)}$.
From Lemma \ref{lemma:outer_product_vec}, we have
\begin{align}
\label{eq:apdx_beta4}
\vecc{ \boldsymbol{\beta}_1^{(r)} \circ \boldsymbol{\beta}_2^{(r)} \circ \boldsymbol{\beta}_3^{(r)} } \cdot \mathbf{x}_{t}' \boldsymbol{\beta}_4^{(r)} & = \vecc{ \boldsymbol{\beta}_1^{(r)} \circ \boldsymbol{\beta}_2^{(r)} \circ \boldsymbol{\beta}_3^{(r)} } \cdot \mathbf{x}_{t}' \boldsymbol{\beta}_4^{(r)} = \mathbf{b}_4 \boldsymbol{\beta}_4^{(r)} \\
\label{eq:apdx_beta1}
 & = \langle \boldsymbol{\beta}_4^{(r)}, \mathbf{x}_t \rangle \big( \boldsymbol{\beta}_3^{(r)} \otimes \boldsymbol{\beta}_2^{(r)} \otimes \mathbf{I}_I \big) \boldsymbol{\beta}_1^{(r)} = \mathbf{b}_1 \boldsymbol{\beta}_1^{(r)} \\
\label{eq:apdx_beta2}
 & = \langle \boldsymbol{\beta}_4^{(r)}, \mathbf{x}_t \rangle \big( \boldsymbol{\beta}_3^{(r)} \otimes \mathbf{I}_J \otimes \boldsymbol{\beta}_1^{(r)} \big) \boldsymbol{\beta}_2^{(r)} = \mathbf{b}_2 \boldsymbol{\beta}_2^{(r)} \\
\label{eq:apdx_beta3}
 & = \langle \boldsymbol{\beta}_4^{(r)}, \mathbf{x}_t \rangle \big( \mathbf{I}_K \otimes \boldsymbol{\beta}_2^{(r)} \otimes \boldsymbol{\beta}_1^{(r)} \big) \boldsymbol{\beta}_3^{(r)} = \mathbf{b}_3 \boldsymbol{\beta}_3^{(r)}.
\end{align}

Define with $\mathbf{y}_t = \vecc{\mathcal{Y}_t}$ and $\boldsymbol{\Sigma}^{-1} = \Sigma_3^{-1} \otimes \Sigma_2^{-1} \otimes \Sigma_1^{-1}$, we obtain
\begin{align}
\notag
 & L(\mathbf{Y} |\boldsymbol{\theta}) \propto \exp\Big( -\frac{1}{2} \sum_{t=1}^T \vecc{\tilde{\mathcal{E}}_t}' \big(\Sigma_3^{-1} \otimes \Sigma_2^{-1} \otimes \Sigma_1^{-1} \big) \vecc{\tilde{\mathcal{E}}_t} \Big) \\ \notag
 & \propto \exp\Bigg( -\frac{1}{2} \sum_{t=1}^T -2\big( \mathbf{y}_t' -\vecc{\mathcal{B}_{-r}\times_4 \mathbf{x}_t}' \big) \boldsymbol{\Sigma}^{-1}\vecc{\boldsymbol{\beta}_1^{(r)} \circ \boldsymbol{\beta}_2^{(r)} \circ \boldsymbol{\beta}_3^{(r)}} \langle \boldsymbol{\beta}_4^{(r)}, \mathbf{x}_t \rangle \\
 & \quad +\vecc{\boldsymbol{\beta}_1^{(r)} \circ \boldsymbol{\beta}_2^{(r)} \circ \boldsymbol{\beta}_3^{(r)}}' \langle \boldsymbol{\beta}_4^{(r)}, \mathbf{x}_t \rangle \boldsymbol{\Sigma}^{-1} \vecc{\boldsymbol{\beta}_1^{(r)} \circ \boldsymbol{\beta}_2^{(r)} \circ \boldsymbol{\beta}_3^{(r)}} \langle \boldsymbol{\beta}_4^{(r)}, \mathbf{x}_t \rangle \Bigg).
\label{eq:apdx_like_all}
\end{align}
Consider the case $j=1$. By exploiting eq. \eqref{eq:apdx_beta1} we get
\begin{align}
\notag
 & L(\mathbf{Y}|\boldsymbol{\theta}) \propto \exp\Bigg( -\frac{1}{2} \sum_{t=1}^T \boldsymbol{\beta}_1^{(r)'} \langle \boldsymbol{\beta}_4^{(r)}, \mathbf{x}_t \rangle^2 \big( \boldsymbol{\beta}_3^{(r)} \otimes \boldsymbol{\beta}_2^{(r)} \otimes \mathbf{I}_{I_1} \big)' \boldsymbol{\Sigma}^{-1} \big( \boldsymbol{\beta}_3^{(r)} \otimes \boldsymbol{\beta}_2^{(r)} \otimes \mathbf{I}_{I_1} \big)  \\ \notag
 & \quad \cdot \boldsymbol{\beta}_1^{(r)} -2\big( \mathbf{y}_t' -\vecc{\mathcal{B}_{-r}\times_4 \mathbf{x}_t}' \big) \boldsymbol{\Sigma}^{-1} \langle \boldsymbol{\beta}_4^{(r)}, \mathbf{x}_t \rangle \big( \boldsymbol{\beta}_3^{(r)} \otimes \boldsymbol{\beta}_2^{(r)} \otimes \mathbf{I}_{I_1} \big) \boldsymbol{\beta}_1^{(r)} \Bigg) \\
 & = \exp\Big( -\frac{1}{2} \boldsymbol{\beta}_1^{(r)'} \mathbf{S}_1^L \boldsymbol{\beta}_1^{(r)} -2\mathbf{m}_1^L \boldsymbol{\beta}_1^{(r)} \Big).
\label{eq:apdx_like_beta1}
\end{align}
Consider the case $j=2$. From eq. \eqref{eq:apdx_beta2} we get
\begin{align}
\notag
 & L(\mathbf{Y}|\boldsymbol{\theta}) \propto \exp\Bigg( -\frac{1}{2} \sum_{t=1}^T \boldsymbol{\beta}_2^{(r)'} \langle \boldsymbol{\beta}_4^{(r)}, \mathbf{x}_t \rangle^2 \big( \boldsymbol{\beta}_3^{(r)} \otimes \mathbf{I}_{I_2} \otimes \boldsymbol{\beta}_1^{(r)} \big) \boldsymbol{\Sigma}^{-1} \big( \boldsymbol{\beta}_3^{(r)} \otimes \mathbf{I}_{I_2} \otimes \boldsymbol{\beta}_1^{(r)} \big) \\ \notag
 & \quad \cdot \boldsymbol{\beta}_2^{(r)} -2\big( \mathbf{y}_t' -\vecc{\mathcal{B}_{-r}\times_4 \mathbf{x}_t}' \big) \boldsymbol{\Sigma}^{-1} \langle \boldsymbol{\beta}_4^{(r)}, \mathbf{x}_t \rangle \big( \boldsymbol{\beta}_3^{(r)} \otimes \mathbf{I}_{I_2} \otimes \boldsymbol{\beta}_1^{(r)} \big) \boldsymbol{\beta}_2^{(r)} \Bigg) \\
 & = \exp\Big( -\frac{1}{2} \boldsymbol{\beta}_2^{(r)'} \mathbf{S}_2^L \boldsymbol{\beta}_2^{(r)} -2\mathbf{m}_2^L \boldsymbol{\beta}_2^{(r)} \Big).
\label{eq:apdx_like_beta2}
\end{align}
Consider the case $j=3$, by exploiting eq. \eqref{eq:apdx_beta3} we get
\begin{align}
\notag
 & L(\mathbf{Y}|\boldsymbol{\theta}) \propto \exp\Bigg( -\frac{1}{2} \sum_{t=1}^T \boldsymbol{\beta}_3^{(r)'} \langle \boldsymbol{\beta}_4^{(r)}, \mathbf{x}_t \rangle^2 \big( \mathbf{I}_{I_3} \otimes \boldsymbol{\beta}_2^{(r)} \otimes \boldsymbol{\beta}_1^{(r)} \big) \boldsymbol{\Sigma}^{-1} \big( \mathbf{I}_{I_3} \otimes \boldsymbol{\beta}_2^{(r)} \otimes \boldsymbol{\beta}_1^{(r)} \big) \\ \notag
 & \quad \cdot \boldsymbol{\beta}_3^{(r)} -2\big( \mathbf{y}_t' -\vecc{\mathcal{B}_{-r}\times_4 \mathbf{x}_t}' \big) \boldsymbol{\Sigma}^{-1} \langle \boldsymbol{\beta}_4^{(r)}, \mathbf{x}_t \rangle \big( \mathbf{I}_{I_3} \otimes \boldsymbol{\beta}_2^{(r)} \otimes \boldsymbol{\beta}_1^{(r)} \big) \boldsymbol{\beta}_3^{(r)} \Bigg) \\
 & = \exp\Big( -\frac{1}{2} \boldsymbol{\beta}_3^{(r)'} \mathbf{S}_3^L \boldsymbol{\beta}_3^{(r)} -2\mathbf{m}_3^L \boldsymbol{\beta}_3^{(r)} \Big).
\label{eq:apdx_like_beta3}
\end{align}
Finally, in the case $j=4$. From eq. \eqref{eq:apdx_like_all} we get
\begin{align}
\notag
 & L(\mathbf{Y}|\boldsymbol{\theta}) \propto \exp\Bigg( -\frac{1}{2} \sum_{t=1}^T -2\big( \mathbf{y}_t' -\vecc{\mathcal{B}_{-r}\times_4 \mathbf{x}_t}' \big) \boldsymbol{\Sigma}^{-1}\vecc{\boldsymbol{\beta}_1^{(r)} \circ \boldsymbol{\beta}_2^{(r)} \circ \boldsymbol{\beta}_3^{(r)}} \\ \notag
 & \quad \cdot \mathbf{x}_t' \boldsymbol{\beta}_4^{(r)} +\boldsymbol{\beta}_4^{(r)'} \mathbf{x}_t \vecc{\boldsymbol{\beta}_1^{(r)} \circ \boldsymbol{\beta}_2^{(r)} \circ \boldsymbol{\beta}_3^{(r)}}' \boldsymbol{\Sigma}^{-1} \vecc{\boldsymbol{\beta}_1^{(r)} \circ \boldsymbol{\beta}_2^{(r)} \circ \boldsymbol{\beta}_3^{(r)}} \mathbf{x}_t' \boldsymbol{\beta}_4^{(r)} \Bigg) \\
 & = \exp\Big( -\frac{1}{2} \boldsymbol{\beta}_4^{(r)'} \mathbf{S}_4^L \boldsymbol{\beta}_4^{(r)} -2\mathbf{m}_4^L \boldsymbol{\beta}_4^{(r)} \Big).
\label{eq:apdx_like_beta4}
\end{align}

\subsubsection{Full conditional distribution of $\beta_1^{(r)}$}
From eq. \eqref{eq:prior_beta}-\eqref{eq:apdx_like_beta1}, the posterior full conditional distribution of $\boldsymbol{\beta}_1^{(r)}$ is
\begin{align*}
% & p(\boldsymbol{\beta}_1^{(r)} | \boldsymbol{\beta}_{-1}^{(r)}, \mathcal{B}_{-r}, W_{1,r}, \phi_r, \tau, \Sigma_1,\Sigma_2,\Sigma_3,\mathbf{Y}) \propto \\
% & \propto \exp\Big( -\frac{1}{2} \boldsymbol{\beta}_1^{(r)'} \mathbf{S}_1^L \boldsymbol{\beta}_1^{(r)} -2\mathbf{m}_1^L \boldsymbol{\beta}_1^{(r)} \Big) \cdot \exp\Big( -\frac{1}{2} \boldsymbol{\beta}_1^{(r)'} (W_{1,r} \phi_r \tau)^{-1} \boldsymbol{\beta}_1^{(r)} \Big) \\
p(\boldsymbol{\beta}_1^{(r)} | -) & \propto \exp\Big( -\frac{1}{2} \boldsymbol{\beta}_1^{(r)'} \mathbf{S}_1^L \boldsymbol{\beta}_1^{(r)} -2\mathbf{m}_1^L \boldsymbol{\beta}_1^{(r)} \Big) \cdot \exp\Big( -\frac{1}{2} \boldsymbol{\beta}_1^{(r)'} (W_{1,r} \phi_r \tau)^{-1} \boldsymbol{\beta}_1^{(r)} \Big) \\
 & = \exp\Big( -\frac{1}{2} \big( \boldsymbol{\beta}_1^{(r)'} \big( \mathbf{S}_1^L + (W_{1,r} \phi_r \tau)^{-1} \big) \boldsymbol{\beta}_1^{(r)} -2\mathbf{m}_1^L \boldsymbol{\beta}_1^{(r)} \big) \Big),
\end{align*}
which is the kernel of the Normal in eq. \eqref{eq:posterior_betas}.

\subsubsection{Full conditional distribution of $\beta_2^{(r)}$}
From eq. \eqref{eq:prior_beta}-\eqref{eq:apdx_like_beta2}, the posterior full conditional distribution of $\boldsymbol{\beta}_2^{(r)}$ is
\begin{align*}
% & p(\boldsymbol{\beta}_2^{(r)} | \boldsymbol{\beta}_{-2}^{(r)}, \mathcal{B}_{-r}, W_{2,r}, \phi_r, \tau, \Sigma_1,\Sigma_2,\Sigma_3,\mathbf{Y}) \propto \\
% & \propto \exp\Big( -\frac{1}{2} \boldsymbol{\beta}_2^{(r)'} \mathbf{S}_2^L \boldsymbol{\beta}_2^{(r)} -2\mathbf{m}_2^L \boldsymbol{\beta}_2^{(r)} \Big) \cdot \exp\Big( -\frac{1}{2} \boldsymbol{\beta}_2^{(r)'} (W_{2,r} \phi_r \tau)^{-1} \boldsymbol{\beta}_2^{(r)} \Big) \\
p(\boldsymbol{\beta}_2^{(r)} | -) & \propto \exp\Big( -\frac{1}{2} \boldsymbol{\beta}_2^{(r)'} \mathbf{S}_2^L \boldsymbol{\beta}_2^{(r)} -2\mathbf{m}_2^L \boldsymbol{\beta}_2^{(r)} \Big) \cdot \exp\Big( -\frac{1}{2} \boldsymbol{\beta}_2^{(r)'} (W_{2,r} \phi_r \tau)^{-1} \boldsymbol{\beta}_2^{(r)} \Big) \\
 & = \exp\Big( -\frac{1}{2} \big( \boldsymbol{\beta}_2^{(r)'} \big( \mathbf{S}_2^L + (W_{2,r} \phi_r \tau)^{-1} \big) \boldsymbol{\beta}_2^{(r)} -2\mathbf{m}_2^L \boldsymbol{\beta}_2^{(r)} \big) \Big),
\end{align*}
which is the kernel of the Normal in eq. \eqref{eq:posterior_betas}.

\subsubsection{Full conditional distribution of $\beta_3^{(r)}$}
From eq. \eqref{eq:prior_beta}-\eqref{eq:apdx_like_beta3}, the posterior full conditional distribution of $\boldsymbol{\beta}_3^{(r)}$ is
\begin{align*}
% & p(\boldsymbol{\beta}_3^{(r)} | \boldsymbol{\beta}_{-3}^{(r)}, \mathcal{B}_{-r}, W_{3,r}, \phi_r, \tau, \Sigma_1,\Sigma_2,\Sigma_3,\mathbf{Y}) \propto \\
% & \propto \exp\Big( -\frac{1}{2} \boldsymbol{\beta}_3^{(r)'} \mathbf{S}_3^L \boldsymbol{\beta}_3^{(r)} -2\mathbf{m}_3^L \boldsymbol{\beta}_3^{(r)} \Big) \cdot \exp\Big( -\frac{1}{2} \boldsymbol{\beta}_3^{(r)'} (W_{3,r} \phi_r \tau)^{-1} \boldsymbol{\beta}_3^{(r)} \Big) \\
p(\boldsymbol{\beta}_3^{(r)} | -) & \propto \exp\Big( -\frac{1}{2} \boldsymbol{\beta}_3^{(r)'} \mathbf{S}_3^L \boldsymbol{\beta}_3^{(r)} -2\mathbf{m}_3^L \boldsymbol{\beta}_3^{(r)} \Big) \cdot \exp\Big( -\frac{1}{2} \boldsymbol{\beta}_3^{(r)'} (W_{3,r} \phi_r \tau)^{-1} \boldsymbol{\beta}_3^{(r)} \Big) \\
 & = \exp\Big( -\frac{1}{2} \big( \boldsymbol{\beta}_3^{(r)'} \big( \mathbf{S}_3^L + (W_{3,r} \phi_r \tau)^{-1} \big) \boldsymbol{\beta}_3^{(r)} -2\mathbf{m}_3^L \boldsymbol{\beta}_3^{(r)} \big) \Big),
\end{align*}
which is the kernel of the Normal in eq. \eqref{eq:posterior_betas}.

\subsubsection{Full conditional distribution of $\beta_4^{(r)}$}
From eq. \eqref{eq:prior_beta}-\eqref{eq:apdx_like_beta4}, the posterior full conditional distribution of $\boldsymbol{\beta}_4^{(r)}$ is
\begin{align*}
% & p(\boldsymbol{\beta}_4^{(r)} | \boldsymbol{\beta}_{-4}^{(r)}, \mathcal{B}_{-r}, W_{4,r}, \phi_r, \tau, \Sigma_1,\Sigma_2,\Sigma_3,\mathbf{Y}) \propto \\
% & \propto \exp\Big( -\frac{1}{2} \boldsymbol{\beta}_4^{(r)'} \mathbf{S}_4^L \boldsymbol{\beta}_4^{(r)} -2\mathbf{m}_4^L \boldsymbol{\beta}_4^{(r)} \Big) \cdot \exp\Big( -\frac{1}{2} \boldsymbol{\beta}_4^{(r)'} (W_{4,r} \phi_r \tau)^{-1} \boldsymbol{\beta}_4^{(r)} \Big) \\
p(\boldsymbol{\beta}_4^{(r)} | -) & \propto \exp\Big( -\frac{1}{2} \boldsymbol{\beta}_4^{(r)'} \mathbf{S}_4^L \boldsymbol{\beta}_4^{(r)} -2\mathbf{m}_4^L \boldsymbol{\beta}_4^{(r)} \Big) \cdot \exp\Big( -\frac{1}{2} \boldsymbol{\beta}_4^{(r)'} (W_{4,r} \phi_r \tau)^{-1} \boldsymbol{\beta}_4^{(r)} \Big) \\
 & = \exp\Big( -\frac{1}{2} \big( \boldsymbol{\beta}_4^{(r)'} \big( \mathbf{S}_4^L + (W_{4,r} \phi_r \tau)^{-1} \big) \boldsymbol{\beta}_4^{(r)} -2\mathbf{m}_4^L \boldsymbol{\beta}_4^{(r)} \big) \Big),
\end{align*}
which is the kernel of the Normal in eq. \eqref{eq:posterior_betas}.

\subsection{Full conditional distribution of $\Sigma_1$}
Define $\tilde{\mathcal{E}}_t = \mathcal{Y}_t -\mathcal{B}\times_4 \mathbf{x}_t$, $\tilde{\mathbf{E}}_{(1),t} = \textnormal{mat}_{(3)}(\tilde{\mathcal{E}}_t)$, $\mathbf{Z}_1 = \Sigma_3^{-1} \otimes \Sigma_2^{-1}$ and $S_1 =\sum_{t=1}^T \tilde{\mathbf{E}}_{(1),t} \mathbf{Z}_1 \tilde{\mathbf{E}}_{(1),t}'$. The posterior full conditional distribution of $\Sigma_1$ is
\begin{align*}
% & p(\Sigma_1 | \mathcal{B}, \mathbf{Y}, \Sigma_2, \Sigma_3, \gamma) \propto \\
% & \propto \abs{\Sigma_1}^{-\frac{\nu_1+I_1+T I_2 I_3+1}{2}} \exp\Big( -\frac{1}{2} \big( \tr{\gamma\Psi_1 \Sigma_1^{-1}} + \sum_{t=1}^T \tr{\tilde{\mathbf{E}}_{(1),t} \mathbf{Z}_1 \tilde{\mathbf{E}}_{(1),t}' \Sigma_1^{-1}} \big) \Big) \\
%p(\Sigma_1 | -) & \propto \abs{\Sigma_1}^{-\frac{\nu_1+I_1+T I_2 I_3+1}{2}} \exp\Big( -\frac{1}{2} \big( \tr{\gamma\Psi_1 \Sigma_1^{-1}} + \sum_{t=1}^T \tr{\tilde{\mathbf{E}}_{(1),t} \mathbf{Z}_1 \tilde{\mathbf{E}}_{(1),t}' \Sigma_1^{-1}} \big) \Big) \\
p(\Sigma_1 | -) & \propto \frac{\exp\Big( -\frac{1}{2} \big( \tr{\gamma\Psi_1 \Sigma_1^{-1}} + \sum_{t=1}^T \tr{\tilde{\mathbf{E}}_{(1),t} \mathbf{Z}_1 \tilde{\mathbf{E}}_{(1),t}' \Sigma_1^{-1}} \big) \Big)}{\abs{\Sigma_1}^{\frac{\nu_1+I_1+T I_2 I_3+1}{2}}} \\
 & \propto \abs{\Sigma_1}^{-\frac{(\nu_1+T I_2 I_3)+I_1+1}{2}} \exp\Big( -\frac{1}{2} \tr{(\gamma\Psi_1 +S_1) \Sigma_1^{-1}} \Big),
\end{align*}
which is the kernel of the Inverse Wishart in eq. \eqref{eq:posterior_Sigmar}.

\subsection{Full conditional distribution of $\Sigma_2$}
Define $\tilde{\mathcal{E}}_t = \mathcal{Y}_t -\mathcal{B}\times_4 \mathbf{x}_t$, $\tilde{\mathbf{E}}_{(2),t} = \textnormal{mat}_{(2)}(\tilde{\mathcal{E}}_t)$ and $S_2 = \sum_{t=1}^T \tilde{\mathbf{E}}_{(2),t} (\Sigma_3^{-1} \otimes \Sigma_1^{-1}) \tilde{\mathbf{E}}_{(2),t}'$. The posterior full conditional distribution of $\Sigma_2$ is 
\begin{align*}
% & p(\Sigma_2 | \mathcal{B}, \mathbf{Y}, \Sigma_1, \Sigma_3, \gamma) \propto \\
% & \propto \abs{\Sigma_2}^{-\frac{\nu_2+I_2+T I_1 I_3+1}{2}} \exp\Big( -\frac{1}{2} \big( \tr{\gamma\Psi_2 \Sigma_2^{-1}} + \tr{\sum_{t=1}^T \tilde{\mathbf{E}}_{(2),t} (\Sigma_3^{-1} \otimes \Sigma_1^{-1}) \tilde{\mathbf{E}}_{(2),t}' \Sigma_2^{-1}} \big) \Big) \\
%p(\Sigma_2 | -) & \propto \abs{\Sigma_2}^{-\frac{\nu_2+I_2+T I_1 I_3+1}{2}} \exp\Big( -\frac{1}{2} \big( \tr{\gamma\Psi_2 \Sigma_2^{-1}} + \tr{\sum_{t=1}^T \tilde{\mathbf{E}}_{(2),t} (\Sigma_3^{-1} \otimes \Sigma_1^{-1}) \tilde{\mathbf{E}}_{(2),t}' \Sigma_2^{-1}} \big) \Big) \\
p(\Sigma_2 | -) & \propto \frac{\exp\Big( -\frac{1}{2} \big( \tr{\gamma\Psi_2 \Sigma_2^{-1}} + \tr{\sum_{t=1}^T \tilde{\mathbf{E}}_{(2),t} (\Sigma_3^{-1} \otimes \Sigma_1^{-1}) \tilde{\mathbf{E}}_{(2),t}' \Sigma_2^{-1}} \big) \Big)}{\abs{\Sigma_2}^{\frac{\nu_2+I_2+T I_1 I_3+1}{2}}} \\
 & \propto \abs{\Sigma_2}^{-\frac{\nu_2+I_2+T I_1 I_3+1}{2}} \exp\Big( -\frac{1}{2} \tr{\gamma\Psi_2 \Sigma_2^{-1} + S_2 \Sigma_2^{-1}} \Big),
\end{align*}
which is the kernel of the Inverse Wishart in eq. \eqref{eq:posterior_Sigmar}.

\subsection{Full conditional distribution of $\Sigma_3$}
Define $\tilde{\mathcal{E}}_t = \mathcal{Y}_t -\mathcal{B}\times_4 \mathbf{x}_t$, $\tilde{\mathbf{E}}_{(1),t} = \textnormal{mat}_{(1)}(\tilde{\mathcal{E}}_t)$, $\mathbf{Z}_3 = \Sigma_2^{-1} \otimes \Sigma_1^{-1}$ and $S_3 =\sum_{t=1}^T \tilde{\mathbf{E}}_{(1),t} \mathbf{Z}_3 \tilde{\mathbf{E}}_{(1),t}'$. The posterior full conditional distribution of $\Sigma_3$ is
\begin{align*}
% & p(\Sigma_3 | \mathcal{B}, \mathbf{Y}, \Sigma_1, \Sigma_2, \gamma) \propto \\
% & \propto \abs{\Sigma_3}^{-\frac{\nu_3+I_3+T I_1 I_2+1}{2}} \exp\Big( -\frac{1}{2} \big( \tr{\gamma\Psi_3 \Sigma_3^{-1}} + \sum_{t=1}^T \vecc{\tilde{\mathcal{E}}_t}' (\Sigma_3^{-1} \otimes \mathbf{Z}_3) \vecc{\tilde{\mathcal{E}}_t} \big) \Big) \\
%p(\Sigma_3 | -) & \propto \abs{\Sigma_3}^{-\frac{\nu_3+I_3+T I_1 I_2+1}{2}} \exp\Big( -\frac{1}{2} \big( \tr{\gamma\Psi_3 \Sigma_3^{-1}} + \sum_{t=1}^T \vecc{\tilde{\mathcal{E}}_t}' (\Sigma_3^{-1} \otimes \mathbf{Z}_3) \vecc{\tilde{\mathcal{E}}_t} \big) \Big) \\
p(\Sigma_3 | -) & \propto \frac{\exp\Big( -\frac{1}{2} \big( \tr{\gamma\Psi_3 \Sigma_3^{-1}} + \sum_{t=1}^T \vecc{\tilde{\mathcal{E}}_t}' (\Sigma_3^{-1} \otimes \mathbf{Z}_3) \vecc{\tilde{\mathcal{E}}_t} \big) \Big)}{\abs{\Sigma_3}^{\frac{\nu_3+I_3+T I_1 I_2+1}{2}}} \\
% & \propto \abs{\Sigma_3}^{-\frac{\nu_3+I_3+T I_1 I_2+1}{2}} \exp\Big( -\frac{1}{2} \big( \tr{\gamma\Psi_3 \Sigma_3^{-1}} + \sum_{t=1}^T \tr{\tilde{\mathbf{E}}_{(1),t}' \mathbf{Z}_3 \tilde{\mathbf{E}}_{(1),t} \Sigma_3^{-1}} \big) \Big) \\
 & \propto \abs{\Sigma_3}^{-\frac{(\nu_3+T I_1 I_2)+I_3+1}{2}} \exp\Big( -\frac{1}{2} \tr{(\gamma\Psi_3 +S_3) \Sigma_3^{-1}} \Big),
\end{align*}
which is the kernel of the Inverse Wishart in eq. \eqref{eq:posterior_Sigmar}.

\subsection{Full conditional distribution of $\gamma$}
The posterior full conditional distribution is
\begin{align*}
p(\gamma | \Sigma_1, \Sigma_2, \Sigma_3) & \propto \prod_{i=1}^3 \abs{\gamma\Psi_i}^{-\frac{\nu_i}{2}} \exp\Big( -\frac{1}{2} \tr{\gamma \Psi_i \Sigma_i^{-1}} \Big) \gamma^{a_\gamma-1} e^{-b_\gamma \gamma} \\
 & \propto \gamma^{a_\gamma-\frac{\sum_{i=1}^3 \nu_i I_i}{2}-1} \exp\Big( -\frac{1}{2} \tr{\sum_{i=1}^3 \Psi_i \Sigma_i^{-1}} -b_\gamma \gamma \Big)
\end{align*}
which is the kernel of the Gamma in eq. \eqref{eq:posterior_gamma}.

\end{document}